\documentclass[a4paper,11pt]{article}

\usepackage[utf8]{inputenc}
\usepackage[british]{babel}
\usepackage[mono=false]{libertine}
\usepackage[T1]{fontenc}
\usepackage{amsthm}
\usepackage[top=2.5cm, bottom=2.5cm, left=2.54cm, right=2.54cm]{geometry}
\usepackage{xcolor}
\usepackage[font=small]{caption}
\usepackage{enumitem}
\usepackage{bm}

\newtheorem{thm}{Theorem}
\newtheorem{lemma}{Lemma} 
\newtheorem{proposition}{Proposition}

\newenvironment{talign*}
 {\let\displaystyle\textstyle\csname align*\endcsname}
 {\endalign}

\usepackage{titlesec}
\makeatletter
\newcommand{\setappendix}{Appendix~\thesection:~~}
\newcommand{\setsection}{\thesection~~}
\titleformat{\section}{\bfseries\LARGE}{%
	\ifnum\pdfstrcmp{\@currenvir}{appendices}=0
	\setappendix
	\else
	\setsection
\fi}{0em}{}
\makeatother
\usepackage[titletoc]{appendix}

\usepackage{cite}
\usepackage[pdftex]{graphicx}
\graphicspath{./figures}
\usepackage{array}
\usepackage{url}
\usepackage{mathtools}
\usepackage{amssymb,amsfonts,amsmath}
\usepackage{dsfont}
\usepackage{stmaryrd}
\usepackage{bm}
\usepackage{notations}
\usepackage[pdfpagemode=UseNone,bookmarksopen=false,colorlinks=true,urlcolor=blue,citecolor=magenta,citebordercolor=blue,linkcolor=blue]{hyperref}
\usepackage{footmisc}
\DefineFNsymbols{mySymbols}{{\ensuremath\dagger}{\ensuremath\Diamond}{\ensuremath{\ast}}{\ensuremath{\star}}}
\setfnsymbol{mySymbols}
\usepackage{notations}
\usepackage{setspace}                    
\DeclareUnicodeCharacter{00A0}{~}

\def \({\left(}
\def \){\right)}
\def \[{\left[}
\def \]{\right]}
\def \nn{\nonumber \\}

\newcommand{\bY}{{\bm {Y}}}

\newcommand{\bU}{{\bm {U}}}
\newcommand{\bV}{{\bm {V}}}
\newcommand{\bW}{{\bm {W}}}
\newcommand{\bZ}{{\bm {Z}}}

\newcommand{\bw}{{\bm {w}}}
\newcommand{\bv}{{\bm {v}}}

\newcommand{\bA}{{\bm {A}}}

\newcommand{\bX}{{\bm {X}}}
\newcommand{\bx}{{\bm {x}}}

\newcommand{\bu}{{\bm {u}}}

\newcommand{\be}{\begin{equation}}
\newcommand{\ee}{\end{equation}}

\newcommand{\bea}{\begin{align}}
\newcommand{\eea}{\end{align}}

\DeclareMathAlphabet{\varmathbb}{U}{bbold}{m}{n}

\newcommand{\EE}{\mathbb{E}}

\def\bx{{\bm x}}
\def\bW{{\bm W}}
\def\bX{{\bm X}}
\def\bY{{\bm Y}}
\def\bZ{{\bm Z}}

\begin{document}

\title{$0$--$1$ phase transitions in sparse spiked matrix estimation}

\author{Jean Barbier$^{1}$ and Nicolas Macris$^{2}$\\
\\
(1) International Center for Theoretical Physics, Trieste, Italy. jbarbier@ictp.it\\
(2) Ecole Polytechnique F\'ed\'erale de Lausanne, Switzerland. nicolas.macris@epfl.ch
}
\date{}
\maketitle
\thispagestyle{empty}
%
% {\let\thefootnote\relax\footnote{
% \!\!\!\!\!\!\!\!\!\!\!\!\!J.B. is with the Abdus Salam International Center for Theoretical Physics, Trieste, Italy. jbarbier@ictp.it\\
% N.M. is with the Ecole Polytechnique F\'ed\'erale de Lausanne, Switzerland. nicolas.macris@epfl.ch}}
\setcounter{footnote}{0}

\vskip 2cm 

\begin{abstract}
\noindent
We consider statistical models of estimation of a rank-one matrix (the spike) corrupted by an additive gaussian noise matrix in the sparse limit. In this limit the underlying hidden vector (that constructs the rank-one matrix) has a number of non-zero components that scales sub-linearly with the total dimension of the vector, and the signal strength tends to infinity at an appropriate speed. We prove explicit low-dimensional variational formulas for the asymptotic mutual information between the spike and the observed noisy matrix in suitable sparse limits. For Bernoulli and Bernoulli-Rademacher distributed vectors, and 
when the sparsity and signal strength satisfy an appropriate scaling relation, these formulas imply sharp $0$--$1$ phase transitions for the asymptotic minimum mean-square-error. A similar phase transition was analyzed recently in the context of sparse high-dimensional linear regression (compressive sensing) \cite{david2017high,reeves2019all}.  
\end{abstract}
{%
	\hypersetup{linkcolor=black}
	\setcounter{tocdepth}{2}
	\tableofcontents	
}

\section{Introduction}\label{sec:intro}

Low rank matrix estimation (or factorization) is an important problem with numerous applications in image processing, principal component analysis (PCA), machine learning, DNA microarray data, tensor decompositions, etc. These modern applications often require to look at the {\it high-dimensional limit and sparse limits} of the problem. Sparsity is often a crucial ingredient for the interpretability of high dimensional statistical models. In this context, it is of great importance to determine computational limits of estimation and to benchmark them by the fundamental information theoretical (i.e., statistical) limits. In this paper we concentrate on information theoretic limits for two probabilistic models, the so-called sparse spiked Wishart and Wigner matrix models.
%Often one is given a large low-rank data matrix with the task to factorise it, estimate its principal components, or complete missing entries. Such problems have found numerous applications and have spurred the interest of various disciplines such as high-dimensional statistics, signal processing, learning theory, theoretical computer science. We refer to \cite{le gars belge qui etait a Bologne} for a recent overview.
%Much progress has been accomplished recently on statistical versions of matrix estimation (factorisation) problems called "matrix spike models" \cite{DBLP:journals/corr/MontanariR14}. 

In the simplest rank-one version one seeks a matrix $\bU\otimes \bV$ constructed from high-dimensional {\it hidden} vectors
$\bU = (U_1, \ldots, U_n)\in \mathbb{R}^n$ and $\bV = (V_1, \ldots, V_m)\in \mathbb{R}^m$, $m=\alpha_n n$, based on a noisy {\it observed} data matrix $\bW$ with entries obtained as
$
W_{ij} \sim {\cal N}(\sqrt{\lambda_n/n}\, U_i V_j,1)
$
for $i=1, \dots, n$, $j=1, \dots, m$ and $\lambda_n > 0$ the signal strength. The hidden vectors have independent identically distributed (i.i.d.) components drawn from two different distributions. The high-dimensional limit means that we look at 
$n,m\to +\infty$, $\alpha_n\to \alpha >0$. We suppose that $\bV$ has on average $\rho_{V, n} m$ non-zero component which scales {\it sub-linearly} for a sequence $\rho_{V, n}\to 0_+$. We will see that non-trivial estimation is only possible if 
$\lambda_n\to +\infty$ (whereas if $\rho_{V,n}\to \rho_V > 0$, $\lambda_n\to \lambda >0$ finite). The problem is to estimate $\bU\otimes\bV$ given the data matrix $\bW$\footnote{One may also be interested in reconstructing the vectors $\bU$ and/or $\bV$ rather than the spike, but in general this is only possible up to a global sign.}. In the Bayesian setting, which is our concern here, it is supposed that the priors and hyper-parameters are all known. We will refer to this problem as the {\it sparse spiked Wishart matrix model}.

A popular version of this model, and one addressed here, corresponds to a fixed 
standard gaussian distribution for $\bU \sim \mathcal{N}(0, {\rm I}_{n})$ (${\rm I}_{n}$ is the $n\times n$ identity matrix) and a Bernoulli-Rademacher distribution for $V_i\sim P_{V,n} = (1-\rho_{V,n})\delta_0  +  \frac12\rho_{V, n}(\delta_{-1}+\delta_1)$. This estimation problem is equivalent to the important ``spiked covariance model'' or ``gaussian sparse-PCA'' \cite{johnstone2001distribution,johnstone2004sparse} which amounts to estimate a sparse binary matrix $\bV\otimes \bV$ from samples generated by the normal law $\mathcal{N}(0, {\rm I}_{n} + \frac{\lambda}{n} \bV\otimes \bV)$. 
%We show that this estimation problem displays a $0-1$ phase transition as explained below.

An even simpler and paradigmatic matrix estimation problem has a {\it symmetric} data matrix $\bW$ with elements drawn as
$
W_{ij} \sim {\cal N}(\sqrt{\lambda_n/n}\,X_i X_j,1)
$ 
for $1\leq i < j \leq n$ and $\bX = (X_1, \dots, X_n)\in \mathbb{R}^n$ with i.i.d. components, with 
$n\to +\infty$ in the high-dimensional limit. Again, the sparse version corresponds to having a {\it sub-linear} number of non-zero components, i.e.,  $\rho_n n$ with $\rho_n\to 0_+$, and non-trivial estimation is possible only for $\lambda_n\to +\infty$. We call this model the {\it sparse spiked Wigner matrix model}.
We will focus in particular
on binary vectors generated from Bernoulli $X_i \sim P_{X,n}= {\rm Ber}(\rho_n)$ or Bernoulli-Rademacher $X_i\sim P_{X,n}=(1-\rho_n)\delta_0 + \frac12\rho_n(\delta_{-1}+\delta_1)$ distributions. 
%We show that this estimation problem also displays a $0-1$ phase transition.

\subsection{Background}

Much progress has been accomplished in recent years on spiked matrix models for {\it non-sparse} settings, by which we mean that the distributions 
$P_X$, $P_U$, $P_V$ are fixed independent of $n, m$, and thus the number of non-zero components of $\bX$, $\bV$, even if ``small'', scales {\it linearly} with $n$. An interesting phenomenology of information theoretical (or statistical) as well as computational limits has been derived \cite{2017arXiv170100858L} by heuristic methods of statistical physics of spin glass theory (the so-called replica method). In the asymptotic regime of $n\!\to \!+\infty$ these limits take the form of sharp phase transitions.
The rigorous mathematical theory of these phase transitions is now largely under control. On one hand, the approximate message passing (AMP) algorithm has been analyzed by state evolution \cite{BayatiMontanari10,Donoho10112009}. And on the other hand, the asymptotic mutual informations per variable between hidden spike and data matrices, have been rigorously computed in a series of works using various methods (cavity method, spatial coupling, interpolation methods) \cite{korada2009exact,krzakala2016mutual,XXT,2016arXiv161103888L,2017arXiv170200473M,BarbierM17a,BarbierMacris2019,2017arXiv170910368B,el2018estimation,barbier2019mutual,mourrat2019hamilton}. The information theoretic phase transitions are then signalled by singularities, as a function of the signal strength, in the limit of the mutual information per variable when $n\to +\infty$.
The phase transition also manifests itself as a jump discontinuity in the minimum mean-square-error (MMSE)\footnote{This is the generic singularity and one speaks of a first order transition. In special cases the MMSE may be continuous with a higher discontinuous derivative of the mutual information.}. 
Once the mutual information is known it is usually possible to deduce the MMSE. For example, in the simplest case of the spiked Wigner model, if $I(\bX\otimes \bX; \bW)$ is the mutual information between the spike $\bX\otimes \bX$ and the data $\bW$, the 
${\rm MMSE}(\bX\otimes\bX\vert \bW) = \mathbb{E} \Vert \bX \otimes\bX - \mathbb{E}[\bX \otimes\bX \vert \bW]\Vert_{\rm F}^2$ satisfies  the I-MMSE relation (such relations are derived in \cite{GuoShamaiVerdu_IMMSE,guo2011estimation}, see also appendix \ref{app:gaussianchannels})
$$
\frac{d}{d\lambda_n} \frac{1}{n} I(\bX\otimes\bX; \bW) = \frac{1}{4n^2}{\rm MMSE}(\bX\otimes\bX\vert \bW)+O(n^{-1})\,.
$$
%For the spiked Wishart model the situation is richer and discussed in section \ref{sec:matrix}.
Closed form expressions for the asymptotic mutual information \cite{2016arXiv161103888L,2017arXiv170200473M,XXT,BarbierM17a,BarbierMacris2019,2017arXiv170910368B} therefore allow
to benchmark the fundamental information theoretical limits of estimation. See also \cite{perry2018optimality,alaoui2017finite,alaoui2018detection} for results on the limits of detecting the precense of a spike in a noisy matrix, rather than estimating it.

\begin{figure}[t]
\centering
\includegraphics[width=0.49\linewidth]{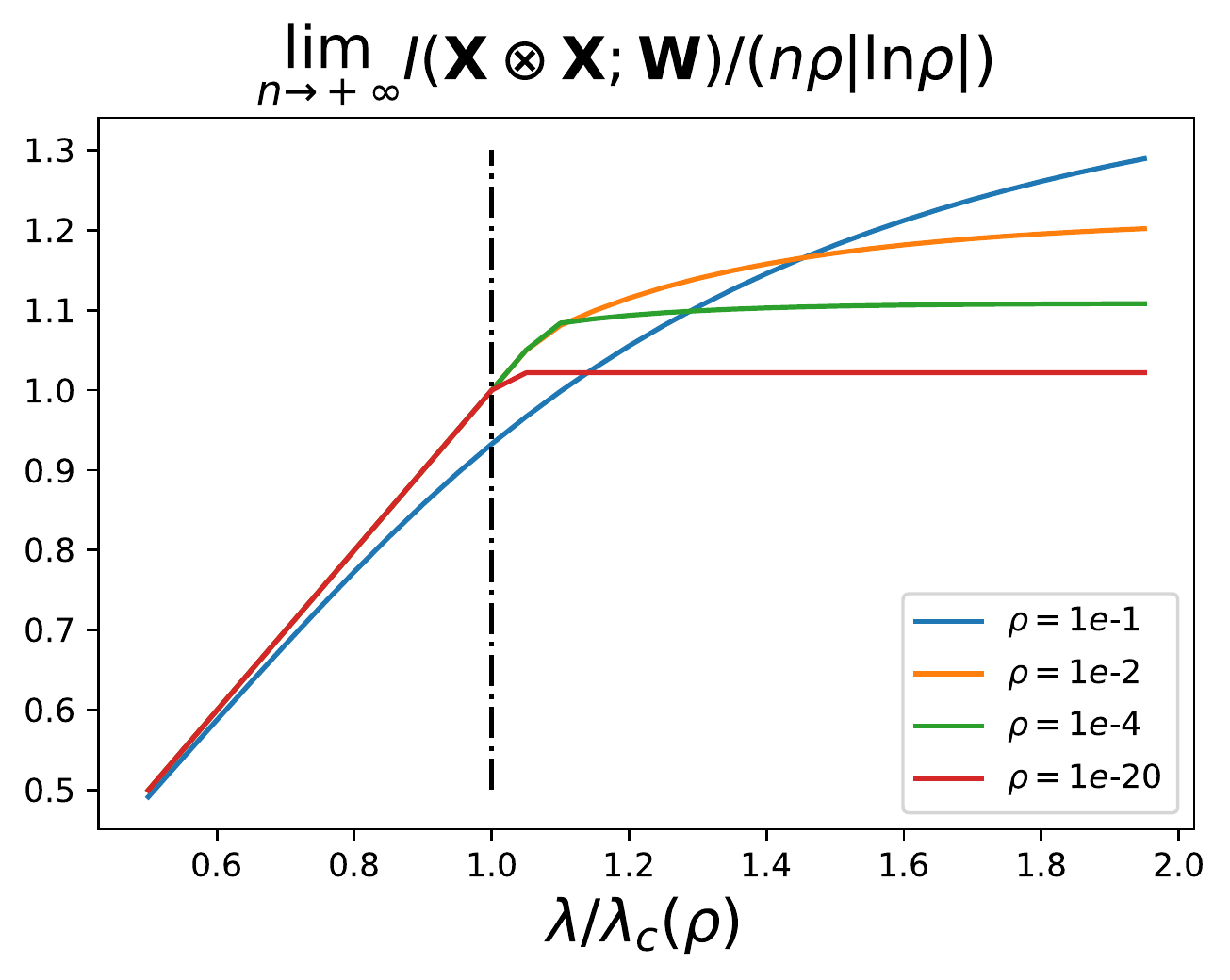}
\includegraphics[width=0.49\linewidth]{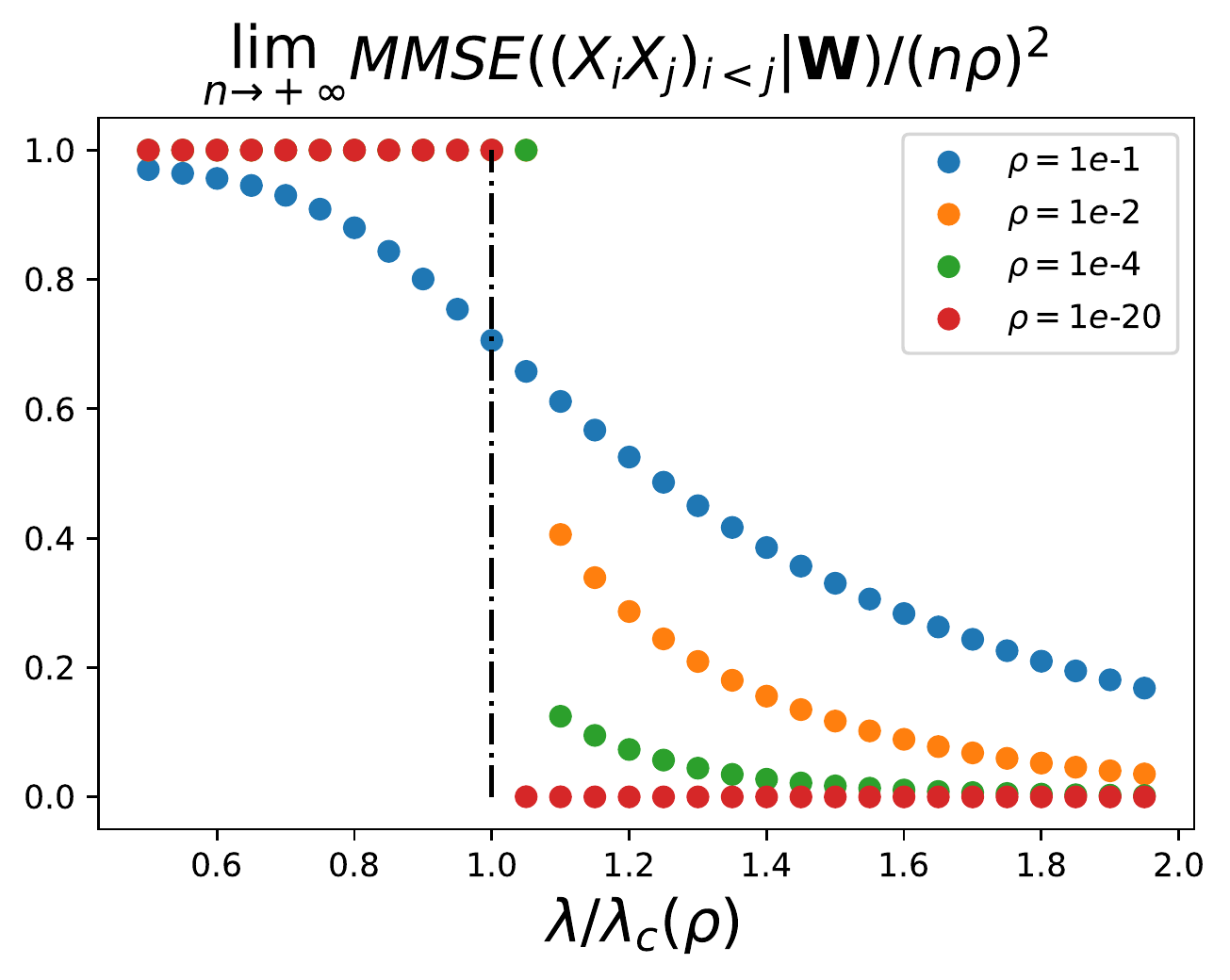}
 \caption{\footnotesize A sequence of suitably normalized mutual information and minimum mean-square-error (MMSE) curves as a function of $\lambda/\lambda_c(\rho)$ for the symmetric matrix estimation model for $X_i \sim {\rm Ber}(\rho)$. Here $\lambda_c(\rho) = 4\vert \ln \rho\vert/\rho$. In the sparse limit $\rho\to 0$ the MMSE curves approach a $0$--$1$ phase transition with the discontinuity at $\lambda = \lambda_c(\rho)$. This corresponds to an angular point for the mutual information.}
\label{fig:MMSEandMI}
\end{figure}

\subsection{Our contributions}

In this paper we are exclusively interested in determining information theoretic phase transitions in regimes of {\it sub-linear sparsity}. 
We identify the correct {\it scaling regimes} of vanishing sparsity and diverging signal strength in which non-trivial information theoretic phase transitions occur. We use the {\it adaptive interpolation method} \cite{BarbierM17a,BarbierMacris2019,2017arXiv170910368B} first introduced in the non-sparse matrix estimation problems, to provide for the sparse limit, closed form expressions of the mutual information in terms of low-dimensional variational expressions (theorems \ref{thm:ws} and \ref{thm:wishart} in section \ref{sec:matrix}). That the adaptive interpolation method can be extended to the sparse limit is interesting and not a priori obvious.
By the I-MMSE relation and the solution of the variational problems we then find, for Bernoulli and Bernoulli-Rademacher distributions of the sparse signal, that the MMSE displays to a {\it $0$--$1$ phase transition} and we determine the {\it exact thresholds}.

%Information theoretic limits for high dimensional sparse-PCA with the additional structure of sparse eigenvectors were rigorously determined (up to constants) in \cite{amini2009} for large but finite sizes. Our results are consistent with these estimates. 

%For the Wigner spiked model would like to have a theory in the case of distributions of the form 
%$P_{X, n} = (1-\rho_n ) \delta_0 + \rho_n p_X$
%where $\rho_n\to 0$ and $p_X$ is fixed reasonable distribution. This captures a truly sparse limit where the average number of non-zero components 
%is $\rho_n n = o(n)$. For the Wishart spiked model we would 
%like to analyze the case of fixed $P_U$ and $P_{V,n} = (1-\rho_{V,n} ) \delta_0 + \rho_{V,n}p_V$
%where again $\rho_n\to 0$ and $p_X$ is fixed reasonable distribution. In the later case when $\bU \sim \mathcal{N}(0, {\rm I}_n)$ the problem is equivalent to the important {\it spiked covariance model} \ref{johnstone2001}, \cite{johnstone2004sparse}
%in which one must estimate $\bV$ from samples distributed as $\mathcal{N}(0, {\rm I}_n + \frac{\lambda}{n}\bV \otimes\bV)$.

Let us describe the regimes studied and the information theoretical thresholds found here (precise statements are found in section \ref{sec:matrix}).
We first note that for sub-linear sparsity, a phase transition appears only if the signal strength tends to infinity. For the Wigner case, for example, this can be seen from the following heuristic argument: the total signal-to-noise ratio (SNR) per non-zero component (i.e., SNR per observation $(\lambda_n/n)\rho_n^2$ times the number of observations $\Theta(n^2)$ divided by the number of non-zero components $\rho_n n$) scales as $(\lambda_n/n)\rho_n^2n^2/(\rho_n n) = \lambda_n \rho_n$ so that $\lambda_n \to +\infty$ is necessary in order to have enough energy to estimate the non-zero components. 
Our analysis shows that non-trivial phase transitions occur when $\lambda_n = \Theta(\vert \ln \rho_n\vert/ \rho_n)$ (Wigner case) and 
$\lambda_n = \Theta(\sqrt{\vert \ln \rho_{V, n}\vert/ \rho_{V, n}})$ (Wishart case)
when $\rho_n$ and $\rho_{V,n}$ tend to zero slowly enough. 

We study in particular the cases of binary signals, i.e., $P_{X,n}$ and $P_{V, n}$ equal to ${\rm Ber}(\rho_n)$ or Bernoulli-Rademacher $(1-\rho_n)\delta_0+\frac12\rho_n(\delta_{-1}+\delta_1)$. For these distributions we find $0$--$1$ phase transitions at the level of the MMSE as long as $\rho_n\to 0_+$ and $\rho_{V,n}\to 0_+$ {\it not too fast}. 
This is illustrated on figure \ref{fig:MMSEandMI} for the Wigner case with Bernoulli distribution. The left hand side shows that as $\rho_n\to 0_+$ the (suitably normalized) mutual information approaches the broken line with an angular point at $\lambda /\lambda_c(\rho_n) =1$ where 
$\lambda_c(\rho_n) = 4 \vert \ln\rho_n\vert /\rho_n$; in the case of Bernoulli-Rademacher distribution the threshold is the same. On the right hand side the (suitably normalized) MMSE approaches a $0$--$1$ curve: it tends to $1$ for $\lambda /\lambda_c(\rho_n) <1$, develops a jump discontinuity at $\lambda /\lambda_c(\rho_n)=1$, and takes the value $0$ when $\lambda /\lambda_c(\rho_n)>1$. 
A similar $0$--$1$ transition is found to hold for the MMSE of $\bV\otimes\bV$ in the spiked covariance model with a threshold
$\lambda_c(\rho_{V, n}) = \sqrt{4 \vert \ln\rho_n\vert / (\alpha_n\rho_n)}$ (with $\alpha_n \to \alpha$). This is illustrated 
on figure \ref{fig:MMSEandMIWish} in section \ref{sec:matrix}. Note that these figures are obtained from the asymptotic prediction where \emph{first} $n\to+\infty$ and \emph{then} $\rho \to 0_+$, so not in the sub-linear sparsity regime. Our analysis confirms that this picture with its sharp transition holds in the \emph{truly} sparse (sub-linear) regime $\rho_n\to 0_+$ \emph{with} $n\to+\infty$.

\subsection{Related work}

Spiked matrix ensembles have played a crucial role in the analysis of threshold phenomena in high-dimensional statistical models for almost two decades. Early rigorous results are found in \cite{baik2005phase} who determined by spectral methods the location of the information theoretic phase transition point in a spiked covariance model, and  \cite{peche2006largest,feral2007largest} for the Wigner case. 
More recently, the information theoretic limits and those of hypothesis testing have been derived, with the additional structure of sparse vectors, for large but finite sizes \cite{amini2009,pmlr-v75-brennan18a,gamarnik2019overlap}. These estimates are consistent with our results.
The additional feature that we provide here, is an asymptotic limit in which a sharp $0$--$1$ phase transition is identified, with fully explicit formulas for the thresholds. Moreover closed form expressions for the mutual information are also determined.

The $0$--$1$ transitions and formulas for the thresholds and mutual information were first computed in \cite{2017arXiv170100858L} using the heuristic replica method of spin-glass theory. However, it must be stressed that, not only this calculation is far from rigorous, but more importantly the limit $n\!\to \!+\infty$ is first taken for fixed parameters $\rho_n\!=\!\rho$, $\rho_{V,n}\!=\!\rho_V$, and the sparse limit $\rho, \rho_V\!\to\! 0_+$ is taken only after. Although the thresholds found in this way agree with our derivation of $\lambda_c(\rho_n)$, this is far from evident a priori. For example, it not clear if this sort of approach yields correct computational thresholds in the sparse limit \cite{JMLR:v17:15-160,2017arXiv170100858L}. 

Similar phase transitions in sublinear sparse regimes for binary signals (Bernoulli or Bernoulli-Rademacher) have been studied in the context of linear estimation or compressed sensing 
\cite{david2017high,reeves2019all} for support recovery. These works focus on the MMSE and prove the occurence of the $0$--$1$ phase transition which they call an ``all-or-nothing'' phenomenon. We note that our approach is technically very different in that it determines the variational expressions for mutual informations and finds the transitions as a consequence.

A lot of efforts have been devoted to computational aspects of sparse PCA with many remarkable results \cite{deshpande2014information,Cai2015,krauthgamer2015,JMLR:v17:15-160,wang2016,pmlr-v30-Berthet13,Ma:2015:SLB:2969239.2969419,gamarnik2019overlap,pmlr-v75-brennan18a}. The picture that has emerged is that the information theoretic and computational phase transition regimes are not on the same scale and that the computational-to-statistical gap diverges in the limit of vanishing sparsity. Note that this is also seen within the context of state evolution for the AMP algorithm 
\cite{2017arXiv170100858L}, but with the sparse limit taken after the $n\to +\infty$ limit. It would be desirable to rigorously determine the thresholds of the AMP algorithm and the correct scaling regime of $\lambda_n\to +\infty$ and  $\rho_n\to 0_+$ or $\rho_{V,n}\to 0_+$ where a computational phase transition is observed. We believe that techniques developed for compressed sensing with finite size samples \cite{RushVenkataramanan} could also apply here. 

\section{Sparse spiked matrix models: setting and main results}\label{sec:matrix}
\subsection{Sparse spiked Wigner matrix model}
We consider a sparse signal-vector $\bX = (X_1, \ldots, X_n)\in \mathbb{R}^n$ with $n$ i.i.d. components distributed according to 
$P_{X,n} = \rho_n p_X+(1-\rho_n)\delta_0$.
Here $\delta_0$ is the Dirac mass at zero and $(\rho_n)\in (0,1]^\mathbb{N}$ is a sequence of weights.
For the distribution $p_X$ we assume that : $i)$ it is independent of $n$, $ii)$ it has finite support in an interval $[-S,S]$, $iii)$ it has second moment equal to $1$ (without loss of generality). One has access to the symmetric data matrix $\bW\in \mathbb{R}^{n\times n}$ with noisy entries
\begin{align}\label{WSM}
\bW = \sqrt{\frac{\lambda_n}{n}} \bX\otimes \bX +  \bZ\,, \quad 1\le i<j\le n\,,
\end{align}
where $\lambda_n >0$ controls the strength of the signal and the noise is i.i.d. gaussian $Z_{ij}\sim{\cal N}(0,1)$ for $i<j$ and symmetric $Z_{ij}=Z_{ji}$. 

We are interested in sparse regimes where $\rho_n\to 0_+$ and $\lambda_n\to +\infty$. While our results are more general (see appendix \ref{app:matrix} and theorem \ref{thm:wsgeneral}) our main interest is in a regime of the form
\begin{align}\label{mainregime}
\lambda_n = \frac{4\gamma \vert \ln \rho_n\vert}{\rho_n}, \qquad \rho_n = \Theta(n^{-\beta}), 
\end{align}
for $\beta, \gamma \in \mathbb{R}_{\ge 0}$ and $\beta$ small enough. We prove that in this regime a phase transition occurs as function of $\gamma$. The phase transition manifests itself as a singularity (more precisely a discontinuous first order derivative) in the mutual information
$I(\bX\otimes \bX;\bW)=H(\bW)-H(\bW|\bX\otimes \bX)$. Note that because the data $\bW$ depends on $\bX$ only through $\bX\otimes \bX$ we have $H(\bW|\bX\otimes \bX)=H(\bW|\bX)$ and therefore $I(\bX\otimes \bX; \bW)=I(\bX; \bW)$. From now on we use the form $I(\bX; \bW)$. To state the precise result we define the {\it potential function}:
\begin{align}\label{26}
i_n^{\rm pot}(q, \lambda,\rho) \equiv \frac{\lambda}{4} (q-\rho)^2+I_n(X;\sqrt{\lambda q}X+Z)\, ,
\end{align}
where $I_n(X;\sqrt{\lambda q}X+Z)$ is the mutual information for a scalar gaussian channel, with $X\sim P_{X,n}$ and $Z\sim{\cal N}(0,1)$. The mutual information $I_n$ is indexed by $n$ because of its dependence on $P_{X,n}$.
\begin{thm}[Sparse spiked Wigner model]\label{thm:ws}
Let the sequences $\lambda_n$ and $\rho_n$ verify \eqref{mainregime} with $\beta\in [0, 1/6)$ and $\gamma >0$. There exists $C>0$ independent of $n$ such that 
\begin{align}\label{mainbound}
\frac{1}{\rho_n|\ln\rho_n|}\Big|\frac{1}{n}I(\bX;\bW) - \inf_{q\in [0,\rho_n]} i^{\rm pot}_n(q;\lambda_n,\rho_n)\Big| 
\le C\frac{(\ln n)^{1/3}}{n^{(1-6\beta)/7}}\,.
\end{align}	
\end{thm}

%
% An analysis of the variational problem for discrete distributions $p_X$ shows that the limit of the f
% $(\rho_n|\ln \rho_n|)^{-1} \inf_{q\in [0,\rho_n]} i^{\rm pot}_n(q;\lambda_n,\rho_n)$ 
% as $n\to \infty$ exists and is finite, and therefore this also the case for $(n\rho_n|\ln \rho_n|)^{-1}I(\bX;\bW)$. 
The mutual information is thus given, to leading order, by a one-dimensional variational problem 
$$
I(\bX;\bW) = n\rho_n \vert \ln \rho_n \vert \inf_{q\in [0,\rho_n]} i^{\rm pot}_n(q;\lambda_n,\rho_n) + {\rm correction\,\,terms}\,.
$$
The factor $\rho_n|\ln \rho_n|$ is naturally related to the entropy (in nats) of the support of the signal
$- n (\rho_n\ln\rho_n + (1-\rho_n)\ln(1-\rho_n))$ which behaves like $n \rho_n\vert \ln\rho_n\vert$ for $\rho_n\to 0_+$. 

In particular, for both the Bernoulli and Bernoulli-Rademacher distributions an analytical solution of the variational problem (given in appendix \ref{sec:allOrNothing}) shows that
\begin{align}\label{limitingpotential}
\lim_{n\to +\infty}\frac{1}{n \rho_n |\ln \rho_n|}I(\bX;\bW) = \gamma \mathbb{I}(\gamma \le 1)+\mathbb{I}(\gamma \geq 1)\,.
\end{align}
This is also seen numerically on figure \ref{fig:MMSEandMI} (for the Bernoulli case). The I-MMSE relation (see introduction) 
then shows that the suitably rescaled MMSE is simply given by a derivative w.r.t. 
$\gamma$ and therefore displays a $0$--$1$ phase transition at $\gamma=1$ (or equivalently at the critical threshold
$\lambda_c(\rho_n) = 4\vert\ln\rho_n\vert/ \rho_n$) as depicted on the right hand side of figure \ref{fig:MMSEandMI}. We do not claim that 
\eqref{limitingpotential} and the consequence for the MMSE are rigorously derived. However these results are ``contained'' in the variational expression for the mutual information and are ``mere consequences'' of a precise analysis of this one-dimensional variational problem.

For more generic distributions than these two cases the situation is richer.
Although one generically observes phase transitions {\it in the same scaling regime}, the limiting curves appear to be more complicated than the simple $0$--$1$ shape and the jumps are not necessarily located at $\gamma=1$. A classification of these transitions is an interesting problem that is out of the scope of this paper.
\subsection{Sparse spiked Wishart model}
The sparse spiked Wishart model is a non-symmetric version of the previous one. There are two distinct vectors $\bU =(U_1,\ldots, U_n)\in \mathbb{R}^{n}$ and $\bV=(V_1,\ldots, V_m)\in\mathbb{R}^{m}$ with dimensions of the same order of magnitude. We set $m = \alpha_n n$ and will let 
$\alpha_n\to \alpha >0$ as $n \to +\infty$. 
%Without loss of generality we consider $\alpha_n\le \alpha_{\rm max}$ (so it can tend to a constant or vanish). 
The data matrix $\bW\in\mathbb{R}^{n\times m}$ is
\begin{align*}
	\bW=\sqrt{\frac{\lambda_n}{n}}\bU \otimes \bV+\bZ
\end{align*}
where $\lambda_n >0$ and $\mathbb{R}^{n\times m}\ni\bZ=(Z_{ij})_{i,j}$, $i=1,\ldots, n$, $j=1,\ldots, m$, is a Wishart noise matrix with i.i.d. standard gaussian entries. Both the entries of $\bU$, $\bV$ are i.i.d. and drawn from possibly sparse distributions. Specifically
$U_i\sim P_{U,n}\equiv \rho_{U,n} p_U+(1-\rho_{U,n})\delta_0$ and $V_i\sim P_{V,n}\equiv \rho_{V,n} p_V+(1-\rho_{V,n})\delta_0$. 
We assume that both $p_U$ and $p_V$ have finite support included in an interval $[-S,S]$ and (without loss of generality) they both have unit second moment. 

Our main interest is in regimes of the form
\begin{align}\label{wishregime}
\alpha_n \to \alpha > 0\,, \quad \rho_{U,n}\to \rho_U > 0\,, \quad \rho_{V,n} = \Theta(n^{-\beta}),  \quad \lambda_n = \sqrt{\frac{4\gamma|\ln \rho_{V,n}|}{\alpha_n \rho_{V,n}}}\,.
\end{align}
This scaling allows to greatly simplify the analysis and is the proper scaling regime to observe the information theoretic phase transition. Many of our results hold in wider generality (see appendix \ref{app:matrix}). The main result is again a variational expression for the mutual information $I(\bU\otimes \bV;\bW)=I((\bU,\bV);\bW)=H(\bW)-H(\bW|\bU,\bV)$ between the spike (or signal-vectors) and the data matrix, in terms of a {\it potential function}:
\begin{align}
i_n^{\rm pot}(q_U,q_V,\lambda,\alpha,\rho_U,\rho_V) &=  \frac{\lambda \alpha}{2} (q_U-\rho_U)(q_V-\rho_V)
\nn
&\qquad\qquad\qquad+ I_{n}(U;\sqrt{\lambda \alpha q_V}U+Z) + \alpha I_{n}(V;\sqrt{\lambda q_U}V+Z)\, \label{pot_wishart}
\end{align}
where $I_{n}(U;\sqrt{\lambda \alpha q_V}U+Z)$ is the mutual information for a scalar gaussian channel, with $U\sim P_{U,n}$ and $Z\sim{\cal N}(0,1)$, while $I_{n}(V;\sqrt{\lambda q_U}V+Z)$ is with $V\sim P_{V,n}$. 
Our main result reads:
\begin{thm}[Sparse spiked Wishart model]\label{thm:wishart_main} 
Consider the scaling regime \eqref{wishregime} with $\beta\in [0, 1/3)$. There exists a constant $C>0$ independent of $n$ such that
\begin{align*} 
&\frac{1}{\sqrt{\rho_{V,n}\vert\ln\rho_{V,n}\vert}} \Big\vert \frac{1}{n}I\big((\bU,\bV);\bW\big) - \!\!
 {\adjustlimits \inf_{q_U\in[0,\rho_{U,n}]}\sup_{q_{V}\in[0,\rho_{V,n}]}}\,i_n^{\rm pot}\big(q_{U},q_V,\lambda_n,\alpha_n,\rho_{U,n},\rho_{V,n}\big) \Big\vert \nn
 &\qquad\qquad\qquad\leq C\frac{(\ln n)^{1/3}}{n^{(4-12\beta)/18}}\,.
\end{align*}
\end{thm}

To leading order the mutual information is given by the solution of a two-dimensional variational problem. An analytical solution of this problem for the spiked covariance model $\bU\sim\mathcal{N}(0, {\rm I}_{n})$ and $V_i \sim (1-\rho_{V, n}) \delta_0 + \frac12\rho_{V,n}(\delta_{-1}+\delta_1)$ shows
(appendix \ref{sec:allOrNothing})
$$
\frac{1}{n\sqrt{\rho_{V,n}\vert\ln\rho_{V,n}\vert}} I\big((\bU,\bV);\bW\big) = \sqrt{\alpha\gamma} + \alpha (1-\gamma) \mathbb{I}(\gamma \geq 1) \sqrt{\rho\vert \ln \rho\vert} + {\rm correction\,\,terms}\,.
$$
Here we see that the phase transition is washed out at leading order and only seen at higher order with a threshold at $\gamma=1$, i.e., $\lambda_c(\rho_{V,n}) = \{4\vert \ln \rho_{V,n}\vert/(\alpha \rho_{V,n})\}^{1/2}$. Note that in the present regime $\rho_{V, n}\to 0$ and $\gamma = \Theta(1)$ so the mutual information remains positive.

The consequences of this formula for the MMSE are richer and more subtle than in the symmetric Wigner case. One can consider three MMSE's
associated to the matrices $\bV\otimes\bV$, $\bU\otimes \bU$, or $\bU\otimes\bV$. All three MMSE's can be computed from the solution $(q_U^*, q_V^*)$ in the variational problem of theorem~\ref{thm:wishart_main}, as shown in \cite{miolane2018phase}.
We have $(\alpha n)^{-2} {\rm MMSE}(\bV\otimes\bV\vert \bW) = \mathbb{E}[V_1^2]^2 - (q_V^*)^2$,
$n^{-2} {\rm MMSE}(\bU\otimes\bU\vert \bW) = \mathbb{E}[U_1^2]^2 - (q_U^*)^2$ and $(n m)^{-1}{\rm MMSE}(\bU\otimes\bV\vert\bW) = \mathbb{E}[U_1^2]\mathbb{E}[V_1^2] - q_U^* q_V^*$. We note that the last expression is equivalent to an I-MMSE relation, i.e., it can be obtained by differentiating the 
mutual information $(2/m)I((\bU,\bV);\bW)$ with respect to $\lambda_n$. An application of these formulas to the analytical solutions of the variational 
problem (found in appendix \ref{sec:allOrNothing}) shows that with suitable rescaling 
$(\alpha n \rho_{V,n})^{-2} {\rm MMSE}(\bV\otimes\bV\vert \bW)$
 displays the $0$--$1$ phase transition.
For the other two MMSE's one cannot expect to see such behavior because $\bU$ is gaussian. Instead one finds asymptotically that these MMSE's (with suitable rescaling) tend to $1$ when $\rho_{V,n} \to 0_+$. The transition at $\gamma=1$ is a higher order effect seen on higher order corrections. These results are illustrated with a numerical calculation depicted on figure \ref{fig:MMSEandMIWish}. 

\begin{figure}[t]
\centering
\includegraphics[width=0.49\linewidth]{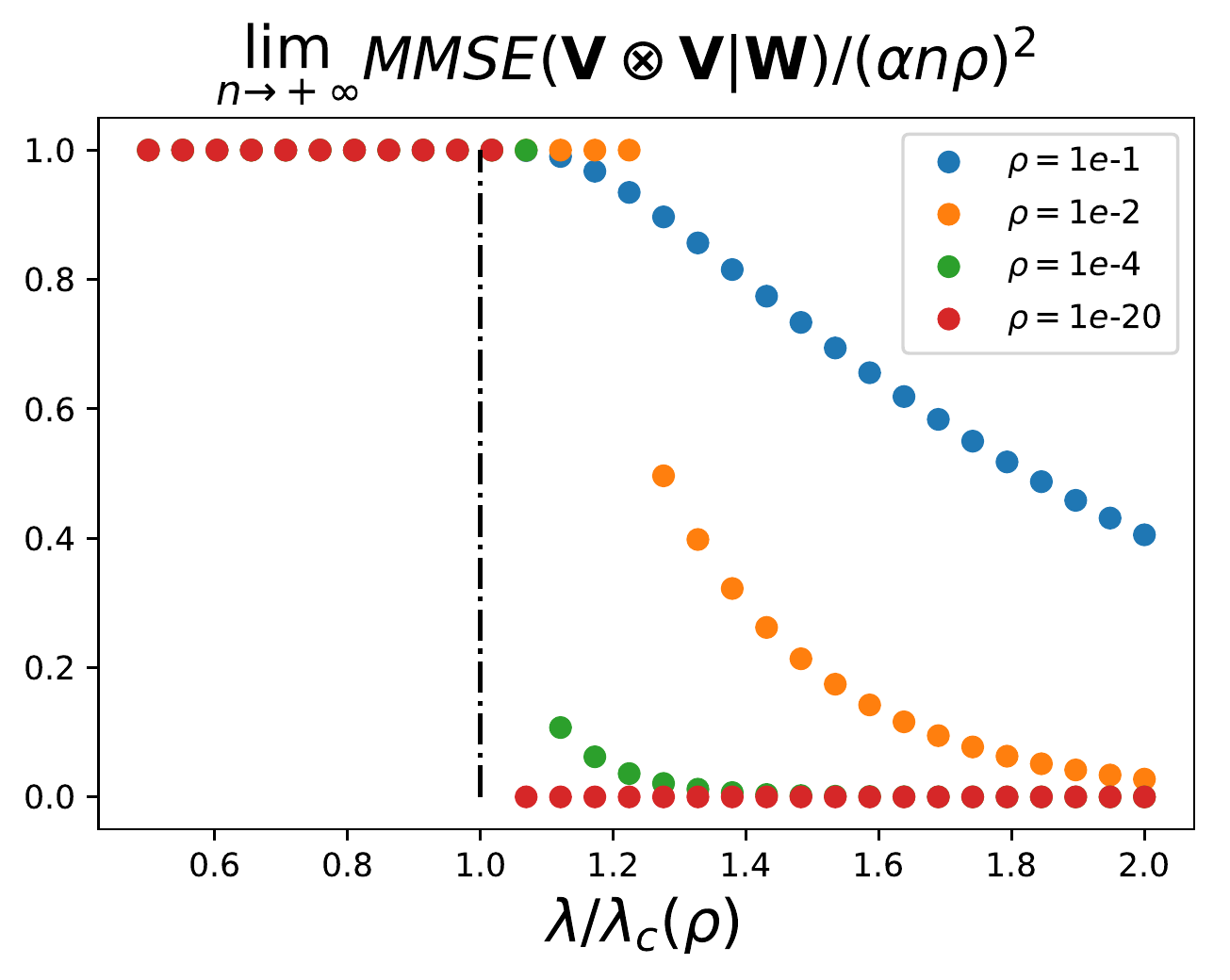}
\includegraphics[width=0.49\linewidth]{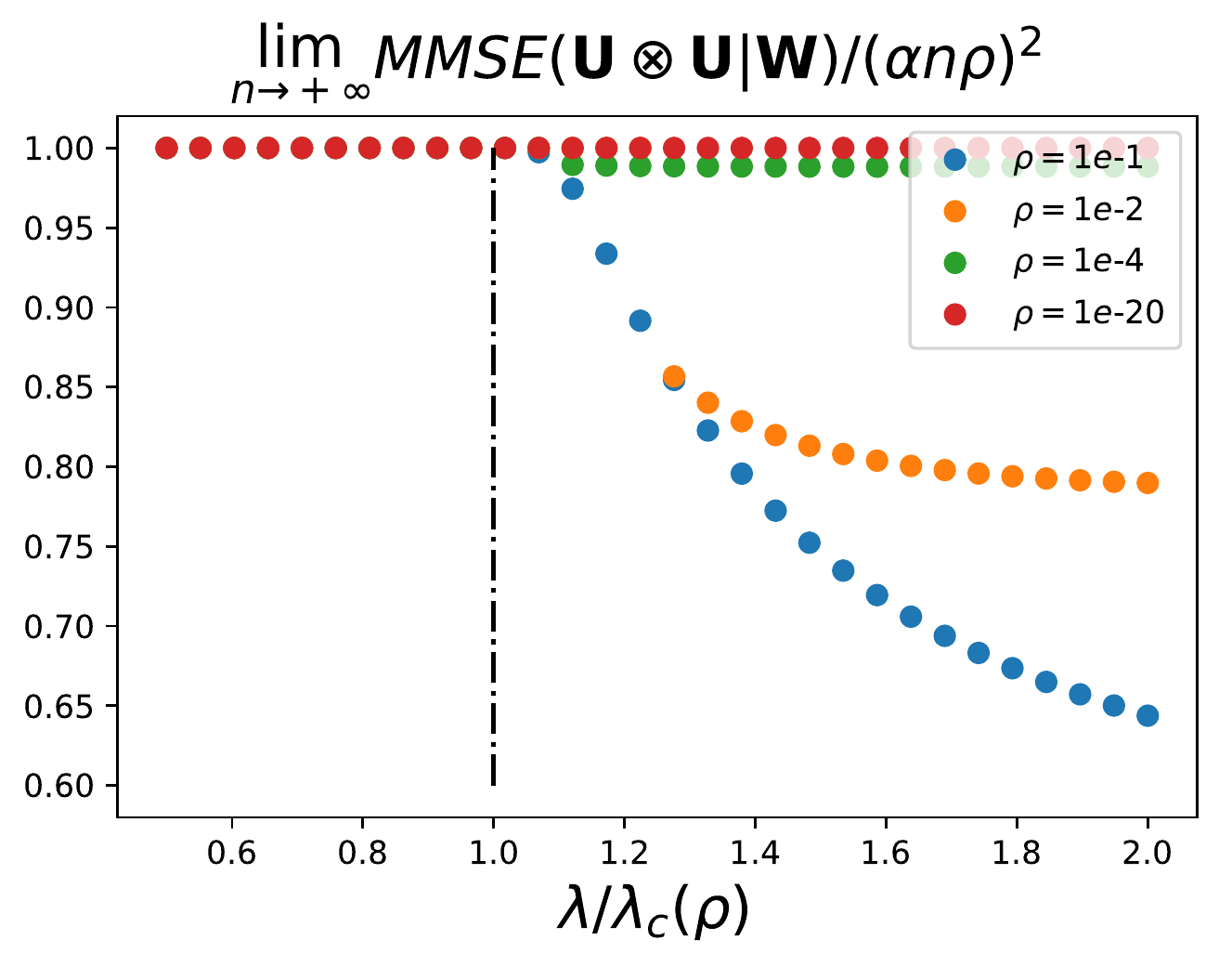}
 \caption{\footnotesize Sequence of MMSE curves as a function of $\lambda$ for the spiked covariance model. \textbf{Left}: In the sparse limit $\rho\to 0_+$ the suitably rescaled 
 of the sparse signal $\bV$ approaches a $0$--$1$ transition with a jump discontinuity at 
 $\lambda_c(\rho)= \{4\gamma \vert \ln \rho\vert/ (\alpha\rho)\}^{1/2}$. \textbf{Right}: In the asymptotic limit the MMSE for the gaussian signal $\bU$ approaches $1$ as $\rho\to 0_+$. The phase transition is seen only as a higher order effect.}
\label{fig:MMSEandMIWish}
\end{figure}

\section{Analysis by the adaptive interpolation: the Wigner case}\label{sec:adapInterp_XX}
In this section we provide the essential architecture for the proof of theorem \ref{thm:ws} which relies on the adaptive interpolation method
\cite{BarbierM17a,BarbierMacris2019}. The proof requires concentration properties for ``free energies'' and ``overlaps'' which are deferred to appendices \ref{app:free-energy} and \ref{appendix-overlap}. We will also employ various known information theoretic properties of 
gaussian channels (I-MMSE relation, concavity of the MMSE with respect to the SNR and input distribution etc). For the convenience of the reader these are presented and adapted to our setting in appendix \ref{app:gaussianchannels}.

An essentially similar analysis can be done for theorem 
\ref{thm:wishart_main} in the Wishart case, and is deferred to appendix \ref{app_wishart}.
When no confusion is possible we use the notation $\EE\|\bA\|^2=\EE[\|\bA\|^2]$.
\subsection{The interpolating model.}

Let $\epsilon \in[s_n, 2s_n]$, for a sequence tending to zero, $s_n = n^{-\alpha}/2 \in(0,1/2)$, for $\alpha >0$ chosen later on. Let $q_{n}: [0, 1]\times [s_n,2s_n] \mapsto [0,\rho_n]$ and set 
$$
R_n(t,\epsilon)\equiv \epsilon+\lambda_n\int_0^tds\,q_n(s,\epsilon)\,.
$$
Consider the following interpolating estimation model, where $t\in[0,1]$, with accessible data $(W_{ij}(t))_{i,j}$ and $(\tilde W_i(t,\epsilon))_i$ obtained through
\begin{align*}
\begin{cases}
W_{ij}(t) =W_{ji}(t)\hspace{-5pt}&=\sqrt{(1-t)\frac{\lambda_n}{n}}\, X_iX_j + Z_{ij}\,, \qquad 1\le i< j\le n\,,\\
\tilde \bW(t,\epsilon)  &= \sqrt{R_n(t,\epsilon)}\,\bX + \tilde \bZ\,,
\end{cases} 
\end{align*}
with standard gaussian noise $\tilde \bZ\sim {\cal N}(0,{\rm I}_n)$, and $Z_{ij}=Z_{ji}\sim {\cal N}(0,1)$. 
%Simplifying all the $\bx$-independent terms with the normalization 
The posterior associated with this model reads (here $\|-\|$ is the $\ell_2$ norm)
\begin{align*}
&dP_{n, t,\epsilon}(\bx|\bW(t),\tilde{\bW}(t,\epsilon))=\frac{1}{\mathcal{Z}_{n, t,\epsilon}(\bW(t),\tilde{\bW}(t,\epsilon))}\Big(\prod_{i=1}^n dP_{X,n}(x_i)\Big) \nn
& \times\!\exp\!\Big\{\sum_{i< j}^n \Big((1-t)\frac{\lambda_n}{n}\frac{x_i^2x_j^2}{2}-\sqrt{(1-t)\frac{\lambda_n}{n}}x_ix_jW_{ij}(t) \Big)+  R_n(t,\epsilon)\frac{\|\bx\|^2}{2} - \sqrt{R_n(t,\epsilon)} \bx\cdot \tilde \bW(t,\epsilon)\Big\}.
\end{align*}
The normalization factor $\mathcal{Z}_{n,t,\epsilon}(\dots)$ is also called partition function. We also define the mutual information density for the interpolating model
\begin{align}
i_{n}(t,\epsilon)&\equiv\frac{1}{n}I\big(\bX;(\bW(t),\tilde{\bW}(t,\epsilon))\big)\, \label{fnt}.
\end{align}
The $(n,t,\epsilon, R_n)$-dependent Gibbs-bracket (that we simply denote $\langle - \rangle_t$ for the sake of readability) is defined for functions $A(\bx)=A$
\begin{align}
\langle A(\bx) \rangle_{t}= \int dP_{n, t,\epsilon}(\bx| \bW(t),\tilde{\bW}(t,\epsilon))\,A(\bx) \,. \label{t_post}
\end{align}
%
%By design of the interpolating model we have
%
\begin{lemma}[Boundary values]\label{lemma:bound1}
The mutual information for the interpolating model verifies 
\begin{align}\label{bound2}
\begin{cases}
i_{n}(0,\epsilon) =\frac1n I(\bX;\bW)+O(\rho_n s_n)\,,\\
i_{n}(1,\epsilon) =I_n(X;\{\lambda_n\int_0^1dt\,q_n(t,\epsilon)\}^{1/2}X+Z)+O(\rho_ns_n)\,. 
\end{cases} 
\end{align}
where $I_n(X;\{\lambda_n\int_0^1dt\,q_n(t,\epsilon)\}^{1/2}X+Z)$ is the mutual information for a scalar gaussian channel with input $X\sim P_{X,n}$ and noise $Z\sim{\cal N}(0,1)$.
\end{lemma}
\begin{proof}
We start with the chain rule for mutual information:
$$
i_{n}(0,\epsilon)= \frac1nI(\bX;{\bW}(0))+\frac1nI(\bX;\tilde{\bW}(0,\epsilon)|{\bW}(0)).
$$
Note that $I(\bX;{\bW}(0))=I(\bX;\bW)$ which is obvious. Moreover we claim
$
\frac1nI(\bX;\tilde{\bW}(0,\epsilon)|{\bW}(0))=O(\rho_ns_n)	
$
which yields the first identity in \eqref{bound2}.
This claim simply follows from the I-MMSE relation (appendix \ref{app:gaussianchannels}) and $R_n(0,\epsilon)=\epsilon$
\begin{align}
\frac{d}{d\epsilon}\frac1nI(\bX;\tilde{\bW}(0,\epsilon)|{\bW}(0))=\frac1{2n} {\rm MMSE}(\bX|\tilde{\bW}(0,\epsilon),{\bW}(0)) \le \frac{\rho_n}{2}\,. \label{id9}
\end{align}
The last inequality above is true because 
${\rm MMSE}(\bX|\tilde{\bW}(0,\epsilon),{\bW}(0))\le \EE\|\bX-\EE\,\bX\|^2=n{\rm Var}(X_1)\le n\rho_n$, 
as the components of $\bX$ are i.i.d. from $P_{X,n}$. Therefore $\frac1nI(\bX;\tilde{\bW}(0,\epsilon)|{\bW}(0))$ is $\frac{\rho_n}{2}$-Lipschitz in $\epsilon\in[s_n,2s_n]$. Moreover we have that $I(\bX;\tilde{\bW}(0,0)|{\bW}(0))=0$. This implies the claim. 

The proof of the second identity in \eqref{bound2} again starts from the chain rule for mutual information
$$
i_{n}(1,\epsilon)=\frac1nI(\bX;\tilde\bW(1,\epsilon))+\frac1nI(\bX;{\bW}(1)|\tilde\bW(1,\epsilon))\, .
$$
Note that $I(\bX;{\bW}(1)|\tilde\bW(1,\epsilon))=0$ as ${\bW}(1)$ does not depend on $\bX$. Moreover, 
\begin{align*}
\frac1nI(\bX;\tilde\bW(1,\epsilon))=I_n(X;\sqrt{R_n(1,\epsilon)}X+Z)=\textstyle{I_n(X;\{\lambda_n\int_0^1dt\,q_n(t,\epsilon)\}^{1/2}X+Z)+O(\rho_ns_n)}\,.
\end{align*}
because $I_n(X;\sqrt{\gamma}X+Z)$ is a $\frac{\rho_n}{2}$-Lipschitz function of $\gamma$, 
by an application of the I-MMSE relation (appendix \ref{app:gaussianchannels}) $\frac{d}{d\gamma}I_n(X;\sqrt{\gamma}X+Z)={\rm MMSE}(X|\sqrt{\gamma}X+Z)/2\le {\rm Var}(X)/2\le \rho_n/2$.
\end{proof}
\subsection{Fundamental sum rule.}
%A crucial identity is:
\begin{proposition}[Sum rule]\label{prop1}
The mutual information verifies the following sum rule:
\begin{align}
\frac1n I(\bX;\bW) &=   i_n^{\rm pot}\big({\textstyle \int_0^1dt\,q_n(t,\epsilon)};\lambda_n,\rho_n\big) +\frac{\lambda_n}{4}\big({\cal R}_1 - {\cal R}_2-{\cal R}_3\big)+ O(\rho_n s_n)+O\Big(\frac{\lambda_n}{n}\Big)	
\label{MF-sumrule}
\end{align}
with non-negative ``remainders'' that depend on $(n,\epsilon, R_{n})$
\begin{align}\label{contributions}
	\begin{cases}
{\cal R}_1\equiv\int_0^1 dt \,\big(q_n(t,\epsilon) - \int_0^1 ds \, q_n(s,\epsilon)\big)^2\,,\\ 
{\cal R}_2\equiv \int_0^1	dt\,\E\big\langle\big(Q-\E\langle Q\rangle_t\big)^2\big\rangle_{t}\,,\\
{\cal R}_3\equiv \int_0^1	dt\,\big(q_n(t,\epsilon)-\E\langle Q\rangle_t\big)^2\,.
\end{cases}
\end{align}
where $Q = \frac{1}{n}\bx \cdot\bX$ is called the {\it overlap}. The constants in the $O(\cdots)$ terms are independent of $n,t,\epsilon$.	
\end{proposition}
\begin{proof}
By the fundamental theorem of calculus 
$i_{n}(0,\epsilon)=i_{n}(1,\epsilon)-\int_0^1dt \frac{d}{dt}i_{n}(t,\epsilon)$. 
Note that $i_{n}(0,\epsilon)$ and $i_{n}(1,\epsilon)$ are given by \eqref{bound2}.
The $t$-derivative of the interpolating mutual information is simply computed combining the I-MMSE relation with the chain rule for derivatives
\begin{align}
\frac{d}{dt}i_n(t,\epsilon)&= -\frac{\lambda_n}{2}\frac{1}{n^2}\sum_{i<j}\EE\big[(X_iX_j-\langle x_ix_j\rangle_{t})^2\big]+\frac{\lambda_nq_n(t,\epsilon)}{2}\frac1n\EE\|\bX-\langle \bx\rangle_{t}\|^2\label{9}\\
&=-\frac{\lambda_n}{4}\frac{1}{n^2}\EE\|\bX\otimes \bX-\langle \bx\otimes \bx\rangle_{t}\|_{\rm F}^2+\frac{\lambda_nq_n(t,\epsilon)}{2}\frac1n\EE\|\bX-\langle \bx\rangle_{t}\|^2+ O\Big(\frac{\lambda_n}{n}\Big)\label{34}\,.
\end{align}
The correction term in \eqref{34} comes from completing the diagonal terms 
in the sum $\sum_{i< j}$ in order to construct the matrix-MMSE, 
namely the first term on the r.h.s. of \eqref{34}. This expression can be simplified by application of the Nishimori identities 
(appendix \ref{app:nishimori} contains a proof of these general identities). Starting with the second term (a vector-MMSE)
\begin{align}
\frac1n\E\|\bX-\langle \bx \rangle_t\|^2 & =\E\big[\|\bX\|^2+\|\langle\bx \rangle_t\|^2- 2\bX\cdot \langle \bx \rangle_t\big]\nn
& =  \frac1n\E\big[\|\bX\|^2 -\bX\cdot \langle \bx \rangle_t\big] =\rho_n -\EE\langle Q\rangle_t\,, \label{mmse}
\end{align}
were we used $\E\|\bX\|^2=n\rho_n$ and the Nishimori identity $\EE\|\langle \bx \rangle_t\|^2 = \EE[\bX\cdot\langle \bx \rangle_t]$.
By similar manipulations we obtain for the matrix-MMSE
\begin{align}\label{manip}
\frac1{n^2}{\rm MMSE}(\bX\otimes \bX|\tilde{\bW}(t,\epsilon),{\bW}(t))=\frac1{n^2}\EE\|\bX\otimes \bX-\langle \bx\otimes \bx\rangle_{t}\|_{\rm F}^2 =\rho_n^2-\EE\langle Q^2\rangle_t\,.
\end{align}
From \eqref{bound2}, \eqref{34}, \eqref{mmse}, \eqref{manip} and the fundamental theorem of calculus we deduce
\begin{align*}
\frac1n I(\bX;\bW) = &   {\textstyle{I_n\big(X;\{\lambda_n\int_0^1dt\,q_n(t,\epsilon)\}^{1/2}X+Z\big)}} 
 \nn & 
 + \frac{\lambda_n}{4}\int_0^1 dt\,\Big\{\rho_n^2-\EE\langle Q^2\rangle_t-2q_n(t,\epsilon)(\rho_n-\EE\langle Q\rangle_t)\Big\}+ O(\rho_n s_n)+O\Big(\frac{\lambda_n}{n}\Big)\,.
\end{align*}
The terms on the r.h.s can be re-arranged so that the potential \eqref{26} appears, and this gives immediately the sum rule \eqref{MF-sumrule}.
\end{proof}

%
%\begin{align*}
%\frac1n I(\bX;\bW) = &  i_n^{\rm pot}\big({\textstyle \int_0^1dt\,q_n(t,\epsilon)};\lambda_n,\rho_n\big) - \frac{\lambda_n}{4}\big\{{\textstyle \int_0^1dt\,q_n(t,\epsilon)}-\rho_n\big\}^2 
%\nn
%&+ \frac{\lambda_n}{4}\int_0^1 dt\,\Big\{(\rho_n^2-\EE\langle Q^2\rangle_t)-2q_n(t,\epsilon)(\rho_n-\EE\langle Q\rangle_t)\Big\}+ O\Big(\rho_n s_n+\frac{\lambda_n}{n}\Big)
%%
%% &=   i_n^{\rm pot}\big({\textstyle \int_0^1dt\,q_n(t,\epsilon)};\lambda_n,\rho_n\big) - \frac{\lambda_n}{4}\big\{{\textstyle \int_0^1dt\,q_n(t,\epsilon)}\big\}^2\nn
%% &\qquad- \frac{\lambda_n}{4}\int_0^1 dt\,\Big\{\EE\langle Q^2\rangle_t-2q_n(t,\epsilon)\EE\langle Q\rangle_t\Big\}+ O\Big(\rho_n s_n+\frac{\lambda_n}{n}\Big)\,,
%\end{align*}
%%
%which, after a line of basic algebra, finally simplifies to the claimed sum rule.	

Theorem~\ref{thm:ws} follows from the upper and lower bounds proven below, and applied for $s_n=\frac{1}{2}n^{-\alpha}$.

\subsection{Upper bound: linear interpolation path.}
 %We start with the simplest bound:
\begin{proposition}[Upper bound]
We have
\begin{align*}
\frac1n I(\bX;\bW) \le \inf_{q\in [0,\rho_n]} i_n^{\rm pot}(q,\lambda_n,\rho_n)+ O(\rho_n s_n)+O\Big(\frac{\lambda_n}{n}\Big)\,.
\end{align*}
\end{proposition}
\begin{proof}
Fix $q_{n}(t,\epsilon)= q_n \in [0, \rho_n]$ a constant independent of $\epsilon, t$. The interpolation path $R_n(t,\epsilon)$  is therefore 
a simple linear function of time.
From \eqref{contributions} ${\cal R}_1$ cancels and since ${\cal R}_2$ and ${\cal R}_3$ are non-negative we get from Proposition
\eqref{prop1}
\begin{align*}
\frac{1}{n}I(\bX;\bW) \leq i_n^{\rm pot}(q,\lambda_n,\rho_n)+ O(\rho_n s_n)+O\Big(\frac{\lambda_n}{n}\Big)\,.
\end{align*}
Note that the error terms $O(\cdots)$ are bounded independently of $q_n$.
Therefore optimizing the r.h.s over the free parameter $q_n\in [0, \rho_n]$ yields the upper bound.
\end{proof}

\subsection{Lower bound: adaptive interpolation path.} 
We start with a definition: the map $\epsilon\mapsto R_n(t,\epsilon)$ is called \emph{regular} if it is a ${\cal C}^1$-diffeomorphism whose jacobian is greater or equal to one for all $t\in [0,1]$. 
\begin{proposition}[Lower bound]
Consider sequences $\lambda_n$ and $\rho_n$ satisfying $c_1\leq \lambda_n \rho_n \leq c_2n^\gamma$ for some constants positive constant $c_1, c_2$ and $\gamma\in [0, 1/2[$.
Then
\begin{align}\label{MF-sumrule_bound}
 \frac1n I(\bX;\bW) \ge\inf_{q\in[0,\rho_n]}i_n^{\rm pot}(q,\lambda_n,\rho_n) + O(\rho_n s_n)+O\Big(\frac{\lambda_n}{n}\Big)+O\Big(\Big(\frac{\lambda_n^4\rho_n}{ns_n^4}\Big)^{1/3}\Big)\,.
\end{align}
\end{proposition}
\begin{proof}
First note that the regime \eqref{mainregime} for the sequences $\lambda_n, \rho_n$ satisfies  the more general condition assumed in this lemma (this is the condition in theorem \ref{thm:wsgeneral} of appendix \ref{app:matrix}).
Assume for the moment that the map $\epsilon\mapsto R_n(t,\epsilon)$ is regular. Then, based on Proposition \ref{L-concentration} and identity \eqref{remarkable} (appendix \ref{appendix-overlap}), we have a bound on the overlap fluctuation. Namely, for some numerical constant $C\ge 0$ independent of $n$
\begin{align}\label{over-concen}
\frac{\lambda_n}{s_n}\int_{s_n}^{2s_n}d\epsilon\, {\cal R}_2 = \frac{\lambda_n}{s_n}\int_{s_n}^{2s_n}d\epsilon \int_0^1dt\,  \mathbb{E}\big\langle (Q - \mathbb{E}\langle Q\rangle_{n,t,R_n(t,\epsilon)} )^2\big\rangle_{n,t,R_n(t,\epsilon)}\le C\Big(\frac{\lambda_n^4\rho_n}{ns_n^4}\Big)^{1/3} \,.
\end{align}
Using this concentration result, and ${\cal R}_1 \geq 0$, and averaging the sum rule \eqref{MF-sumrule} over $\epsilon\in [s_n, 2s_n]$ (recall the error terms are independent of $\epsilon$) we find
\begin{align}\label{MF-sumrule-simple}
\frac1n I(\bX;\bW) \ge & \frac1{s_n}\int_{s_n}^{2s_n}d\epsilon  i_n^{\rm pot}\big({\textstyle \int_0^1dt\,q_n(t,\epsilon)},\lambda_n,\rho_n\big) -  \frac{\lambda_n}{4}\frac1{s_n}\int_{s_n}^{2s_n}d\epsilon\int_0^1dt\, \big(q_n(t,\epsilon)-\EE\langle Q\rangle_{t}\big)^2 \nonumber\\
& +  O(\rho_n s_n)+O\Big(\frac{\lambda_n}{n}\Big)+O\Big(\Big(\frac{\lambda_n^4\rho_n}{ns_n^4}\Big)^{1/3}\Big)\,.
\end{align}
At this stage it is natural to see if we can choose $q_n(t,\epsilon)$ to be the solution of
$q_n(t,\epsilon)=\mathbb{E}\langle Q\rangle_{t}$.
Setting $F_n(t,R_n(t,\epsilon))\equiv \mathbb{E}\langle Q\rangle_{n,t,R_n(t,\epsilon)}$, we recognize 
a {\it first order ordinary differential equation}
\begin{align}\label{odexx}
\frac{d }{dt}R_n(t,\epsilon) = F_n(t,R_n(t,\epsilon))\quad\text{with initial condition}\quad R_n(0,\epsilon)=\epsilon\,.
\end{align}
As $F_n(t,R_n(t,\epsilon))$ is ${\cal C}^1$ with bounded derivative w.r.t. its second argument the Cauchy-Lipschitz theorem implies that \eqref{odexx} admits a unique global solution $R_{n}^*(t,\epsilon)=\epsilon+\int_0^t ds\,q_{n}^*(s,\epsilon)$, where $q_{n}^*:[0,1]\times [s_n,2s_n]\mapsto [0,\rho_n]$.
Note that any solution must satisfy $q_{n}^*(t,\epsilon)\in [0,\rho_n]$ because $\mathbb{E}\langle Q\rangle_{n,t,\epsilon}\in [0,\rho_n]$ as can be seen from a Nishimori identity (appendix \ref{app:nishimori}) and \eqref{mmse}. 

We check that $R_n^*$ is regular. By Liouville's formula the jacobian of the flow $\epsilon\mapsto R_{n}^*(t,\epsilon)$ satisfies 
$$
\frac{d}{d\epsilon}  R_{n}^*(t,\epsilon)=\exp\Big\{\int_0^t ds\, \frac{d}{dR} F_n(s,R)\Big|_{R=R_{n}^*(s,\epsilon)}\Big\}\,.
$$ 
Applying repeatedly the Nishimori identity of Lemma \ref{NishId} (appendix \ref{app:nishimori}) one obtains (this computation does not present any difficulty and can be found in section 6 of \cite{BarbierM17a})
\begin{align}
\frac{d}{dR} F_n(s,R)=\frac{1}{n}\sum_{i,j=1}^n\mathbb{E}\big[(\langle x_ix_j\rangle_{n,s,R}-\langle x_i\rangle_{n,s,R}\langle x_j\rangle_{n,s,R})^2\big]\ge  0\label{jacPos}
\end{align}
so that the flow has a jacobian greater or equal to one. In particular it is locally invertible (surjective). Moreover it is injective because of the unicity of the solution of the differential equation, and therefore it is a $C^1$-diffeomorphism. Thus $\epsilon\mapsto R_{n}^*(t,\epsilon)$ is regular.  With the choice $R_{n}^*$, i.e., by suitably \emph{adapting} the interpolation path, we cancel ${\cal R}_3$. This yields
\begin{align*}
 \frac1n I(\bX;\bW) &\ge \frac1{s_n}\int_{s_n}^{2s_n}d\epsilon \, i_n^{\rm pot}\big({\textstyle \int_0^1dt\,q_n^*(t,\epsilon)},\lambda_n,\rho_n\big) +O(\cdots)
 \nn &
 \ge\inf_{q\in[0,\rho_n]}i_n^{\rm pot}(q,\lambda_n,\rho_n) + O(\cdots)
\end{align*}
where the $O(\cdots)$ is a shorthand notation for the three error terms in \eqref{MF-sumrule-simple}. This the desired result.	
\end{proof}

\section*{Appendices}
\section{General results on the mutual information of sparse spiked matrix models}\label{app:matrix}
% In this paper expectations $\EE$ apply to everything that they face, and we use brackets only when necessary. E.g., $\EE \,(A+B)^2 = \EE[(A+B)^2]\neq \EE[A+B]^2$, similarly $\EE\|A\|^2=\EE[\|A\|^2]\neq \EE[\|A\|]^2$, or $\EE\, A B= \EE[AB]$ and so on. 

In this appendix we give a more general form of theorems \ref{thm:ws} and \ref{thm:wishart} in section \ref{sec:matrix}. 

\subsection{Spiked Wigner model}

Our analysis by the adaptive interpolation method works for any regime where the sequences $\lambda_n$ and $\rho_n$ verify: 
\begin{align}
C\le \lambda_n\rho_n = O(n^{\gamma})\, \quad \text{for some constants} \quad \gamma \in [0,1/2) \quad \text{and} \quad C>0\,.
\label{app-scalingRegime}
\end{align}
Of course this contains the regime \eqref{mainregime} as a special case. Our general result is a statement on the smallness of 
\begin{align*}
\Delta I_n^{\rm Wig} \equiv \frac{1}{\rho_n|\ln\rho_n|}\Big|\frac{1}{n}I(\bX;\bW) -\inf_{q\in [0,\rho_n]} i^{\rm pot}_n(q,\lambda_n,\rho_n)\Big|\,.
\end{align*}
The analysis of section \ref{sec:adapInterp_XX} leads to the following general theorem.

% The total signal-to-noise ratio (SNR) per parameter is: $\# \,{\rm observations} \cdot {\rm SNR_{obs} }/\# \,{\rm parameters\ to\ infer}$, where $\rm SNR_{obs}$ denotes the SNR per observation. In the present case 
% we have access to $n(n-1)/2$ independent observations and ${\rm SNR_{obs}}=\mathbb{E}[(X_1X_2)^2]\lambda_n/n = (\rho_n\tau)^2\lambda_n/n$. Therefore the total SNR is $$\frac{n(n-1)/2\cdot(\rho_n\tau)^2\lambda_n/n}{n} = \frac{(\rho_n\tau)^2\lambda_n}{2}+O(\rho_n^2\lambda_n/n)\,.$$ 
%
% The posterior, or Gibbs distribution, is 
% %
% \begin{align*}
% dP(\bx | \bW(\bX,\bZ) ) = \frac{1}{{\cal Z}_n({\bm X}, {\bm Z})}\prod_{i=1}^n dP_{X,n}(x_i)\,e^{-{\cal H}({\bm x} ; {\bm X}, {\bm Z} )}\,,
% \end{align*}
% %
% where the Hamiltonian and partition function are
% \begin{align*}
% {\cal H}({\bm x} ; {\bm X}, {\bm Z} )&\equiv\lambda_n\sum_{i< j}^n \Big(\frac{x_i^2x_j^2}{2n}-\frac{x_ix_jX_iX_j}{n}-\frac{x_ix_jZ_{ij}}{\sqrt{n\lambda_n}}\Big)\,,\\
% {\cal Z}_n(\bX, \bZ) &\equiv \int \prod_{i=1}^n dP_{X,n}(x_i) \,e^{-{\cal H}(\bx ; \bX,\textbf{Z})}\,.
% \end{align*}

\begin{thm}[Sparse spiked Wigner model]\label{thm:wsgeneral}
Let the sequences $\lambda_n$ and $\rho_n$ verify \eqref{app-scalingRegime} and let $\alpha>0$. 
There exists a constant $C>0$ independent of $n$, such that the mutual information for the Wigner spike model verifies
\begin{align*}
\Delta I_n^{\rm Wig}\le \frac{C}{|\ln\rho_n|}\max\Big\{\frac{1}{n^\alpha},\frac{\lambda_n}{n\rho_n},\Big(\frac{\lambda_n^4}{n^{1-4\alpha}\rho_n^2}\big(1+\lambda_n\rho_n^2\big)\Big)^{1/3}\Big\}\,.
\end{align*}	
In particular, choosing $\lambda_n= \Theta(|\ln\rho_n|/\rho_n)$ (which is the appropriate scaling to observe a phase transition) 
\begin{align*}
\Delta I_n^{\rm Wig} \le C\max\Big\{\frac{1}{n^{\alpha}|\ln\rho_n|},\frac{1}{n\rho_n^2},\Big(\frac{|\ln \rho_n|}{n^{1-4\alpha}\rho_n^6}\Big)^{1/3}\Big\}\,.
\end{align*}
If in addition we set $\rho_n=\Theta(n^{-\beta})$, $\beta\ge 0$ (which is the regime \eqref{mainregime}) we have
\begin{align*}
\Delta I_n^{\rm Wig} \le C\max\Big\{\frac{1}{n^{\alpha}\ln n},\frac{1}{n^{1-2\beta}},\Big(\frac{\ln n}{n^{1-4\alpha-6\beta}}\Big)^{1/3}\Big\}\,.
\end{align*}
This bound vanishes as $n$ grows if $\beta\in[0,1/6)$ and $\alpha\in(0, (1-6\beta)/4]$. The last bound is optimized (up to polylog factors) setting $\alpha=(1-6\beta)/7$. In this case (again, when $\lambda_n= \Theta(|\ln\rho_n|/\rho_n)$ and $\rho_n=\Theta(n^{-\beta})$)
\begin{align*}
\Delta I_n^{\rm Wig} \le C\frac{(\ln n)^{1/3}}{n^{(1-6\beta)/7}}\,.
\end{align*}
\end{thm}

\subsection{Spiked Wishart model}

The following regime is of particular interest and is the one mostly studied in the literature given in the introduction on spiked covariance models:
\begin{align}\label{sparsePCAscaling}
\alpha_n \to \alpha > 0\,, \quad \rho_{U,n}\to \rho_U>0\,, \quad \omega(1/n)=\rho_{V,n}\to 0_+,  \quad \lambda_n = \Theta\Big(\sqrt{\frac{|\ln \rho_{V,n}|}{\rho_{V,n}}}\Big)\,.
\end{align}
The notation $\omega(1/n)=\rho_{V,n}$ means that the sequence $\rho_{V,n}$ vanishes at a rate slower than $1/n$. The analysis of appendix
\ref{app_wishart} leads to the following general theorem on the smallness of 
\begin{align*}
\Delta I_n^{\rm Wish}\equiv\frac{1}{\sqrt{\rho_{V,n}\vert\ln\rho_{V,n}\vert}} \Big\vert \frac{1}{n} I\big((\bU,\bV);\bW\big) -
 \ {\adjustlimits \inf_{q_U\in[0,\rho_{U,n}]}\sup_{q_{V}\in[0,\rho_{V,n}]}}\,i_n^{\rm pot}\big(q_{U},q_V,\lambda_n,\alpha_n,\rho_{U,n},\rho_{V,n}\big) \Big\vert\,.	
\end{align*}
%Our main result reads:
\begin{thm}[Sparse spiked Wishart model]\label{thm:wishart} 
Under the scalings \eqref{sparsePCAscaling}, there exists a constant $C>0$ independent of $n$ such that the mutual information for the spiked Wishart model verifies for any $\alpha >0$
\begin{align*} 
\Delta_nI^{\rm Wish}\leq C \max\Big\{ \frac{1}{n^\alpha \sqrt{\rho_{V,n}\vert \ln\rho_{V,n}\vert}} , \frac{|\ln \rho_{V,n}|^{-1/24}}{n^{1/3-2\alpha}\rho_{V,n}^{11/12}}\Big(\frac{1}{n^{\alpha}}+\sqrt{\rho_{V,n}|\ln \rho_{V,n}|}\Big)^{1/3} \Big\}\,.	
\end{align*}
We set $\rho_{V,n}=\Theta(n^{-\beta})$. Optimizing over $\alpha$ (up to polylog factors) such that $n^{-\alpha}< \sqrt{\rho_{V,n}|\ln\rho_{V,n}|}$ yields $\alpha=(4-3\beta)/18$. In this case the bound simplifies to
\begin{align*} 
\Delta_nI^{\rm Wish}\leq  C\frac{(\ln n)^{1/3}}{n^{(4-12\beta)/18}}
\end{align*}
for some $C>0$. This bound vanishes if $\beta\in[0,1/3)$.
\end{thm}
\section{Proof of theorem \ref{thm:wishart} by the adaptive interpolation method}\label{app_wishart}

In this appendix we prove theorem \ref{thm:wishart} by the adapative interpolation method. The analysis is similar to the one of the Wigner case in section \ref{sec:adapInterp_XX}.

\subsection{The interpolating model.}
Let ${\bm \epsilon}=(\epsilon_U,\epsilon_V) \in[s_n, 2s_n]^2$ for some sequence $s_n=\frac{1}{2}n^{-\alpha}$. Let $q_{U,n}: [0, 1]\times [s_n,2s_n] \mapsto [0,\rho_{U,n}]$ and similarly for $q_{V,n}$. Set 
\begin{align*}
\begin{cases}
	R_{U,n}(t,{\bm \epsilon})\equiv \epsilon_{U}+\lambda_n\int_0^tds\,q_{U,n}(s,{\bm \epsilon})\,,\\
R_{V,n}(t,{\bm \epsilon})\equiv \epsilon_{V}+\lambda_n\alpha_n\int_0^tds\,q_{V,n}(s,{\bm \epsilon})\,.
\end{cases}	
\end{align*}
Consider the following interpolating estimation model, where $t\in[0,1]$, with accessible data
\begin{align}\label{interpUV}
\begin{cases}
\bW(t)&=\sqrt{(1-t)\frac{\lambda_n}{n}}\, \bU\otimes \bV + \bZ\,, \\
\tilde \bW_U(t,{\bm \epsilon})  &= \sqrt{R_{V,n}(t,{\bm \epsilon})}\,\bU + \tilde \bZ_U\,,\\
\tilde \bW_V(t,{\bm \epsilon})  &= \sqrt{R_{U,n}(t,{\bm \epsilon})}\,\bV + \tilde \bZ_V\,,
\end{cases} 
\end{align}
with independent standard gaussian noise $\tilde \bZ_U,\tilde \bZ_V \sim {\cal N}(0,{\rm I}_n)$, $\bZ=(Z_{ij})_{ij}$ with i.i.d. $Z_{ij}\sim {\cal N}(0,1)$. The fact that the $R_{V,n}$ function appears as the SNR of the decoupled gaussian channel related to $\bU$ (and vice-versa) comes from the bipartite nature of the problem. The Gibbs-bracket, simply denoted $\langle - \rangle_t$, is the expectation w.r.t. the posterior distribution, which is proportional to
(here $\|-\|_{\rm F}$ and $\Vert-\Vert$ are the Frobenius and $\ell_2$ norms)
\begin{align*}
&dP_{n, t,{\bm \epsilon}}\big(\bu,\bv|\bW(t), \tilde \bW_U(t,{\bm \epsilon}),\tilde \bW_V(t,{\bm \epsilon})\big)\nn
&\qquad\qquad\propto \Big(\prod_{i=1}^n dP_{U,n}(u_i)\Big) \Big(\prod_{j=1}^{m}dP_{V,n}(v_j)\Big) \exp \Big\{-\frac12\big\|\bW(t)-\sqrt{(1-t)\frac{\lambda_n}{n}} \bu \otimes \bv\big\|_{\rm F}^2 \nn
&\qquad\qquad\qquad\qquad\qquad-\frac12 \big\| \tilde\bW_U(t,{\bm \epsilon})-\sqrt{R_{V,n}(t,{\bm \epsilon})}\,\bu\big\|^2-\frac12 \big\| \tilde \bW_V(t,{\bm \epsilon})-\sqrt{R_{U,n}(t,{\bm \epsilon})}\,\bv\big\|^2 \Big\}\,.	
\end{align*}
The mutual information density for this interpolating model is
\begin{align*}
i_{n}(t,{\bm \epsilon})\equiv\frac{1}{n}I\big((\bU,\bV);(\bW(t),\tilde{\bW}_U(t,{\bm \epsilon}),\tilde{\bW}_V(t,{\bm \epsilon}))\big)\,.
\end{align*}
The proof of the following lemma is similar to the one of Lemma~\ref{lemma:bound1}.
\begin{lemma}[Boundary values]
Let $\bar \rho_n = \max(\rho_{U,n},\rho_{V,n})$, $U\sim P_{U,n}$, $V\sim P_{V,n}$ and $Z\sim{\cal N}(0,1)$. Then 
\begin{align*}
\begin{cases}
i_{n}(0,{\bm \epsilon}) =\frac1n I((\bU,\bV);\bW)+O(\bar \rho_{n} s_n)\,,\\
i_{n}(1,{\bm \epsilon}) =I_{n}(U;\{\lambda_n \alpha_n \int_0^1q_{V,n}(s,{\bm \epsilon})\}^{1/2}U\!+\!Z)\!+\!\alpha_nI_{n}(V;\{\lambda_n \int_0^1q_{U,n}(s,{\bm \epsilon})\}^{1/2}V\!+\!Z)\!+\!O(\bar \rho_ns_n)\,.
\end{cases} 
\end{align*}
\end{lemma}
\subsection{Fundamental sum-rule.}
As before, our proof is bases on an important sum-rule.
\begin{proposition}[Sum rule]\label{prop:sumsule-wishart}
Let $\bar \rho_n = \max(\rho_{U,n},\rho_{V,n})$, $\alpha_n = m/n$, $U\sim P_{U,n}$, $V\sim P_{V,n}$ and $Z\sim{\cal N}(0,1)$. Then
\begin{align*}
&\frac1n I((\bU,\bV);\bW) =   I_{n}(U;\{\lambda_n \alpha_n \textstyle{\int_0^1q_{V,n}}(s,{\bm \epsilon})\}^{1/2}U+Z)+\alpha_nI_{n}(V;\{\lambda_n \textstyle{\int_0^1q_{U,n}}(s,{\bm \epsilon})\}^{1/2}V+Z)\nn
&\qquad+\frac{\lambda_n\alpha_n}{2}\rho_{U,n}\rho_{V,n}+\frac{\lambda_n\alpha_n}{2}\int_0^1dt\Big\{\EE\langle Q_U\rangle_t \EE\langle Q_V\rangle_t-\EE\langle Q_U Q_V\rangle_t\Big\} +O(\bar \rho_ns_n) \nn
&\qquad +\frac{\lambda_n\alpha_n}{2}\int_0^1dt\Big\{q_{V,n}(t,{\bm \epsilon})(\EE\langle Q_U\rangle_t-\rho_{U,n})+q_{U,n}(t,{\bm \epsilon})(\EE\langle Q_V\rangle_t-\rho_{V,n})-\EE\langle Q_U\rangle_t \EE\langle Q_V\rangle_t\Big\}
\end{align*}
where the overlaps are defined as 
$$
Q_U\equiv \frac{1}{n}\bu\cdot \bU\,, \qquad Q_V\equiv \frac{1}{m}\bv\cdot \bV \,.
$$
%
% with non-negative  $(n,R_{n},\epsilon)$-dependent ``remainders'' 
% %
% \begin{align}
% 	\begin{cases}
% {\cal R}_1\equiv??
% \end{cases}
% \end{align}	
\end{proposition}
\begin{proof}
We compare the boundaries values \eqref{bound2} using the fundamental theorem of calculus $i_{n}(0,{\bm \epsilon})=i_{n}(1,{\bm \epsilon})-\int_0^1 dt \frac{d}{dt}i_{n}(t,{\bm \epsilon})$. Using the I-MMSE formula (first equality) and then the Nishimori identity (second equality) we have
\begin{align*}
& \frac{d}{dt}i_n(t,{\bm \epsilon})
\nn & 
= -\frac{\lambda_n\alpha_n}{2nm}\EE\|\bU\otimes\bV-\langle \bu\otimes\bv\rangle_{t}\|_{\rm F}^2
\!+\!\frac{\lambda_n\alpha_nq_{V,n}(t,{\bm \epsilon})}{2n}\EE\|\bU-\langle \bu\rangle_{t}\|^2\!+\!\frac{\lambda_n\alpha_nq_{U,n}(t,{\bm \epsilon})}{2m}\EE\|\bV-\langle \bv\rangle_{t}\|^2
\nn &
\overset{\rm N}{=} \frac{\lambda_n\alpha_n}{2}\Big\{\EE\langle Q_UQ_V\rangle_t-\rho_{U,n}\rho_{V,n}+q_{V,n}(t,{\bm \epsilon})(\rho_{U,n}-\EE\langle Q_U\rangle_t)+q_{U,n}(t,{\bm \epsilon})(\rho_{V,n}-\EE\langle Q_V\rangle_t)\Big\}\,.
\end{align*}
The $\rm N$ stands for ``Nishimori'', and each time we use the Nishimori identity of Lemma~\ref{NishId} for a simplification we write a $\rm N$ on top of the equality. Replacing this result and the boundary values \eqref{bound2} in the fundamental theorem of calculus yields the sum rule after few lines of algebra.
\end{proof}

We now derive two matching bounds, under the scalings \eqref{sparsePCAscaling}, which implie Theorem~\ref{thm:wishart}.

\subsection{Upper bound: partially adaptive interpolation path.}
We start again with the simplest bound:
\begin{proposition}[Upper bound] Under the scalings \eqref{sparsePCAscaling} we have
\begin{align} 
\frac1n I((\bU,\bV);\bW) &\le  {\adjustlimits \inf_{q_U\in[0,\rho_{U,n}]}\sup_{q_{V}\in[0,\rho_{V,n}]}}\,i_n^{\rm pot}\big(q_{U},q_V,\lambda_n,\alpha_n,\rho_{U,n},\rho_{V,n}\big) \nn
&\qquad\qquad\qquad+O\Big(s_n+\frac{|\ln \rho_{V,n}|^{11/24}}{s_n^{2}n^{1/3}\rho_{V,n}^{9/24}}\Big(s_n+\sqrt{\rho_{V,n}|\ln \rho_{V,n}|}\Big)^{1/3}\Big)	\,.	\label{upperboundwishart}
\end{align}
\end{proposition}
\begin{proof}
For this bound only one of the interpolation function is adapted. Consider the following Cauchy problem for $R_n(t,{\bm \epsilon})=(R_{U,n}(t,{\bm \epsilon}),R_{V,n}(t,{\bm \epsilon}))$: 
\begin{align*}
\frac{dR_{n}}{dt}(t,{\bm \epsilon})=\big(\lambda_n q_{U}, G_{V,n}(t,R_n(t,{\bm \epsilon}))\big), \qquad R_n(0,{\bm \epsilon})={\bm \epsilon}\,,
\end{align*}
where $q_{U}\in[0,\rho_{U,n}]$ and $G_{V,n}(t,R_n(t,{\bm \epsilon}))\equiv\lambda_n\alpha_n \EE\langle Q_V\rangle_t\in[0,\lambda_n\alpha_n\rho_{V,n}]$, i.e.,
\begin{align*}
\big(q_{U,n}(t,{\bm \epsilon}),q_{V,n}(t,{\bm \epsilon})\big)=\big(q_{U}, \EE\langle Q_V\rangle_t\big), \qquad R_n(0,{\bm \epsilon})={\bm \epsilon}\,.
\end{align*}
By the Cauchy-Lipschitz theorem this ODE admits a unique global solution 
\begin{talign*}
R_{n}^*(t,{\bm \epsilon})=(R_{U,n}^*(t,{\bm \epsilon})=\epsilon_{U}+\lambda_n q_{U,n}t ,R_{V,n}^*(t,{\bm \epsilon})=\epsilon_{V}+\lambda_n\alpha_n\int_0^tds\,q^*_{V,n}(s,{\bm \epsilon}))	\,.
\end{talign*}
Because the function $(q_{U}, \EE\langle Q_V\rangle_t)$ is ${\cal C}^1$ the solution $R_n^*$ is ${\cal C}^1$ in all its arguments. By the Liouville formula the Jacobian determinant $J_{n}(t,{\bm \epsilon})$ of the flow ${\bm \epsilon}\mapsto R_{n}^*(t,{\bm \epsilon})$ satisfies
\begin{align}
J_{n}(t,{\bm \epsilon})\equiv {\rm det}\Big(\frac{\partial R_{n,t}(t,{\bm \epsilon})}{\partial {\bm \epsilon}}\Big)=\exp \Big\{\int_0^t \frac{\partial G_{V,n}}{\partial R_V}(s, R_{U,n}^*(s,{\bm \epsilon}),R_V=R_{V,n}^*(s,{\bm \epsilon}))ds\Big\}\ge 1\,.\label{JacUV}
\end{align}
We show at the end of the proof that $\frac{\partial G_{V,n}}{\partial R_V}\ge 0$. The intuition is the same as before: increasing the SNR $R_V$ cannot decrease the overlap $\EE\langle Q_V\rangle_t$, or equivalently it cannot increase the MMSE $\rho_{V,n}-\EE\langle Q_V\rangle_t$. The flow ${\bm \epsilon} \mapsto R_n^*(t,{\bm \epsilon})$ thus has Jacobian greater or equal to one, and is surjective. It is also injective by unicity of the solution of the differential equation, and is thus a ${\cal C}^1$-diffeomorphism. A $\mathcal{C}^1$-diffeomrophic flow with Jacobian greater or equal to one is called regular.

By the Cauchy-Schwarz inequality and Fubini's theorem we have
\begin{align*}
	&\frac{\lambda_n\alpha_n}{2s_n^2}\Big|\int d{\bm \epsilon} \int_0^1dt\Big\{\EE\langle Q_U\rangle_t \EE\langle Q_V\rangle_t-\EE\langle Q_U Q_V\rangle_t\Big\}\Big|\nn
	&\qquad=\frac{\lambda_n\alpha_n}{2s_n^2}\Big|\int_0^1dt\int d{\bm \epsilon}\, \EE\big\langle (Q_U-\EE\langle Q_U\rangle_t) (Q_V-\EE\langle Q_V\rangle_t)\big\rangle_t\Big|\nn
	&\qquad\qquad \le \frac{\lambda_n\alpha_n}{2s_n^2}\int_0^1dt\Big\{\int d{\bm \epsilon}\, \EE\big\langle (Q_U-\EE\langle Q_U\rangle_t)^2 \big\rangle_t \Big\}^{1/2} \Big\{\int d{\bm \epsilon}\, \EE\big\langle  (Q_V-\EE\langle Q_V\rangle_t)^2\big\rangle_t\Big\}^{1/2}\,.
\end{align*}
By the regularity of the flow we are allowed to use Propositions~\ref{L-concentration-wishart}, \ref{LU-concentration-wishart} of section~\ref{appendix-overlap}. Together with inequality \eqref{LV_QV} and a similar one for ${\cal L}_U$ (see section~\ref{appendix-overlap}) we obtain under the scalings \eqref{sparsePCAscaling},
\begin{align}
&\frac{\lambda_n\alpha_n}{2s_n^2}\Big|\int d{\bm \epsilon} \int_0^1dt\Big\{\EE\langle Q_U\rangle_t \EE\langle Q_V\rangle_t-\EE\langle Q_U Q_V\rangle_t\Big\}\Big|\nn
&\qquad\qquad \le \frac{C}{s_n^2}\sqrt{\frac{|\ln \rho_{V,n}|}{\rho_{V,n}}}\Big(\frac1n\sqrt{\frac{|\ln \rho_{V,n}|}{\rho_{V,n}}}\big(s_n+\sqrt{\rho_{V,n}|\ln \rho_{V,n}|}\big)^2\times\frac{(\ln \rho_{V,n})^2}{n\rho_{V,n}}\Big)^{1/6}\nn
&\qquad\qquad= C\frac{|\ln \rho_{V,n}|^{11/24}}{s_n^{2}n^{1/3}\rho_{V,n}^{9/24}}\Big(s_n+\sqrt{\rho_{V,n}|\ln \rho_{V,n}|}\Big)^{1/3}\,.\label{over-concen-wishart}
\end{align}
Therefore, averaging the sum-rule over ${\bm \epsilon}\in [s_n, 2s_n]^2$ and using the solution $R^*_n$ of the above Cauchy problem, we obtain 
\begin{align*}
\frac1n I((\bU,\bV);\bW) &=   \frac{1}{s_n^2}\int_{[s_n,2s_n]^2} d{\bm \epsilon}\,i_n^{\rm pot}\big(q_{U},{\textstyle \int_0^1 }q_{V,n}^*(t,{\bm \epsilon})dt,\lambda_n,\alpha_n,\rho_{U,n},\rho_{V,n}\big) \nn
&\qquad\qquad\qquad+O\Big(s_n+\frac{|\ln \rho_{V,n}|^{11/24}}{s_n^{2}n^{1/3}\rho_{V,n}^{9/24}}\Big(s_n+\sqrt{\rho_{V,n}|\ln \rho_{V,n}|}\Big)^{1/3}\Big)\,.
\end{align*}
Because this inequality is true for any $q_{U}\in[0,\rho_{U,n}]$ we obtain the result.

It remains to prove that $\frac{\partial G_{V,n}}{\partial R_V}\ge 0$, i.e., $\frac{\partial\EE\langle Q_V\rangle_t}{\partial R_V}\ge 0$. We drop un-necessary dependencies. Let $\Delta \ge 0$. Consider the following modification of the model \eqref{interpUV}:
\begin{align*}
\begin{cases}
\bW&=\sqrt{(1-t)\frac{\lambda_n}{n}}\, \bU\otimes \bV + \bZ\,, \\
\hat \bW_U(R_{V},\Delta)  &= \bU + \tilde \bZ_U/\sqrt{R_{V}}+\hat \bZ_U\sqrt{\Delta/R_{V}}\,,\\
\tilde \bW_V  &= \sqrt{R_{U}}\,\bV + \tilde \bZ_V\,,
\end{cases} 
\end{align*}
where $\hat \bZ_U\sim{\cal N}(0,{\rm I}_n)$ independently of the rest. By stability of the gaussian distribution under addition we have in law $\tilde \bZ_U/\sqrt{R_{V}}+\hat \bZ_U\sqrt{\Delta/R_{V}}=\tilde \bZ_U\sqrt{(\Delta+1)/R_{V}}$, therefore the MMSE for model \eqref{interpUV} ${\rm MMSE}(\bV|\bW,\tilde \bW_U(R_{V}),\tilde \bW_V)$ (we made explicit the dependence of $\tilde \bW_U$ in $R_V=R_{V,n}(t,{\bm \epsilon})$) verifies
$${\rm MMSE}(\bV|\bW,\hat \bW_U(R_V,\Delta),\tilde \bW_V)={\rm MMSE}(\bV|\bW,\tilde \bW_U(R_{V}/(\Delta+1)),\tilde \bW_V)\,.$$
We then have
\begin{align*}
	&{\rm MMSE}(\bV|\bW,\tilde \bW_U(R_{V}),\tilde \bW_V)={\rm MMSE}(\bV|\bW,\hat \bW_U(R_V,\Delta),\tilde \bW_V,\hat \bZ_U)\nn
	&\qquad\qquad \le {\rm MMSE}(\bV|\bW,\hat \bW_U(R_V,\Delta),\tilde \bW_V)={\rm MMSE}(\bV|\bW,\tilde \bW_U(R_{V}/(\Delta+1)),\tilde \bW_V)
\end{align*}
where the inequality follows from Lemma~\ref{lemma:conditioningMMSE}. Because $\frac{R_V}{\Delta+1}\le R_V$, ${\rm MMSE}(\bV|\bW,\tilde \bW_U(R_{V}),\tilde \bW_V)$ is non-increasing in $R_V$. Recalling $${\rm MMSE}(\bV|\bW,\tilde \bW_U(R_{V}),\tilde \bW_V)\overset{\rm N}{=}\rho_{V,n}-\EE\langle Q_V\rangle_t$$ this proves $\frac{\partial G_{V,n}}{\partial R_V}\ge 0$.

We provide here an alternative proof. Consider the interpolating model \eqref{interpUV} where a positive quantity $\Delta$ is added to $R_{V}$. We denote $I_\Delta\big((\bU,\bV);(\bW,\tilde{\bW}_U,\tilde{\bW}_V)\big)$ the mutual information for this new model, so $i_n(t,{\bm \epsilon})=\frac1nI_0\big((\bU,\bV);(\bW,\tilde{\bW}_U,\tilde{\bW}_V)\big)$. By Lemma~\ref{lemma:MIadditivity} this model is mutual information-wise equivalent to the following one:
\begin{align}\label{model:interp_Delta}
\begin{cases}
\bW&=\sqrt{(1-t)\frac{\lambda_n}{n}}\, \bU\otimes \bV + \bZ\,, \\
\tilde \bW_U  &= \sqrt{R_{V}}\,\bU + \tilde \bZ_U\,,\\
\hat \bW_U(\Delta)  &= \sqrt{\Delta}\,\bU + \hat \bZ_U\,,\\
\tilde \bW_V(t,{\bm \epsilon})  &= \sqrt{R_{U}}\,\bV + \tilde \bZ_V\,,
\end{cases} 
\end{align}
where $\hat \bZ_U\sim{\cal N}(0,{\rm I}_n)$ independently of the rest. Namely, 
$$I_\Delta\big((\bU,\bV);(\bW,\tilde{\bW}_U,\tilde{\bW}_V)\big)=I\big((\bU,\bV);(\bW,\tilde{\bW}_U,\hat{\bW}_U(\Delta),\tilde{\bW}_V)\big)\,.$$ 
Using the chain rule for mutual information it is re expressed as
\begin{align*}
&I_\Delta\big(\bV;(\bW(t),\tilde{\bW}_U(t,{\bm \epsilon}),\tilde{\bW}_V(t,{\bm \epsilon}))\big)+I_\Delta\big(\bU;(\bW(t),\tilde{\bW}_U)|\bV\big)\nn
&\qquad=I\big(\bV;(\bW,\tilde{\bW}_U,\hat{\bW}_U(\Delta),\tilde{\bW}_V)\big)+I\big(\bU;(\bW,\tilde{\bW}_U,\hat{\bW}_U(\Delta))|\bV\big)\,.	
\end{align*}
The two mutual information conditioned of $\bV$ are independent of $R_{V}$. Taking a $R_{V}$ derivative on both sides, by the I-MMSE formula Lemma~\ref{app:I-MMSE} the associated MMSE's verify 
$${\rm MMSE}_\Delta(\bV|\bW,\tilde{\bW}_U,\tilde{\bW}_V)={\rm MMSE}(\bV|\bW,\tilde{\bW}_U,\hat{\bW}_U(\Delta),\tilde{\bW}_V)\,.$$
Lemma~\ref{lemma:conditioningMMSE} then implies 
$${\rm MMSE}_\Delta(\bV|\bW,\tilde{\bW}_U,\tilde{\bW}_V)\le {\rm MMSE}(\bV|\bW,\tilde{\bW}_U,\tilde{\bW}_V)$$ 
or equivalently 
$$\EE\langle Q_V\rangle_{t,\Delta}\ge \EE\langle Q_V\rangle_t$$
where ${\rm MMSE}_\Delta(\bV|\cdots)$ and $\EE\langle Q_V\rangle_{t,\Delta}$ are the average MMSE and overlap for $\bV$ corresponding to model \eqref{model:interp_Delta} or equivalently model \eqref{interpUV} with $R_{V,n}$ replaced by $R_{V,n}+\Delta$. This proves $\frac{\partial G_{V,n}}{\partial R_V}\ge 0$.
\end{proof}

\subsection{Lower bound: fully adaptive interpolation path.}
For the converse bound we need to adapt both interpolating functions.
\begin{proposition}[Lower bound] Under the scalings \eqref{sparsePCAscaling} the converse of the bound \eqref{upperboundwishart} holds.
\end{proposition}
\begin{proof}
Consider this time the following Cauchy problem: 
\begin{align}
&\frac{dR_{n}}{dt}(t,{\bm \epsilon})=\big(G_{U,n}(t,R_n(t,{\bm \epsilon})), G_{V,n}(t,R_n(t,{\bm \epsilon}))\big), \qquad R_n(0,{\bm \epsilon})={\bm \epsilon}\,,\label{ODEwishart}
\end{align}
with the functions $G_{U,n}(t,R_n(t,{\bm \epsilon}))\equiv \lambda_n(\rho_{U,n}-{\rm MMSE}(U|\sqrt{\lambda_n \alpha_n \EE\langle Q_V\rangle_t}U+Z))\in[0,\lambda_n\rho_{U,n}]$ and $G_{V,n}(t,R_n(t,{\bm \epsilon}))\equiv\lambda_n\alpha_n \EE\langle Q_V\rangle_t\in[0,\lambda_n\alpha_n\rho_{V,n}]$, or in other words,
% %
\begin{align*}
\big(q_{U,n}(t,{\bm \epsilon}),q_{V,n}(t,{\bm \epsilon})\big)=\big(\rho_{U,n}-{\rm MMSE}(U|\sqrt{\lambda_n \alpha_n \EE\langle Q_V\rangle_t}U+Z), \EE\langle Q_V\rangle_t\big), \qquad R_n(0,{\bm \epsilon})={\bm \epsilon}\,.
\end{align*}
This ODE admits a unique global ${\cal C}^1$ solution $R_{n}^*(t,{\bm \epsilon})=(R_{U,n}^*(t,{\bm \epsilon}),R_{V,n}^*(t,{\bm \epsilon}))$ by the Cauchy-Lipschitz theorem. By the Liouville formula, the Jacobian determinant of the flow ${\bm \epsilon}\mapsto R_{n}^*(t,{\bm \epsilon})$ satisfies
\begin{align*}
	J_{n}(t,{\bm \epsilon})&=\exp \Big\{\int_0^t \Big(\frac{\partial G_{U,n}}{\partial R_U}(s, R_U=R_{U,n}^*(s,{\bm \epsilon}),R_{V,n}^*(s,{\bm \epsilon}))\nn
	&\qquad\qquad\qquad\qquad+\frac{\partial G_{V,n}}{\partial R_V}(s, R_{U,n}^*(s,{\bm \epsilon}),R_V=R_{V,n}^*(s,{\bm \epsilon}))\Big)ds\Big\}\,.
\end{align*}	
Both partial derivatives are positive by the same proof as in the previous paragraph. Then using teh same arguments as previously we conclude that the flow is regular (a $\mathcal{C}^1$-diffeomorphism with Jacobian greater or equal to one). Using this solution we can thus use Propositions~\ref{L-concentration-wishart}, \ref{LU-concentration-wishart} of section~\ref{appendix-overlap} to deduce from the sum rule of Proposition~\ref{prop:sumsule-wishart}
\begin{align*}
&\frac1n I((\bU,\bV);\bW) \nn
&\quad=   \frac{1}{s_n^2}\int d{\bm\epsilon} \Big[I_{n}(U;\{\lambda_n \alpha_n \textstyle{\int_0^1q^*_{V,n}}(t,{\bm \epsilon})dt\}^{1/2}U+Z)+\alpha_nI_{n}(V;\{\lambda_n \textstyle{\int_0^1q^*_{U,n}}(t,{\bm \epsilon})dt\}^{1/2}V+Z)\nn
&\qquad \qquad+\frac{\lambda_n\alpha_n}{2}\int_0^1dt(q^*_{U,n}(t,{\bm \epsilon})-\rho_{U,n})(q^*_{V,n}(t,{\bm \epsilon})-\rho_{V,n})\Big]\nn
&\qquad\qquad\qquad\qquad+O\Big(s_n+\frac{|\ln \rho_{V,n}|^{11/24}}{s_n^{2}n^{1/3}\rho_{V,n}^{9/24}}\big(s_n+\sqrt{\rho_{V,n}|\ln \rho_{V,n}|}\big)^{1/3}\Big)\nn
&\quad\ge \frac{1}{s_n^2}\int d{\bm \epsilon}\int_0^1dt\,i_n^{\rm pot}\big(q_{U,n}^*(t,{\bm \epsilon}),q_{V,n}^*(t,{\bm \epsilon}),\lambda_n,\alpha_n,\rho_{U,n},\rho_{V,n}\big) +O(\cdots)\,.
\end{align*}
To get the last inequality we used the concavity in the SNR of the mutual information for gaussian channels, see Lemma~\ref{lemma:Iconcave} of section~\ref{app:gaussianchannels}. Now note that
\begin{align*}
i_n^{\rm pot}\big(q_{U,n}^*(t,{\bm \epsilon}),q_{V,n}^*(t,{\bm \epsilon}),\lambda_n,\alpha_n,\rho_{U,n},\rho_{V,n}\big)	=\sup_{q_V\in [0,\rho_{V,n}]} i_n^{\rm pot}\big(q_{U,n}^*(t,{\bm \epsilon}),q_{V},\lambda_n,\alpha_n,\rho_{U,n},\rho_{V,n}\big)\,.
\end{align*}
Indeed, the function $g_n(q_U,\cdot):q_V\mapsto i_n^{\rm pot}(q_{U},q_{V};\lambda_n,\alpha_n,\rho_{U,n},\rho_{V,n})$ is concave (by concavity of the mutual information in the SNR, see Lemma~\ref{lemma:Iconcave}) with $q_V$-derivative $$\frac{dg}{dq_V}(q_U,q_V)=\frac{\lambda_n\alpha_n}{2}\big(q_u-\rho_{U,n}+{\rm MMSE}(U|\sqrt{\lambda_n\alpha_n q_{V}}U+Z)\big)$$
(using the I-MMSE relation). By definition of the solution $R_n^*$ of the ODE \eqref{ODEwishart} we have $$\frac{dg}{dq_V}(q_{U,n}^*(t,{\bm \epsilon}),q_V=q_{V,n}^*(t,{\bm \epsilon}))=0\,.$$ By concavity this corresponds to a maximum. Therefore
\begin{align*}
\frac1n I((\bU,\bV);\bW) &\ge \frac{1}{s_n^2}\int d{\bm \epsilon}\int_0^1dt\,\sup_{q_V\in [0,\rho_{V,n}]} i_n^{\rm pot}\big(q_{U,n}^*(t,{\bm \epsilon}),q_{V},\lambda_n,\alpha_n,\rho_{U,n},\rho_{V,n}\big) +O(\cdots)\nn
&\ge {\adjustlimits \inf_{q_U\in[0,\rho_{U,n}]}\sup_{q_{V}\in[0,\rho_{V,n}]}}\,i_n^{\rm pot}\big(q_{U},q_V,\lambda_n,\alpha_n,\rho_{U,n},\rho_{V,n}\big) +O(\cdots)\,.
\end{align*}
\end{proof}

\section{Concentration of free energies}\label{app:free-energy}

For this appendix it is convenient to use the language of statistical mechanics.

\subsection{Statistical mechanics notations for the spiked Wigner (interpolating) model.}
 We express the posterior of the interpolating model 
 \begin{align}\label{tpost}
&dP_{n, t,\epsilon}(\bx|\bW(t),\tilde{\bW}(t,\epsilon))=\frac{1}{\mathcal{Z}_{n, t,\epsilon}(\bW(t),\tilde{\bW}(t,\epsilon))}\nn
&\qquad\qquad\qquad\qquad\qquad\qquad\times\Big(\prod_{i=1}^n dP_{X,n}(x_i)\Big) \exp\big\{-{\cal H}_{n, t, \epsilon}(\bx, \bW(t),\tilde{\bW}(t,\epsilon))\big\}
\end{align}
with normalization constant (partition function) $\mathcal{Z}_{n,t,\epsilon}$ and ``hamiltonian''
\begin{align}
&{\cal H}_{n, t, \epsilon}(\bx , \bW(t),\tilde{\bW}(t,\epsilon))={\cal H}_{n, t, \epsilon}(\bx, \bX,\bZ,\tilde \bZ)\label{Ht}\\\
&\ \equiv\sum_{i< j}^n \Big((1-t)\frac{\lambda_n}{n}\frac{x_i^2x_j^2}{2}-\sqrt{(1-t)\frac{\lambda_n}{n}}x_ix_jW_{ij}(t) \Big)+  R_n(t,\epsilon)\frac{\|\bx\|^2}{2} - \sqrt{R_n(t,\epsilon)} \bx\cdot \tilde \bW(t,\epsilon) \nn
&\ =(1-t)\lambda_n\sum_{i< j}^n \Big(\frac{x_i^2x_j^2}{2n}-\frac{x_ix_jX_iX_j}{n} -\frac{x_ix_jZ_{ij}}{\sqrt{n(1-t)\lambda_n}}\Big)+ R_n(t,\epsilon)\Big(\frac{\|\bx\|^2}{2} - \bx \cdot \bX-\frac{\bx\cdot \tilde{\bZ}}{\sqrt{R_n(t,\epsilon)}}\Big).\nonumber
\end{align}
It will also be convenient to work with ``free energies'' rather than mutual informations. The free energy $F_{n}(t,\epsilon)$ and (its expectation $f_{n}(t,\epsilon)$) for the interpolating model is simply minus the (expected) log-partition function:
\begin{align}
F_{n,t,\epsilon}(\bW(t),\tilde{\bW}(t,\epsilon))&\equiv -\frac1n \ln \mathcal{Z}_{n, t, \epsilon}(\bW(t),\tilde{\bW}(t,\epsilon))\,,\label{F_nonav}\\
f_{n}(t,\epsilon)&\equiv \EE\,F_{n, t, \epsilon}(\bW(t),\tilde{\bW}(t,\epsilon))\label{f_av}	\,.
\end{align}
The expectation $\EE$ carries over the data. The averaged free energy is related to the mutual information $i_{n}(t,\epsilon)$ given by \eqref{fnt} through
\begin{align}
i_{n}(t,\epsilon)	=f_{n}(t,\epsilon)+\frac{n-1}{n}\frac{\rho^2\lambda(1-t)}{4}+\frac{\rho R_n(t,\epsilon)}{2}\,.\label{infomut_freeen}
\end{align}

\subsection{Statistical mechanics notations for the spiked Wishart (interpolating) model.}
 Let the set ${\cal D}_{n,t,{\bm \epsilon}}=\{\bW(t), \tilde \bW_U(t,{\bm \epsilon}),\tilde \bW_V(t,{\bm \epsilon})\}$. In the Wishart case the posterior reads
\begin{align}
dP_{n, t,\epsilon}(\bu,\bv|{\cal D}_{n,t,{\bm \epsilon}})&=\frac{1}{{\cal Z}_{n, t, \epsilon}({\cal D}_{n,t,{\bm \epsilon}})}\Big(\prod_{i=1}^n dP_{U,n}(u_i)\Big) \Big(\prod_{j=1}^{m}dP_{V,n}(v_j)\Big)\nn
&\qquad\qquad\times \exp \big\{-{\cal H}_{n, t, \bm\epsilon}(\bu,\bv, {\cal D}_{n,t,{\bm \epsilon}})\big\} \label{postWish_Ham}
\end{align}
with hamiltonian 
\begin{align}
	{\cal H}_{n, t, \bm\epsilon}(\bu,\bv, {\cal D}_{n,t,{\bm \epsilon}})& \equiv  (1-t)\frac{\lambda_n}{n}\frac{\|\bu\|^2\|\bv\|^2}{2}-\sqrt{(1-t)\frac{\lambda_n}{n}} \bu\cdot (\bW(t)\bv) \nn
	&\quad+  R_{V,n}(t,{\bm \epsilon})\frac{\|\bu\|^2}{2} - \sqrt{R_{V,n}(t,{\bm \epsilon})} \bu\cdot \tilde \bW_U(t,{\bm \epsilon})\nn
	&\quad+  R_{U,n}(t,{\bm \epsilon})\frac{\|\bv\|^2}{2} - \sqrt{R_{U,n}(t,{\bm \epsilon})} \bv\cdot \tilde \bW_V(t,{\bm \epsilon}) \nn
& =(1-t)\frac{\lambda_n}{n}\frac{\|\bu\|^2\|\bv\|^2}{2}-(1-t)\frac{\lambda_n}{n} (\bu\cdot \bU) (\bv\cdot \bV)-\sqrt{(1-t)\frac{\lambda_n}{n}}\bu\cdot (\bZ \bv)\nn
	&\quad +  R_{V,n}(t,{\bm \epsilon})\frac{\|\bu\|^2}{2} - {R_{V,n}(t,{\bm \epsilon})} \bu\cdot \bU -\sqrt{R_{V,n}(t,{\bm \epsilon})} \bu\cdot \tilde \bZ_U \nn
	&\quad+  R_{U,n}(t,{\bm \epsilon})\frac{\|\bv\|^2}{2} - {R_{U,n}(t,{\bm \epsilon})} \bv\cdot \bV-\sqrt{R_{U,n}(t,{\bm \epsilon})} \bv\cdot \tilde \bZ_V\,.\label{Ht_wishart}
 \end{align}
The free energy and its expectation (over the data) are
\begin{align}
F_{n, t, \bm\epsilon}({\cal D}_{n,t,{\bm \epsilon}})&\equiv -\frac1n \ln \mathcal{Z}_{n, t,\epsilon}({\cal D}_{n,t,{\bm \epsilon}})\,,\label{F_nonav_wishart}\\
f_{n}(t,\bm\epsilon)&\equiv \EE\,F_{n,t,\bm\epsilon}({\cal D}_{n,t,{\bm \epsilon}})\label{f_av_wishart}	\,.
\end{align}
Similarly to \eqref{infomut_freeen} the averaged free energy is related to the mutual information by an additive constant (linear in $R_{U,n}$ and $R_{V,n}$) that does not change its concavity properties.
\subsection{Free energy concentration for the Wigner case}
In this section we prove a concentration identity for the free energy \eqref{F_nonav} onto its average \eqref{f_av}.
\begin{proposition}[Free energy concentration for the spiked Wigner model]\label{prop:fconc}
We have 
\begin{align*}
\mathbb{E}\Big[\Big(F_{n,t,\epsilon}(\bW(t),&\tilde{\bW}(t,\epsilon))-f_{n}(t,\epsilon)\Big)^2\Big]
\le\frac{2\rho_nS^2}{n}\Big((2s_n+\lambda_n\rho_n)^2+S^4\Big)+\frac32\frac{\lambda_n\rho_n^2}{n}+2\frac{s_n\rho_n}{n}\,.	
\end{align*}
Considering sequences $\lambda_n$ and $\rho_n$ verifying \eqref{app-scalingRegime} and with $s_n=(1/2)n^{-\alpha}\to0_+$ the bound simplifies to $C(S) \lambda_n^2\rho_n^3/n$ with positive constant $C(S)\le \frac52+8S^2+2S^6$.
\end{proposition}
The proof is based on two classical concentration inequalities,
\begin{proposition}[Gaussian Poincar\'e inequality]\label{poincare}
	Let $\bU = (U_1, \dots, U_N)$ be a vector of $N$ independent standard normal random variables. Let $g: \mathbb{R}^N \to \mathbb{R}$ be a continuously differentiable function. Then
 \begin{align*}
	 \Var (g(\bU)) \leq \E \| \nabla g (\bU) \|^2  \,.
 \end{align*}
\end{proposition}
\begin{proposition}[Efron-Stein inequality]\label{efron_stein}
	Let $\,\mathcal{U}\subset \R$, and a function $g: \mathcal{U}^N \to \mathbb{R}$. Let $\,\bU=(U_1, \dots, U_N)$ be a vector of $N$ independent random variables with law $P_U$ that take values in $\,\mathcal{U}$. Let $\,\bU^{(i)}$ a vector which differs from $\bU$ only by its $i$-th component, which is replaced by $U_i'$ drawn from $P_U$ independently of $\,\bU$. Then
 \begin{align*}
	 \Var(g(\bU)) \leq \frac{1}{2} \sum_{i=1}^N \EE_{\bU}\EE_{U_i'}\big[(g(\bU)-g(\bU^{(i)}))^2\big] \,.
 \end{align*}
\end{proposition}

We start by proving the concentration w.r.t. the gaussian variables. It is convenient to make explicit the dependence of the partition function of the interpolating model in the independent quenched variables instead of the data: 
${\cal Z}_{n,t,\epsilon}(\bX, \bZ, \tilde \bZ)=\mathcal{Z}_{n,t,\epsilon}(\bW(t),\tilde{\bW}(t,\epsilon))$.
\begin{lemma}[Concentration w.r.t. the gaussian variables] We have
\begin{align*}
\mathbb{E}\Big[\Big(\frac{1}{n}\ln {\cal Z}_{n,t,\epsilon}(\bX, \bZ, \tilde \bZ)-\frac{1}{n}\mathbb{E}_{\bZ,\tilde \bZ}\ln {\cal Z}_{n,t,\epsilon}(\bX, \bZ, \tilde \bZ)\Big)^2\Big]\le \frac32\frac{\lambda_n\rho_n^2}{n}+2\frac{s_n\rho_n}{n}\,.	
\end{align*}
\end{lemma}
\begin{proof}
Fix all variables except $\bZ, \tilde \bZ$. Let $g(\bZ, \tilde \bZ)\equiv-\frac{1}{n}\ln {\cal Z}_{n,t,\epsilon}(\bX, \bZ, \tilde \bZ)$ be the free energy seen as a function of the gaussian variables only. The free energy gradient reads $\EE\|\nabla g\|^2=\EE\|\nabla_\bZ g\|^2+\EE\|\nabla_{\tilde \bZ} g\|^2$. Let us denote ${\cal H}(t)\equiv {\cal H}_{n,t,\epsilon}$ the interpolating Hamiltonian \eqref{Ht}.
\begin{align*}
\EE\|\nabla_\bZ g\|^2= \frac1{n^2}\EE\|\langle \nabla_\bZ{\cal H}(t)\rangle_t\|^2	&=\frac{(1-t)\lambda_n}{n^3}\sum_{i<j}\EE[\langle x_ix_j \rangle_t^2]\le\frac{(1-t)\lambda_n}{n^3}\sum_{i<j}\EE\langle (x_ix_j)^2 \rangle_{t}\nn
&\overset{\rm N}{=}\frac{(1-t)\lambda_n}{n^3}\sum_{i<j}\EE[(X_iX_j)^2]\le \frac{\lambda_n\rho_n^2}{2n}\,.
\end{align*}
where we used a Nishimori identity for the last equality.
Similarly, and using $\lambda_n\rho_n\ge1$ and $s_n< 1/2$,
\begin{align*}
\EE\|\nabla_{\tilde \bZ} g\|^2=\frac{R(\epsilon)}{n^2}\EE\|\langle \bx \rangle_t\|^2\le\frac{R(\epsilon)}{n^2}\EE\langle \|\bx\|^2 \rangle_t\overset{\rm N}{=}\frac{R(\epsilon)}{n^2}\EE\|\bX\|^2\le \frac{(2s_n + \rho_n\lambda_n)\rho_n}{n}\,.
\end{align*}
Therefore Proposition~\ref{poincare} directly implies the stated result.
\end{proof}

We now consider the fluctuations due to the signal realization:
\begin{lemma}[Concentration w.r.t. the spike] We have
\begin{align*}
\mathbb{E}\Big[\Big(-\frac{1}{n}\mathbb{E}_{\bZ,\tilde \bZ}\ln {\cal Z}_{n,t,\epsilon}(\bX, \bZ, \tilde \bZ)-f_{n}(t,\epsilon)\Big)^2\Big]\le \frac{2\rho_nS^2}{n}\Big((2s_n+\lambda_n\rho_n)^2+S^4\Big)\,.	
\end{align*}

\end{lemma}
\begin{proof}
Let $g(\bX)\equiv-\frac{1}{n}\mathbb{E}_{\bZ,\tilde \bZ}\ln {\cal Z}_{n,t,\epsilon}( \bX, \bZ, \tilde \bZ)$. Define $\bX^{(i)}$ as a vector with same entries as $\bX$ except the $i$-th one that is replaced by $X_i'$ drawn independently from $P_{X,n}$. Let us estimate $(g(\bX)-g(\bX^{(i)}))^2$ by interpolation. Let ${\cal H}(t,s\bX+(1-s)\bX^{(i)})$ be the interpolating Hamiltonian \eqref{Ht} with $\bX$ replaced by $s\bX+(1-s)\bX^{(i)}$. Then
\begin{align*}
\EE\big[(g(\bX)-&g(\bX^{(i)}))^2\big] = \EE\Big[\Big(\int_0^1 ds \frac{dg}{ds}(s\bX+(1-s)\bX^{(i)})\Big)^2\Big] \nn
&=\frac1{n^2}\EE\Big[\Big(\int_0^1 ds \Big\langle\frac{d{\cal H}}{ds}(t,s\bX+(1-s)\bX^{(i)})\Big\rangle_t\Big)^2\Big]\nn
&= \frac{1}{n^2}\EE\Big[\Big( (X_i-X_i')\Big\langle R_\epsilon(t) x_i +  \frac{1-t}{n}x_i\sum_{j(\neq i)}X_jx_j\Big\rangle_t\Big)^2\Big]\nn
&\le \frac{2}{n^2}\EE\Big[(X_i-X_i')^2\Big(\langle x_i\rangle_t^2(2s_n+\lambda_n\rho_n)^2+\frac1{n^2}\sum_{j,k(\neq i)}X_jX_k\langle x_ix_j\rangle_t\langle x_ix_k\rangle_t\Big)\Big]\nn
&\le \frac{2}{n^2}\EE\big[(X_i-X_i')^2\big]\Big(S^2(2s_n+\lambda_n\rho_n)^2+S^6\Big)\nn
&\le\frac{4\rho_nS^2}{n^2}\Big((2s_n+\lambda_n\rho_n)^2+S^4\Big)\,.
\end{align*}
We used $(a+b)^2\le 2(a^2+b^2)$ for the second inequality and $\EE[(X_i-X_i')^2]=2{\rm Var}(X_i)\le 2\rho_n$. Therefore Proposition~\ref{efron_stein} implies the claim.
\end{proof}
\subsection{Free energy concentration for the Wishart case}
In this section we prove a concentration identity for the free energy \eqref{F_nonav_wishart} onto its average \eqref{f_av_wishart}. 
\begin{proposition}[Free energy concentration for the spiked Wishart model]\label{prop:fconc_wishart}We have 
\begin{align*}
&\mathbb{E}\Big[\Big(F_{n,t,\bm \epsilon}(\bW(t), \tilde \bW_U(t,{\bm \epsilon}),\tilde \bW_V(t,{\bm \epsilon}))-f_{n}(t,\bm \epsilon)\Big)^2\Big]\!\le\! \frac{C_{F,n}}{n}
\end{align*}
where
\begin{align*}
C_{F,n}&\equiv 2\rho_{U,n}S^2\big\{(2s_n+\lambda_n\alpha_n\rho_{V,n})^2+\alpha_n^2S^4\big\}+2\alpha_n\rho_{V,n}S^2\big\{(2s_n+\lambda_n\rho_{U,n})^2+S^4\big\}\nn
&\hspace{5cm}+3\lambda_n\alpha_n\rho_{U,n} \rho_{V,n}+2s_n(1+\alpha_n)\bar \rho_n\,.
\end{align*}
In the particular case of the scalings \eqref{sparsePCAscaling} we have $C_{F,n}\le C |\ln \rho_{V,n}|$ for some positive constant $C$ that may depend on anything but $n$.
\end{proposition}

The proofs are brief as they are similar to those for the spiked Wigner model. The partition function expressed with the independent quenched variables is ${\cal Z}_{n,t,\bm\epsilon}(\bU,\bV,  \bZ, \tilde \bZ_U,\tilde \bZ_V)\equiv {\cal Z}_{n,t,\bm\epsilon}({\cal D}_{n,t,{\bm \epsilon}})$.
\begin{lemma}[Concentration w.r.t. the gaussian variables] Let $\bar \rho_n\equiv \max(\rho_{U,n},\rho_{V,n})$. We have 
\begin{align*}
&\mathbb{E}\Big[\Big(\frac{1}{n}\ln {\cal Z}_{n,t,\bm\epsilon}(\bU,\bV,  \bZ, \tilde \bZ_U,\tilde \bZ_V)-\frac{1}{n}\mathbb{E}_{\bZ,\tilde \bZ_U,\tilde \bZ_V}\ln {\cal Z}_{n,t,\bm\epsilon}(\bU,\bV,  \bZ, \tilde \bZ_U,\tilde \bZ_V)\Big)^2\Big]\nn
&\qquad\qquad\qquad\qquad\qquad\qquad\qquad\qquad\qquad\qquad\le 3\frac{\lambda_n\alpha_n\rho_{U,n} \rho_{V,n}}{n}+2\frac{s_n(1+\alpha_n)\bar \rho_n}{n}\,.	
\end{align*}
\begin{proof}
	Let $g(\bZ, \tilde \bZ_U,\tilde \bZ_V)$ be the free energy \eqref{F_nonav_wishart} seen as a function of only the gaussian variables. Based on the hamiltonian expression \eqref{Ht_wishart} we compute the gradient:
	\begin{align*}
\EE\|\nabla_\bZ g\|^2=\frac{(1-t)\lambda_n}{n^3}\EE\|\langle \bu \otimes \bv \rangle_t\|_{\rm F}^2\le \frac{\lambda_n}{n^3}\EE\| \bU\|^2  \EE\|\bV\|^2\le \frac{\lambda_n\alpha_n\rho_{U,n}\rho_{V,n}}{n}
\end{align*}
where the bracket is w.r.t. the interpolating model posterior \eqref{postWish_Ham}. We used that $\bu,\bU\in\mathbb{R}^n$ while $\bv,\bV\in\mathbb{R}^m$, and $\alpha_n\equiv m/n$. Similarly
\begin{align*}
\EE\| \nabla_{\tilde \bZ_U} g\|^2&=\frac{R_{V,n}(t,{\bm \epsilon})^2}{n^2}\EE\|\langle \bu \rangle_t\|^2\le \frac{(2s_n + \rho_{V,n}\alpha_n\lambda_n)\rho_{U,n}}{n}\,,\nn
\EE\| \nabla_{\tilde \bZ_V} g\|^2&=\frac{R_{U,n}(t,{\bm \epsilon})^2}{n^2}\EE\|\langle \bv \rangle_t\|^2\le \frac{(2s_n + \rho_{U,n}\lambda_n)\alpha_n\rho_{V,n}}{n}\,.
\end{align*}
Proposition~\ref{poincare} implies the result.
\end{proof}
\end{lemma}
\begin{lemma}[Concentration w.r.t. the spikes] We have
\begin{align*}
&\mathbb{E}\Big[\Big(-\frac{1}{n}\mathbb{E}_{\bZ,\tilde \bZ_U,\tilde \bZ_V}\ln {\cal Z}_{n,t,\bm\epsilon}(\bU,\bV,  \bZ, \tilde \bZ_U,\tilde \bZ_V)-f_{n}(t,\bm\epsilon)\Big)^2\Big]\nn
&\hspace{2cm}\le \frac{2\rho_{U,n}S^2}{n}\Big((2s_n+\lambda_n\alpha_n\rho_{V,n})^2+\alpha_n^2S^4\Big)+\frac{2\alpha_n\rho_{V,n}S^2}{n}\Big((2s_n+\lambda_n\rho_{U,n})^2+S^4\Big)\,.	
\end{align*}

\end{lemma}
\begin{proof}
Let $g(\bU)$ be the free energy \eqref{F_nonav_wishart} seen as a function of $\bU$ only. Define $\bU^{(i)}$ as a vector with same entries as $\bU$ except the $i$-th one that is replaced by $U_i'$ drawn independently from $P_{U,n}$. Let ${\cal H}(t,s\bU+(1-s)\bU^{(i)})$ be the interpolating Hamiltonian \eqref{Ht_wishart} with $\bU$ replaced by $s\bU+(1-s)\bU^{(i)}$. We bound
\begin{align*}
\EE\big[(g(\bU)-g(\bU^{(i)}))^2\big] &= \frac1{n^2}\EE\Big[\Big(\int_0^1 ds \Big\langle\frac{d{\cal H}}{ds}(t,s\bU+(1-s)\bU^{(i)})\Big\rangle_t\Big)^2\Big]\nn
&= \frac{1}{n^2}\EE\Big[\Big( (U_i-U_i')\Big\langle R_{V,n} u_i +  \frac{1-t}{n}u_i(\bv \cdot \bV)\Big\rangle_t\Big)^2\Big]\nn
&\le \frac{2}{n^2}\EE\Big[(U_i-U_i')^2\Big(\langle u_i\rangle_t^2(2s_n+\lambda_n\alpha_n\rho_{V,n})^2+\frac1{n^2}\langle u_i(\bv\cdot \bV)\rangle_t^2\Big)\Big]\nn
&\le \frac{2}{n^2}\EE\big[(U_i-U_i')^2\big]\Big(S^2(2s_n+\lambda_n\alpha_n\rho_{V,n})^2+\alpha_n^2S^6\Big)\nn
&\le\frac{4\rho_{U,n}S^2}{n^2}\Big((2s_n+\lambda_n\alpha_n\rho_{V,n})^2+\alpha_n^2S^4\Big)\,.
\end{align*}
Similarly, and with an anlogous notation $\bV^{(i)}$, we obtain
\begin{align*}
\EE\big[(g(\bV)-g(\bV^{(i)}))^2\big] \le\frac{4\rho_{V,n}S^2}{n^2}\Big((2s_n+\lambda_n\rho_{U,n})^2+S^4\Big)\,.
\end{align*}
Proposition~\ref{efron_stein} then implies the claim.
\end{proof}
\section{Concentration for the overlaps}\label{appendix-overlap}
\subsection{Overlap concentration for the Wigner case: proof of inequality \eqref{over-concen}}\label{Wigner-overlap}
The derivations below will apply for any $t\in[0,1]$ so we drop all un-necessary notations and indices. Only the dependence of the free energies in $R(\epsilon)\equiv R_n(t,\epsilon)$ matters, so we denote $F(R(\epsilon))\equiv F_{n,t,\epsilon}(\bW(t),\tilde{\bW}(t,\epsilon))$ and $f(R(\epsilon))\equiv f_{n}(t,\epsilon)$.

Let $\mathcal{L}$ be the $R(\epsilon)$-derivative of the Hamiltonian \eqref{Ht} divided by $n$:
\begin{align}
\mathcal{L}(\bx,\bX,\tilde \bZ) =\mathcal{L} \equiv \frac1n \frac{d{\cal H}_{n,t,\epsilon}}{dR(\epsilon)}= \frac{1}{n}\Big(\frac{\|\bx\|^2}{2} - \bx\cdot \bX - \frac{\bx\cdot \tilde \bZ}{2\sqrt{R(\epsilon)}} \Big)\,.\label{def_L}
\end{align}
The overlap fluctuations are upper bounded by those of $\mathcal{L}$, which are easier to control, as
\begin{align}
\mathbb{E}\big\langle (Q - \mathbb{E}\langle Q \rangle_{t})^2\big\rangle_{t} \le 4\,\mathbb{E}\big\langle (\mathcal{L} - \mathbb{E}\langle \mathcal{L}\rangle_{t})^2\big\rangle_{t}\,.\label{remarkable}
\end{align}
The bracket is again the expectation w.r.t. the posterior of the interpolating model \eqref{t_post}.
A detailed derivation of this inequality can be found in appendix \ref{proof:remarkable_id} and involves only elementary algebra using the Nishimori identity
and integrations by parts w.r.t.\ the gaussian noise $\tilde \bZ$.

We have the following identities: for any given realisation of the quenched disorder
\begin{align}
 \frac{dF}{dR(\epsilon)}  &= \langle \mathcal{L} \rangle_{t} \,,\label{first-derivative}\\
 \frac{1}{n}\frac{d^2F}{dR(\epsilon)^2}  &= -\big\langle (\mathcal{L}  - \langle \mathcal{L} \rangle_t)^2\big\rangle_{t}+
 \frac{1}{4 n^2R(\epsilon)^{3/2}} \langle \bx\rangle_{t} \cdot \tilde \bZ\,.\label{second-derivative}
\end{align}
The gaussian integration by part formula \eqref{GaussIPP} with hamiltonian \eqref{Ht} yields
\begin{align}
\frac{\mathbb{E} \big\langle \tilde \bZ\cdot  \bx \big\rangle_t}{\sqrt{R(\epsilon)}} 
	 =   \mathbb{E}\big\langle \|\bx\|^2 \big\rangle_t - \EE\|\langle \bx \rangle_t\|^2 \overset{\rm N}{=}\mathbb{E}\big\langle \|\bx\|^2 \big\rangle_t - \EE \big\langle\bX\cdot \bx \big\rangle_t= \mathbb{E}\big\langle \|\bx\|^2 \big\rangle_t - n\,\EE \langle Q\rangle_t\,.\label{NishiTildeZ}
\end{align}
Therefore averaging \eqref{first-derivative} and \eqref{second-derivative} we find 
\begin{align}
 \frac{df}{d R(\epsilon)} &= \mathbb{E}\langle \mathcal{L} \rangle_{t} 
 \overset{\rm N}{=}-\frac{1}{2}\mathbb{E}\langle Q\rangle_{t}\,,\label{first-derivative-average}\\
 \frac{1}{n}\frac{d^2f}{dR(\epsilon)^2} &= -\mathbb{E}\big\langle (\mathcal{L} - \langle \mathcal{L} \rangle_{t})^2\big\rangle_{t}
 +\frac{1}{4n^2R(\epsilon)} \mathbb{E}\big\langle \|\bx- \langle \bx \rangle_t\|^2\big\rangle_t\,.\label{second-derivative-average}
\end{align} 
We always work under the assumption that the map $\epsilon\in [s_n, 2s_n] \mapsto R(\epsilon)\in [R(s_n), R(2s_n)]$ is regular, and do not repeat this assumption in the statements below. The concentration inequality \eqref{over-concen} is  a direct consequence of the following result (combined with Fubini's theorem):
\begin{proposition}[Total fluctuations of ${\cal L}$]\label{L-concentration} Let the sequences $\lambda_n$ and $\rho_n$ verify \eqref{app-scalingRegime}. Then $$\int_{s_n}^{2s_n} d\epsilon\,\mathbb{E}\big\langle (\mathcal{L} - \mathbb{E}\langle \mathcal{L}\rangle_{t})^2\big\rangle_{t} \le C\Big(\frac{\lambda_n\rho_n}{ns_n}\big(1+\lambda_n\rho_n^2\big)\Big)^{1/3}$$
for a constant $C>0$ that is independent of $n$, as long as the r.h.s. is $\omega(1/n)$.
\end{proposition}

The proof of this proposition is broken in two parts, using the decomposition
\begin{align*}
\mathbb{E}\big\langle (\mathcal{L} - \mathbb{E}\langle \mathcal{L}\rangle_{t})^2\big\rangle_{t}
& = 
\mathbb{E}\big\langle (\mathcal{L} - \langle \mathcal{L}\rangle_t)^2\big\rangle_{t}
+ 
\mathbb{E}\big[(\langle \mathcal{L}\rangle_t - \mathbb{E}\langle \mathcal{L}\rangle_t)^2\big]\,.
\end{align*}
Thus it suffices to prove the two following lemmas. The first lemma expresses concentration w.r.t.\ the posterior distribution (or ``thermal fluctuations'') and is a direct consequence of concavity properties of the average free energy and the Nishimori identity.
\begin{lemma}[Thermal fluctuations of $\cal L$]\label{thermal-fluctuations}
	We have $$\int_{s_n}^{2s_n} d\epsilon\, 
  \mathbb{E} \big\langle (\mathcal{L} - \langle \mathcal{L}\rangle_t)^2 \big\rangle_t  \le \frac{\rho_n}{n}\Big(1+\frac{\ln2}{4}\Big)\,.$$
\end{lemma}
\begin{proof}

We emphasize again that the interpolating free energy \eqref{fnt} is here viewed as a function of $R(\epsilon)$. In the argument that follows we consider derivatives of this function w.r.t. $R(\epsilon)$.
By \eqref{second-derivative-average}
\begin{align}
\mathbb{E}\big\langle (\mathcal{L} - \langle \mathcal{L} \rangle_t)^2\big\rangle_t
& = 
-\frac{1}{n}\frac{d^2f}{dR(\epsilon)^2}
+\frac{1}{4n^2R(\epsilon)} \big(\mathbb{E}\big\langle \|\bx\|^2 \big\rangle_t - \EE\|\langle \bx \rangle_t\|^2\big)
\nonumber \\ &
\leq 
-\frac{1}{n}\frac{d^2f}{dR(\epsilon)^2} +\frac{\rho_n}{4n\epsilon} \,,
\label{directcomputation}
\end{align}
where we used $R(\epsilon)\geq \epsilon$ and $\frac1n\mathbb{E}\langle \|\bx\|^2\rangle_t \overset{\rm N}{=} \mathbb{E}[X_1^2]=\rho_n$. We integrate this inequality over $\epsilon\in [s_n, 2s_n]$. Recall the map 
$\epsilon\mapsto R(\epsilon)$ has a Jacobian $\ge 1$, is ${\cal C}^1$ 
and has a well defined ${\cal C}^1$ inverse since we have assumed that it is regular. Thus integrating \eqref{directcomputation} and performing a change of variable (to get the second inequality) we obtain
\begin{align*}
\int_{s_n}^{2s_n} d\epsilon\, \mathbb{E}\big\langle (\mathcal{L} - \langle \mathcal{L} \rangle_t)^2\big\rangle_t
& \leq 
- \frac{1}{n}\int_{s_n}^{2s_n} d\epsilon \,\frac{d^2f}{dR(\epsilon)^2} + \frac{\rho_n}{4n}\int_{s_n}^{2s_n} \,\frac{d\epsilon}{\epsilon} 
\nonumber \\ &
\leq 
- \frac{1}{n}\int_{R(s_n)}^{R(2s_n)} dR(\epsilon) \,\frac{d^2f}{dR(\epsilon)^2}
+ \frac{\rho_n}{4n}\int_{s_n}^{2s_n} \,\frac{d\epsilon}{\epsilon} 
\nonumber \\ &
=
\frac{1}{n}\Big(\frac{df}{dR(\epsilon)}(R(s_n)) 
- \frac{df}{dR(\epsilon)}(R(2s_n))\Big)+ \frac{\rho_n}{4n}\ln 2\,.
\end{align*}
We have $|f'(R(\epsilon))| = |\mathbb{E}\langle Q\rangle_t/2|\le \rho_n/2$ so the first term is certainly smaller in absolute value than $\rho_n/n$. This concludes the proof of Lemma \ref{thermal-fluctuations}.
\end{proof}

The second lemma expresses the concentration w.r.t.\ the quenched disorder variables
and is a consequence of the concentration of the free energy onto its average (w.r.t. the quenched variables).
\begin{lemma}[Quenched fluctuations of $\cal L$]\label{disorder-fluctuations}
	Let the sequences $\lambda_n$ and $\rho_n$ verify \eqref{app-scalingRegime}. Then $$\int_{s_n}^{2s_n} d\epsilon\, 
  \mathbb{E}\big[ (\langle \mathcal{L}\rangle_t - \mathbb{E}\langle \mathcal{L}\rangle_t)^2\big] \le C\Big(\frac{\lambda_n\rho_n}{ns_n}\big(1+\lambda_n\rho_n^2\big)\Big)^{1/3}$$
  for a constant $C>0$ that is independent of $n$, as long as the r.h.s. is $\omega(1/n)$. 
\end{lemma}
\begin{proof}
Consider the following functions of $R(\epsilon)$:
\begin{align}\label{new-free}
 & \tilde F(R(\epsilon)) \equiv F(R(\epsilon)) +S\frac{\sqrt{R(\epsilon)}}{n} \sum_{i=1}^n\vert \tilde Z_i\vert\,,
 \nonumber \\ &
 \tilde f(R(\epsilon)) \equiv \mathbb{E} \,\tilde F(R(\epsilon))= f(R(\epsilon)) + S\sqrt{R(\epsilon)} \mathbb{E}\,\vert \tilde Z_1\vert\,.
\end{align}
Because of 
\eqref{second-derivative} we see that the second derivative of $\tilde F(R(\epsilon))$ w.r.t. $R(\epsilon)$ is negative so that it is concave. Note $F(R(\epsilon))$ itself is not necessarily concave in $R(\epsilon)$, although $f(R(\epsilon))$ is. Concavity of $f(R(\epsilon))$ is not obvious from \eqref{second-derivative-average} (obtained from differentiating $\EE\langle{\cal L}\rangle_t$ w.r.t. $R(\epsilon)$) but can be seen from \eqref{secondDer_f_pos} (obtained instead by differentiating $-\frac12\EE\langle{Q}\rangle_t$) which reads $\frac{d}{dR(\epsilon)}\EE\langle Q\rangle_t=-2\frac{d^2}{dR(\epsilon)^2}f\ge 0$. Equivalently it follows from the relation \eqref{infomut_freeen} between mutual information and free energy and the concavity of the mutual information Lemma~\ref{lemma:Iconcave}. Evidently $\tilde f(R(\epsilon))$ is concave too.
Concavity then allows to use the following standard lemma:
\begin{lemma}[A bound for concave functions]\label{lemmaConvexity}
Let $G(x)$ and $g(x)$ be concave functions. Let $\delta>0$ and define $C^{-}_\delta(x) \equiv g'(x-\delta) - g'(x) \geq 0$ and $C^{+}_\delta(x) \equiv g'(x) - g'(x+\delta) \geq 0$. Then
\begin{align*}
|G'(x) - g'(x)| \leq \delta^{-1} \sum_{u \in \{x-\delta,\, x,\, x+\delta\}} |G(u)-g(u)| + C^{+}_\delta(x) + C^{-}_\delta(x)\,.
\end{align*}
\end{lemma}
First, from \eqref{new-free} we have 
\begin{align}\label{fdiff}
 \tilde F(R(\epsilon)) - \tilde f(R(\epsilon)) = F(R(\epsilon)) - f(R(\epsilon)) + S\sqrt{R(\epsilon)}  A_n 
\end{align} 
with $A_n \equiv \frac{1}{n}\sum_{i=1}^n \vert \tilde Z_i\vert -\mathbb{E}\,\vert \tilde Z_1\vert$.
Second, from \eqref{first-derivative}, \eqref{first-derivative-average} we obtain for the $R(\epsilon)$-derivatives
\begin{align}\label{derdiff}
 \tilde F'(R(\epsilon)) - \tilde f'(R(\epsilon)) = 
\langle \mathcal{L} \rangle_t-\mathbb{E}\langle \mathcal{L} \rangle_{t} + \frac{SA_n}{2\sqrt{R(\epsilon)}} \,.
\end{align}
From \eqref{fdiff} and \eqref{derdiff} it is then easy to show that Lemma \ref{lemmaConvexity} implies
\begin{align}\label{usable-inequ}
\vert \langle \mathcal{L}\rangle_t - \mathbb{E}\langle \mathcal{L}\rangle_t\vert&\leq 
\delta^{-1} \sum_{u\in \{R(\epsilon) -\delta,\, R(\epsilon),\, R(\epsilon)+\delta\}}
 \big(\vert F(u) - f(u) \vert + S\vert A_n \vert \sqrt{u} \big)\nonumber\\
 &\qquad\qquad\qquad\qquad
  + C_\delta^+(R(\epsilon)) + C_\delta^-(R(\epsilon)) + \frac{S\vert A_n\vert}{2\sqrt \epsilon} 
\end{align}
where $C_\delta^-(R(\epsilon))\equiv \tilde f'(R(\epsilon)-\delta)-\tilde f'(R(\epsilon))\ge 0$ 
and $C_\delta^+(R(\epsilon))\equiv \tilde f'(R(\epsilon))-\tilde f'(R(\epsilon)+\delta)\ge 0$. 
We used $R(\epsilon)\ge \epsilon$ for the term $S\vert A_n\vert/(2\sqrt \epsilon)$. Note that $\delta$ will 
be chosen later on strictly smaller than $s_n$ so that $R(\epsilon) -\delta \geq \epsilon - \delta \geq s_n -\delta$ remains 
positive. Remark that by independence of the noise variables $\mathbb{E}[A_n^2]  = (1-2/\pi)/n\le 1/n$. 
We square the identity \eqref{usable-inequ} and take its expectation. Then using $(\sum_{i=1}^pv_i)^2 \le p\sum_{i=1}^pv_i^2$, and 
that $R(\epsilon)\le 2s_n+\lambda_n\rho_n$, as well as the free energy concentration Proposition \ref{prop:fconc} (under the assumption that $\lambda_n$ and $\rho_n$ verify \eqref{app-scalingRegime}),
\begin{align}\label{intermediate}
 \frac{1}{9}\mathbb{E}\big[(\langle \mathcal{L}\rangle_t - \mathbb{E}\langle \mathcal{L}\rangle_t)^2\big]
 &
 \leq \, 
 \frac{3}{n\delta^2} \Big(C\lambda_n^2\rho_n^3 +S(2s_n+\lambda_n\rho_n+\delta)\Big) 
 \nonumber \\ & \qquad \qquad\qquad\qquad\qquad+ C_\delta^+(R(\epsilon))^2 + C_\delta^-(R(\epsilon))^2
 + \frac{S}{4n\epsilon} \,.
\end{align}
Recall $|C_\delta^\pm(R(\epsilon))|=|\tilde f'(R(\epsilon)\pm\delta)-\tilde f'(R(\epsilon))|$. By \eqref{first-derivative-average}, \eqref{new-free}  and $R(\epsilon)\ge \epsilon$ we have
\begin{align}
|\tilde f'(R(\epsilon))|  \leq \frac12\Big(\rho_n  +\frac{S}{\sqrt{R(\epsilon)}} \Big)\leq \frac12\Big(\rho_n  +\frac{S}{\sqrt \epsilon} \Big)\label{boudfprime}	
\end{align}
Thus, as $\epsilon\ge s_n$,  
$$|C_\delta^\pm(R(\epsilon))|\le \rho_n  +\frac{S}{\sqrt{\epsilon-\delta}}\le \rho_n  +\frac{S}{\sqrt{s_n-\delta}}\,.$$ We reach
\begin{align*}
 \int_{s_n}^{2s_n} d\epsilon\, \big\{C_\delta^+(R(\epsilon))^2 + C_\delta^-(R(\epsilon))^2\big\}
 &\leq 
 \Big(\rho_n +\frac{S}{\sqrt {s_n-\delta}}\Big)
 \int_{s_n}^{2s_n} d\epsilon\, \big\{C_\delta^+(R(\epsilon)) + C_\delta^-(R(\epsilon))\big\}
 \nonumber \\ 
  &\leq  \Big(\rho_n +\frac{S}{\sqrt {s_n-\delta}}\Big)
 \int_{R(s_n)}^{R(2s_n)} dR(\epsilon)\, \big\{C_\delta^+(R(\epsilon)) + C_\delta^-(R(\epsilon))\big\}
 \nonumber \\ 
&= \Big(\rho_n +\frac{S}{\sqrt {s_n-\delta}}\Big)\Big[\Big(\tilde f(R(s_n)+\delta) - \tilde f(R(s_n)-\delta)\Big)\nn
&\qquad\qquad\qquad\qquad+ \Big(\tilde f(R(2s_n)-\delta) - \tilde f(R(2s_n)+\delta)\Big)\Big]
\end{align*}
where we used that the Jacobian of the ${\cal C}^1$-diffeomorphism $\epsilon\mapsto R(\epsilon)$ is $\ge 1$ (by regularity) for 
the second inequality. The mean value theorem and \eqref{boudfprime} 
imply $|\tilde f(R(\epsilon)-\delta) - \tilde f(R(\epsilon)+\delta)|\le \delta(\rho_n  +\frac{S}{\sqrt{s_n-\delta}})$. Therefore
\begin{align*}
 \int_{s_n}^{2s_n} d\epsilon\, \big\{C_\delta^+(R(\epsilon))^2 + C_\delta^-(R(\epsilon))^2\big\}\leq 
 2\delta \Big(\rho_n +\frac{S}{\sqrt{s_n-\delta}}\Big)^2\,.
\end{align*}
Set $\delta  = \delta_n = o(s_n)$. Thus, integrating \eqref{intermediate} over $\epsilon\in [s_n, 2s_n]$ yields
\begin{align*}
 &\int_{s_n}^{2s_n} d\epsilon\, 
 \mathbb{E}\big[(\langle \mathcal{L}\rangle_t - \mathbb{E}\langle \mathcal{L}\rangle_{t})^2\big]\nonumber\\
 &\qquad\qquad\leq \frac{27s_n}{n\delta_n^2}\Big(C\lambda_n^2\rho_n^3 +S(2s_n+\lambda_n\rho_n+\delta_n)\Big) +18\delta_n \Big(\rho_n +\frac{S}{\sqrt {s_n-\delta_n}}\Big)^2 +  \frac{9 S\ln 2}{4n}  \nonumber\\
 &\qquad\qquad\leq\frac{Cs_n\lambda_n\rho_n}{n\delta_n^2}(1+\lambda_n\rho_n^2)+\frac{C\delta_n}{s_n}+\frac{C}{n}
\end{align*}
where the constant $C$ is generic, and may change from place to place. Finally we optimize the bound choosing $\delta_n^3=  s_n^2\lambda_n\rho_n(1+\lambda_n\rho_n^2)/n$. We verify the condition $\delta_n =o(s_n)$: we have $(\delta_n/s_n)^3=O(\lambda_n\rho_n(1+\lambda_n\rho_n^2)/(ns_n))$ which, by \eqref{app-scalingRegime}, indeed tends to $0_+$ for an appropriately chosen sequence $s_n$. So the dominating term $\delta_n/s_n$ gives the result. 
\end{proof}

\subsection{Overlap concentration for the Wishart case: proof of inequality \eqref{over-concen-wishart}}\label{appendix-overlap-wishart}
\subsubsection{Controlling $\bm{Q_V}$}
Again we drop all un-necessary notations and indices and keep only the dependence of the free energies on $R({\bm \epsilon})=(R_{U}({\bm \epsilon}),R_{V}({\bm \epsilon}))\equiv (R_{U,n}(t,{\bm \epsilon}),R_{V,n}(t,{\bm \epsilon}))$. We denote $F(R({\bm \epsilon}))$ and $f(R({\bm \epsilon}))$, respectively, the free energies \eqref{F_nonav_wishart} and \eqref{f_av_wishart}. We start proving the ovelap concentration for $Q_V\equiv \bv\cdot \bV/m$. As the computations are similar as for the spiked Wigner model we are more brief.

Let $\mathcal{L}_V$ be the $R_U({\bm \epsilon})$-derivative of the hamiltonian \eqref{Ht_wishart} divided by $m=\alpha_nn$:
\begin{align}
\mathcal{L}_V \equiv \frac{1}{m}\Big(\frac{\|\bv\|^2}{2} - \bv\cdot \bV - \frac{\bv\cdot \tilde \bZ_V}{2\sqrt{R_U({\bm \epsilon})}} \Big)\,.\label{def_L_U}
\end{align}
We have as before 
\begin{align}
\mathbb{E}\big\langle (Q_V - \mathbb{E}\langle Q_V \rangle_{t})^2\big\rangle_{t} \le 4\,\mathbb{E}\big\langle (\mathcal{L}_V - \mathbb{E}\langle \mathcal{L}_V\rangle_{t})^2\big\rangle_{t}\,.	\label{LV_QV}
\end{align}
We relate ${\cal L}_V$'s fluctuations to the free energy through
\begin{align}
 \frac{dF}{dR_U({\bm \epsilon})}  &= \alpha_n\langle \mathcal{L}_V \rangle_{t} \,,\label{first-derivative-wishart}\\
 \frac{1}{n}\frac{d^2F}{dR_U({\bm \epsilon})^2}  &= -\alpha_n^2\big\langle (\mathcal{L}_V  - \langle \mathcal{L}_V \rangle_t)^2\big\rangle_{t}+
 \frac{1}{4 n^2R_U({\bm \epsilon})^{3/2}} \langle \bv\rangle_{t} \cdot \tilde \bZ_V\,,\label{second-derivative-wishart}\\
 \frac{df}{d R_U({\bm \epsilon})} &=\alpha_n\EE\langle \mathcal{L}_V\rangle_t\overset{\rm N}{=} -\frac{\alpha_n}{2}\mathbb{E}\langle Q_V\rangle_{t}\,,\label{first-derivative-average-wishart}\\
 \frac{1}{n}\frac{d^2f}{dR_U({\bm \epsilon})^2} &= -\alpha_n^2\mathbb{E}\big\langle (\mathcal{L}_V - \langle \mathcal{L}_V \rangle_{t})^2\big\rangle_{t}
 +\frac{1}{4n^2R_U({\bm \epsilon})} \mathbb{E}\big\langle \|\bv- \langle \bv \rangle_t\|^2\big\rangle_t\,.\label{second-derivative-average-wishart}
\end{align} 
We work under the assumption that the map ${\bm \epsilon}\in [s_n, 2s_n]^2 \mapsto (R_U({\bm \epsilon}),R_V({\bm \epsilon}))$ is regular (that is ${\cal C}^1$ with a ${\cal C}^1$ inverse and a Jacobian determinant $\ge 1$). The concentration inequality \eqref{over-concen-wishart} follows from:
\begin{proposition}[Total fluctuations of $\mathcal{L}_V$]\label{L-concentration-wishart} For any sequences $(\delta_n), (s_n)$ verifying $\delta_n<s_n$ the fluctuations $\int_{[s_n,2s_n]^2} d{\bm \epsilon}\,\mathbb{E}\langle (\mathcal{L}_V - \mathbb{E}\langle \mathcal{L}_V\rangle_{t})^2\rangle_{t}$ are bounded by the sum of the r.h.s. of inequalities \eqref{ineq:thermal-wishart} and \eqref{generic_quenched_wishart} below. In the special case of the scalings \eqref{sparsePCAscaling} there exists $C>0$ independent of $n$ such that  $$\int_{[s_n,2s_n]^2} d{\bm \epsilon}\,\mathbb{E}\big\langle (\mathcal{L}_V - \mathbb{E}\langle \mathcal{L}_V\rangle_{t})^2\big\rangle_{t} \le C\Big(\frac1n\sqrt{\frac{|\ln \rho_{V,n}|}{\rho_{V,n}}}\big(s_n+\sqrt{\rho_{V,n}|\ln \rho_{V,n}|}\big)^2\Big)^{1/3}$$
as long as the right hand side is $\omega(s_n/n)$.
\end{proposition}

We start with the proof of the thermal fluctuations:
\begin{lemma}[Thermal fluctuations of ${\cal L}_V$]\label{thermal-fluctuations-wishart}
	We have 
	\begin{align}
		\int_{[s_n,2s_n]^2} d{\bm \epsilon}\, \mathbb{E}\big\langle (\mathcal{L}_V - \langle \mathcal{L}_V \rangle_t)^2\big\rangle_t  \le \frac{\lambda_n\rho_{V,n}^2}{n}+ \frac{s_n\rho_{V,n}}{n\alpha_n}\Big(1+\frac{\ln 2}{4}\Big)\,.\label{ineq:thermal-wishart}
	\end{align}	
\end{lemma}
\begin{proof}
Integrating \eqref{second-derivative-average-wishart}, using $R_U\ge \epsilon_U$ and the regularity assumption for ${\bm \epsilon}\mapsto R({\bm \epsilon})$ we obtain
\begin{align*}
\alpha_n^2\int_{[s_n,2s_n]^2} d{\bm \epsilon}\, \mathbb{E}\big\langle (\mathcal{L}_V - \langle \mathcal{L}_V \rangle_t)^2\big\rangle_t
&
\leq 
- \frac{1}{n}\int_{R([s_n,2s_n]^2)} dR({\bm \epsilon}) \,\frac{d^2f}{dR_U({\bm \epsilon})^2}
+  \frac{\rho_{V,n}\alpha_n }{4n}\int_{[s_n,2s_n]^2}\frac{d{\bm \epsilon}}{\epsilon_U}\,.
\end{align*}
We have $R([s_n,2s_n]^2)\subseteq {\cal R}\equiv [s_n,2s_n+\lambda_n\rho_{U,n}]\times [s_n,2s_n+\lambda_n\alpha_n\rho_{V,n}]$. Moreover the second $R_U$-derivative of $f$ is negative. This is not immediately obvious from \eqref{second-derivative-average-wishart} but can be easily shown similarly to the Wigner case, and  is equivalent to say that the averaged overlap cannot decrease when the SNR $R_U$ increases. Therefore we can integrate over the larger set ${\cal R}$ to get a bound:
\begin{align*}
&\alpha_n^2\int_{[s_n,2s_n]^2} d{\bm \epsilon}\, \mathbb{E}\big\langle (\mathcal{L}_V - \langle \mathcal{L}_V \rangle_t)^2\big\rangle_t
\leq 
- \frac{1}{n}\int_{\cal R} dR_VdR_U \,\frac{d^2f}{dR_U^2}
+  \frac{\rho_{V,n}\alpha_n s_n}{4n}\ln 2\nonumber \\ 
&\leq  \frac{1}{n}\int_{s_n}^{2s_n+\lambda_n\alpha_n\rho_{V,n}} \!dR_V\Big(\frac{df}{dR_U}(R_U=s_n,R_V)\! 
-\! \frac{df}{dR_U}(R_U=2s_n+\lambda_n\rho_{U,n},R_V)\Big)\!+\!  \frac{\rho_{V,n}\alpha_n s_n}{4n}\ln 2\\
&
\le 
\frac{(s_n+\lambda_n\alpha_n\rho_{V,n})}{n}\rho_{V,n}\alpha_n+ \frac{\rho_{V,n}\alpha_n s_n}{4n}\ln 2\,.
\end{align*}
In the last line we used $|f'(R_U)| \le \rho_{V,n}\alpha_n/2$ which follows from \eqref{first-derivative-average-wishart}. This concludes the proof of Lemma \ref{thermal-fluctuations-wishart}.
\end{proof}

We now consider the randomness due to the quenched variables.
\begin{lemma}[Quenched fluctuations of ${\cal L}_V$]\label{disorder-fluctuations-wishart}
	 For any sequences $(\delta_n), (s_n)$ verifying $\delta_n<s_n$ we have the generic bound \eqref{generic_quenched_wishart} below. In the special case of the scalings \eqref{sparsePCAscaling} there exists $C>0$ independent of $n$ s.t.
	 %$$\int_{[s_n,2s_n]^2} d{\bm \epsilon}\, \mathbb{E}\big[ (\langle \mathcal{L}_V\rangle_t - \mathbb{E}\langle \mathcal{L}_V\rangle_t)^2\big] \le C\Big(\frac{|\ln \rho_{V,n}|^{3}\rho_{V,n}}{n^{2}}\Big)^{1/6}$$
	 $$
	 \int_{[s_n,2s_n]^2} d{\bm \epsilon}\, \mathbb{E}\big[ (\langle \mathcal{L}_V\rangle_t - \mathbb{E}\langle \mathcal{L}_V\rangle_t)^2\big] \le C\Big(\frac1n\sqrt{\frac{|\ln \rho_{V,n}|}{\rho_{V,n}}}\big(s_n+\sqrt{\rho_{V,n}|\ln \rho_{V,n}|}\big)^2\Big)^{1/3}
	 $$
	 as long as the right hand side is $\omega(s_n/n)$.
\end{lemma}
\begin{proof}
Consider the following functions of $R_U({\bm \epsilon})$:
\begin{align*}
 & \tilde F(R_U({\bm \epsilon})) \equiv F(R({\bm \epsilon})) +S\frac{\sqrt{R_U({\bm\epsilon})}}{n} \sum_{i=1}^m\vert \tilde Z_{V,i}\vert\,,
 \nonumber \\ &
 \tilde f(R_U({\bm \epsilon})) \equiv \mathbb{E} \,\tilde F(R({\bm \epsilon}))= f(R({\bm \epsilon})) + S\alpha_n \sqrt{R_U({\bm \epsilon})}  \mathbb{E}\,\vert \tilde Z_{V,1}\vert\,.
\end{align*}
Both functions are concave in ${R_U({\bm \epsilon})}$. Letting $A_n \equiv \frac{1}{n}\sum_{i=1}^m \vert \tilde Z_{V,i}\vert -\alpha_n\mathbb{E}\,\vert \tilde Z_{V,1}\vert$, Lemma \ref{lemmaConvexity} implies
\begin{align}\label{usable-inequ-wishart}
\alpha_n\vert \langle \mathcal{L}_V\rangle_t - \mathbb{E}\langle \mathcal{L}_V\rangle_t\vert
\leq &
\delta^{-1} \sum_{u\in \{R_U -\delta,\, R_U,\, R_U+\delta\}}
 \big(\vert F(R_U=u) - f(R_U=u) \vert + S\vert A_n \vert \sqrt{u} \big)
 \nonumber\\
 &
  + C_\delta^+(R_U) + C_\delta^-(R_U) + \frac{S\vert A_n\vert}{2\sqrt{\epsilon_U}} 
\end{align}
where $C_\delta^-(R_U)\equiv \tilde f'(R_U-\delta)-\tilde f'(R_U)\ge 0$ 
and $C_\delta^+(R_U)\equiv \tilde f'(R_U)-\tilde f'(R_U+\delta)\ge 0$. 
We used $R_U\ge \epsilon_U$. We have $\mathbb{E}[A_n^2]  \le  \alpha_n/n$. $\delta$ will 
be chosen strictly smaller than $s_n$ so that $R_U -\delta \geq \epsilon_U - \delta \geq s_n -\delta$ remains 
positive. We square the identity \eqref{usable-inequ-wishart} and take its expectation. Then using $(\sum_{i=1}^pv_i)^2 \le p\sum_{i=1}^pv_i^2$, and 
that $R_U\le 2s_n+\lambda_n\rho_{U,n}$, as well as the free energy concentration Proposition \ref{prop:fconc_wishart}
\begin{align}\label{intermediate-wishart}
 \frac{\alpha_n^2}{9}\mathbb{E}\big[(\langle \mathcal{L}_V\rangle_t - \mathbb{E}\langle \mathcal{L}_V\rangle_t)^2\big]
 \leq \, &
 \frac{3}{n\delta^2} \Big(C_{F,n} +S\alpha_n(2s_n+\lambda_n\rho_{U,n}+\delta)\Big) 
 \nonumber \\ & 
 + C_\delta^+(R_U)^2 + C_\delta^-(R_U)^2
 + \frac{S\alpha_n}{4n\epsilon_U} \,.
 \end{align}
 We have $|\tilde f'(R_U)|  \leq \alpha_n(\rho_{V,n}  +S/\sqrt {\epsilon_U} )/2$. Thus $|C_\delta^\pm(R_U)|\le \alpha_n(\rho_{V,n}  +S/\sqrt{s_n-\delta})$. We reach, using the regularity of the map ${\bm \epsilon}\mapsto R({\bm \epsilon})$ and that $R([s_n,2s_n]^2)\subseteq [s_n, 2s_n+\lambda_n\rho_{U,n}]\times  [s_n, 2s_n+\lambda_n\alpha_n\rho_{V,n}]$,
\begin{align*}
 \int_{[s_n,2s_n]^2} d{\bm \epsilon}\, &\big\{C_\delta^+(R_U({\bm \epsilon}))^2 + C_\delta^-(R_U({\bm \epsilon}))^2\big\}\nn
 &\leq 
 \alpha_n\Big(\rho_{V,n} +\frac{S}{\sqrt {s_n-\delta}}\Big)
 \int_{[s_n,2s_n]^2} d{\bm \epsilon}\, \big\{C_\delta^+(R_U) + C_\delta^-(R_U)\big\}
 \nonumber \\ 
 &\leq  \alpha_n\Big(\rho_{V,n} +\frac{S}{\sqrt {s_n-\delta}}\Big)
 \int_{R([s_n,2s_n]^2)} dR({\bm \epsilon})\, \big\{C_\delta^+(R_U({\bm \epsilon})) + C_\delta^-(R_U({\bm \epsilon}))\big\}
 \nonumber \\ 
 &\leq  \alpha_n\Big(\rho_{V,n} +\frac{S}{\sqrt {s_n-\delta}}\Big)(s_n+\lambda_n\alpha_n\rho_{V,n})
 \int_{s_n}^{2s_n+\lambda_n\rho_{U,n}} dR_U \, \big\{C_\delta^+(R_U) + C_\delta^-(R_U)\big\}
 \nonumber \\  
&= \alpha_n\Big(\rho_{V,n} +\frac{S}{\sqrt {s_n-\delta}}\Big)(s_n+\lambda_n\alpha_n\rho_{V,n})\Big[\Big(\tilde f(s_n+\delta) - \tilde f(s_n-\delta)\Big)\nn
&\qquad+ \Big(\tilde f(2s_n+\lambda_n\rho_{U,n}-\delta) - \tilde f(2s_n+\lambda_n\rho_{U,n}+\delta)\Big)\Big]\nn
&\le 2\delta\alpha_n^2\Big(\rho_{V,n} +\frac{S}{\sqrt {s_n-\delta}}\Big)^2(s_n+\lambda_n\alpha_n\rho_{V,n})\,.
\end{align*}
For the last inequality we employed the mean value theorem to assert $|\tilde f(R_U-\delta) - \tilde f(R_U+\delta)|\le \delta\alpha_n(\rho_{V,n}  +S/\sqrt{s_n-\delta})$. Thus, integrating \eqref{intermediate-wishart} over ${\bm\epsilon}\in [s_n, 2s_n]^2$ yields that for any $\delta_n <s_n $ we have (the sequence $C_{F,n}$ comes from Proposition \ref{prop:fconc_wishart})
\begin{align}
 \int_{[s_n,2s_n]} d{\bm \epsilon}\, 
 \mathbb{E}\big[(\langle \mathcal{L}_V\rangle_t - \mathbb{E}\langle \mathcal{L}_V\rangle_{t})^2\big]&\leq \frac{27s_n^2}{\delta_n^2n\alpha_n^2}\Big(C_{F,n} +S \alpha_n(2s_n+\lambda_n\rho_{U,n}+\delta_n)\Big) \nn
 &\hspace{-1cm} +18\delta_n \Big(\rho_{V,n} +\frac{S}{\sqrt {s_n-\delta_n}}\Big)^2(s_n+\lambda_n\alpha_n\rho_{V,n}) +  \frac{9 Ss_n\ln 2}{4\alpha_nn} .\label{generic_quenched_wishart}
\end{align}
Under the scalings \eqref{sparsePCAscaling} and choosing $\delta_n= o(s_n)$ this simplifies to
\begin{align}
 &\int_{[s_n,2s_n]} \!d{\bm \epsilon}\, 
 \mathbb{E}\big[(\langle \mathcal{L}_V\rangle_t - \mathbb{E}\langle \mathcal{L}_V\rangle_{t})^2\big]\leq\frac{Cs_n^2}{\delta_n^2n}\sqrt{\frac{|\ln \rho_{V,n}|}{\rho_{V,n}}} \!+\!\frac{C\delta_n}{s_n}\big(s_n+\sqrt{\rho_{V,n}|\ln \rho_{V,n}|}\big) \! +\!  \frac{Cs_n}{n}\label{tofindequivalent}
\end{align}
where the constant $C$ is generic, and may change from place to place. Optimizing $\delta_n$ yields $$\delta_n^3 = \Theta\Big(\frac{s_n^3\sqrt{|\ln \rho_{V,n}|}}{n(s_n\sqrt{\rho_{V,n}}+\rho_{V,n}\sqrt{|\ln\rho_{V,n}|})}\Big)\,.$$
It is easy to see that if $\rho_{V,n}=\omega(1/n)$, i.e., $n\rho_{V,n}\to +\infty$ then $\delta_n= o(s_n)$. This proves the result.
\end{proof}

\subsubsection{Controlling $\bm{Q_U}$}
For the control of $Q_{U}\equiv\bu\cdot \bU/n$ we follow  the same derivation, except for working with 
\begin{align*}
\mathcal{L}_U \equiv \frac{1}{n}\Big(\frac{\|\bu\|^2}{2} - \bu\cdot \bU - \frac{\bu\cdot \tilde \bZ_U}{2\sqrt{R_V({\bm \epsilon})}} \Big)\,.
\end{align*}
The overlap fluctuations are bounded as $\mathbb{E}\langle (Q_U - \mathbb{E}\langle Q_U \rangle_{t})^2\rangle_{t} \le 4\,\mathbb{E}\langle (\mathcal{L}_U - \mathbb{E}\langle \mathcal{L}_U\rangle_{t})^2\rangle_{t}$. The free energy $R_V$-derivatives and ${\cal L}_U$ are then related by similar identities as \eqref{first-derivative-wishart}--\eqref{second-derivative-average-wishart} but with $V$ replaced by $U$, $\bv$ by $\bu$ and $\alpha_n$ replaced by $1$. Working out the thermal fluctuations then gives
\begin{align*}
		\int_{[s_n,2s_n]^2} d{\bm \epsilon}\, \mathbb{E}\big\langle (\mathcal{L}_U - \langle \mathcal{L}_U \rangle_t)^2\big\rangle_t  \le \frac{\lambda_n\rho_{U,n}^2}{n}+ \frac{s_n\rho_{U,n}}{n}\Big(1+\frac{\ln 2}{4}\Big)\,.
	\end{align*}	
Considering now the quenched fluctuations, a careful derivation of the equivalent identity to \eqref{tofindequivalent} under the scalings \eqref{sparsePCAscaling} yields (under the assumption $\delta_n=o(s_n)$)
\begin{align*}
 &\int_{[s_n,2s_n]} d{\bm \epsilon}\, 
 \mathbb{E}\big[(\langle \mathcal{L}_U\rangle_t - \mathbb{E}\langle \mathcal{L}_U\rangle_{t})^2\big]\leq\frac{Cs_n^2}{\delta_n^2n}|\ln \rho_{V,n}| +C\frac{\delta_n}{s_n}\sqrt{\frac{|\ln \rho_{V,n}|}{\rho_{V,n}}}  +  \frac{Cs_n}{n}\,.
\end{align*}
Optimizing over $\delta_n$ yields $\delta_n = \Theta(s_nn^{-1/3}(\rho_{V,n}|\ln \rho_{V,n}|)^{1/6})$, so $\delta_n= o(s_n)$. This finally gives, once combined with the thermal fluctuations bound:
\begin{proposition}[Total fluctuations of $\mathcal{L}_U$]\label{LU-concentration-wishart} Under the scalings \eqref{sparsePCAscaling} there exists a constant $C>0$ independent of $n$ such that $$\int_{[s_n,2s_n]^2} d{\bm \epsilon}\,\mathbb{E}\big\langle (\mathcal{L}_U - \mathbb{E}\langle \mathcal{L}_U\rangle_{t})^2\big\rangle_{t} \le C\Big(\frac{(\ln \rho_{V,n})^2}{n\rho_{V,n}}\Big)^{1/3}$$
as long as the right hand side is $\omega(s_n/n)$.
\end{proposition}
\section{Proof of inequality \eqref{remarkable}}\label{proof:remarkable_id}

Let us drop the index in the bracket $\langle -\rangle_t$ and simply denote $R\equiv R_{n}(t,\epsilon)$. We start by proving the identity
\begin{align}
-2\,\mathbb{E}\big\langle Q(\mathcal{L} - \mathbb{E}\langle \mathcal{L}\rangle)\big\rangle
&=\mathbb{E}\big\langle (Q - \mathbb{E}\langle Q \rangle)^2\big\rangle
+ \mathbb{E}\big\langle (Q-   \langle Q \rangle)^2\big\rangle\,.\label{47}
\end{align}
Using the definitions $Q \equiv \frac1n \bx\cdot \bX$ and \eqref{def_L} gives
\begin{align}
2\,\mathbb{E}\big\langle Q ({\cal L} -\mathbb{E}\langle {\cal L} \rangle) \big\rangle
	 =\, &  \mathbb{E} \Big [ \frac{1}{n}\big\langle Q \|\bx\|^2 \big\rangle - 2 \langle Q^2 \rangle - \frac{1}{n\sqrt{R}}\big\langle Q (\tilde \bZ \cdot \bx) \big\rangle \Big ] \nonumber \\
	&  -  \mathbb{E}\langle Q \rangle  \, \mathbb{E} \Big [ \frac{1}{n}\big\langle \|\bx\|^2 \big\rangle - 2 \langle Q \rangle - \frac{1}{n\sqrt{R}} \tilde \bZ \cdot \langle \bx\rangle \Big ]\,. \label{eq:QL:1}
\end{align}
The gaussian integration by part formula \eqref{GaussIPP} with Hamiltonian \eqref{Ht} yields
\begin{align*}
\frac{1}{n\sqrt{R}}\mathbb{E}\big\langle Q (\tilde \bZ \cdot \bx) \big\rangle 
	& = \frac{1}{n}\mathbb{E}\big\langle Q \|\bx\|^2 \big\rangle - \frac{1}{n}\EE\big\langle Q (\bx \cdot \langle \bx \rangle) \big \rangle \overset{\rm N}{=}\frac1{n}\mathbb{E}\big\langle Q \|\bx\|^2 \big\rangle - \EE [\langle Q  \rangle^2]\,.
\end{align*}
Fort the last equality we used the Nishimori identity as follows
$$
\frac1n\EE\big\langle Q (\bx \cdot \langle \bx \rangle) \big \rangle=\frac1{n^2}\EE\big\langle (\bx\cdot \bX) (\bx \cdot \langle \bx \rangle) \big \rangle\overset{\rm N}{=} \frac1{n^2}\EE \big\langle (\bX\cdot \bx) (\bX \cdot \langle \bx \rangle) \big \rangle=\EE [\langle Q \rangle^2]\,.
$$ 
Note that we already proved \eqref{NishiTildeZ}, namely 
$$
\frac{1}{n\sqrt{R}}\mathbb{E} \langle \tilde \bZ\cdot  \bx \rangle 
=   \frac1n \mathbb{E}\big\langle \|\bx\|^2 \big\rangle - \EE \langle Q\rangle \, .
$$
Therefore \eqref{eq:QL:1} finally simplifies to 
\begin{align*}
2\,\mathbb{E}\big\langle Q &({\cal L} -\mathbb{E}\langle {\cal L} \rangle) \big\rangle = \mathbb{E}[\langle Q\rangle^2] - 2\,\mathbb{E}\langle Q^2\rangle +  \mathbb{E}[\langle Q\rangle]^2 = -  \big ( \mathbb{E}\langle Q^2\rangle - \mathbb{E}[\langle Q\rangle]^2 \big ) -  \big ( \mathbb{E}\langle Q^2\rangle - \mathbb{E}[\langle Q\rangle^2] \big ).
\end{align*} 
which is identity \eqref{47}.

This identity implies the inequality
\begin{align*}
2\big|\mathbb{E}\big\langle Q(\mathcal{L} - \mathbb{E}\langle \mathcal{L}\rangle)\big\rangle\big|&=2\big|\mathbb{E}\big\langle (Q-\mathbb{E}\langle Q \rangle)(\mathcal{L} - \mathbb{E}\langle \mathcal{L}\rangle)\big\rangle\big|\ge \mathbb{E}\big\langle (Q - \mathbb{E}\langle Q \rangle)^2\big\rangle
\end{align*}
and an application of the Cauchy-Schwarz inequality gives
\begin{align*}
2\big\{\mathbb{E}\big\langle (Q-\mathbb{E}\langle Q \rangle)^2\big\rangle\, \mathbb{E}\big\langle(\mathcal{L} - \mathbb{E}\langle \mathcal{L}\rangle)^2\big\rangle \big\}^{1/2}	\ge\mathbb{E}\big\langle (Q - \mathbb{E}\langle Q \rangle)^2\big\rangle\,.
\end{align*}
This ends the proof of \eqref{remarkable}.\\
\section{Heurisitic derivation of the phase transition}\label{sec:allOrNothing}
\subsection{The Wigner case}
In this section we analyze the potential function in order to heuristically locate the information theoretic transition in the special case of the spiked Wigner model with Bernoulli prior $P_X={\rm Ber}(\rho)$. The main hypotheses behind this computation are $i)$ that the SNR $\lambda=\lambda(\rho)$ varies with $\rho$ as $\lambda=4\gamma |\ln \rho|/\rho$ with $\gamma >0$ and independent of $\rho$; that $ii)$ in this SNR regime the potential possesses only two minima $\{q^+,q^-\}$ that approach, as $\rho\to 0_+$, the boundary values $q^-=o(\rho/|\ln \rho|)$ and $q^+\to\rho$. For the Bernoulli prior the potential explicitly reads
\begin{align*}
	i_n^{\rm pot}(q,\lambda,\rho)&\equiv\frac{\lambda (q^2+\rho^2)}{4} -(1-\rho)\EE\ln\Big\{1-\rho+\rho e^{-\frac12\lambda q+\sqrt{\lambda q}Z}\Big\}-\rho\,\EE\ln\Big\{1-\rho+\rho e^{\frac12\lambda q+\sqrt{\lambda q}Z}\Big\}\,.
\end{align*}
We used that 
\begin{align}
I(X;\sqrt{\gamma} X+Z)=-\EE \ln \int dP_X(x)e^{-\frac12\gamma x^2+\gamma Xx+\sqrt{\gamma}Zx}+\frac12\EE[X^2]\gamma\,.	\label{linkPsiI}
\end{align}
Let us compute this function around its assumed minima. Starting with $q^-=o(\rho/|\ln \rho|)$ (this means that this quantity goes to $0_+$ faster than $\rho/|\ln \rho|$ as $\rho$ vanishes) we obtain at leading order after a careful Taylor expansion in $\lambda q^- \to 0_+$ (the symbol $\approx$ means equality up to lower order terms as $\rho \to 0_+$)
\begin{align}
	i_n^{\rm pot}(q^-,\lambda,\rho)
	&\approx \frac{\lambda (q^-)^2}4+\frac{\lambda\rho^2}4-\frac{\rho(\lambda q^-)^2}8\approx\frac{\lambda\rho^2}4= \gamma \rho |\ln \rho|\,.	\label{solq-}
\end{align}
For the other minimum $q^+\to\rho$, because $\lambda q^+ \to +\infty$ the $Z$ contribution in the exponentials appearing in the potential can be dropped due to the precense of the square root. We obtain at leading order
\begin{align*}
	i_n^{\rm pot}(q^+,\lambda,\rho)&\approx 2\gamma\rho|\ln \rho|-\ln\{1+\rho^{1+2\gamma}\}-\rho \ln\{1+\rho^{1-2\gamma}\}\,.
\end{align*}
Here there are two cases to consider: $\gamma>1/2$ and $0<\gamma \le 1/2$. We start with $\gamma>1/2$. In this case the potential simplifies to
\begin{align*}
	i_n^{\rm pot}(q^+,\lambda,\rho)&\approx  \rho|\ln \rho|\,. %\label{solq+}
\end{align*}
Now for $0<\gamma\le 1/2$ we have
\begin{align*}
	i_n^{\rm pot}(q^+,\lambda,\rho)&\approx  2\gamma\rho|\ln \rho|\,.
\end{align*}
The information theoretic threshold $\lambda_c=\lambda_c(\rho)$ is defined as the first non-analiticy in the mutual information. In the present setting this corresponds to a discontinuity of the first derivative w.r.t. the SNR of the mutual information (and we therefore speak about a``first-order phase transition''). By the I-MMSE formula this threshold manifests itself as a discontinuity in the MMSE. In the high sparsity regime $\rho \to 0_+$ the transition is actually as sharp as it can be with a $0$--$1$ behavior. This translates, at the level of the potential, as the SNR threshold where its minimum is attained at $q^-$ just below and instead at $q^+$ just above. So we equate $\lim_{\rho\to 0_+}i_n^{\rm pot}(q^-,\lambda_c,\rho)=\lim_{\rho\to 0_+}i_n^{\rm pot}(q^+,\lambda_c,\rho)$ and solve for $\lambda_c$. This is only possible, under the constraint $\gamma >0$ independent of $\rho$, in the case $\gamma >1/2$ and gives $\gamma=1$ which is the claimed information theoretic threshold $\lambda_c(\rho)=4|\ln \rho|/\rho$. Repeating this analysis for the Bernoulli-Rademacher prior $P_X=(1-\rho)\delta_0+\frac12\rho(\delta_{-1}+\delta_1)$ leads the same threshold.

Another piece of information gained from this analysis is that around the transition the mutual information divided by $n$ is $\Theta(\rho |\ln \rho|)$. Therefore the proper normalization for the mutual information is $(n \rho|\ln \rho|)^{-1}I(\bX;\bW)$ for it to have a well defined non trivial limit in the regime $\rho\to 0_+$.

Finally for $\gamma \le 1$ the minimum of the potential is attained at $q^-$ and the rescaled mutual information $(n \rho|\ln \rho|)^{-1}I(\bX;\bW)$ equals $\gamma$ as seen from \eqref{solq-}. If instead $\gamma \ge 1$ the minimum is attained at $q^+$ and the mutual information instead saturates to $1$, so we get formula \eqref{limitingpotential}.
\subsection{The Wishart case}
We do the same analysis but for the spiked covariance model with Bernoulli-Rademacher distributed $\bV$, namely $P_U={\cal N}(0,1)$ (so $\rho_U=1$) and $P_{V}=(1-\rho)\delta_{0}+\frac{1}{2}\rho(\delta_{-1}+\delta_1)$. But again, the analysis is similar for Bernoulli prior $P_{V}={\rm Ber}(\rho)$ and leads to the same threshold. In the Bernoulli-Rademacher case the potential simplifies to
\begin{align*}
i_n^{\rm pot}(q_U,q_V,\lambda,\alpha,1,\rho)  =  \frac{\lambda \alpha}{2} (q_U-1)(q_V-\rho)
+ \frac12\ln(1+\lambda \alpha q_V) +\alpha I_{n}(V;\sqrt{\lambda q_U}V+Z)\,.
\end{align*}
 This potential is concave in $q_V$. Equating the $q_V$-derivative of this potential to zero yields the stationary condition
\begin{align}
	q_U=q_U(q_V) = \frac{\lambda \alpha q_V}{1+\lambda \alpha q_V} \label{qUfuncqV}\,.% \qquad \Leftrightarrow \qquad q_V^*=\frac{q_U}{\lambda \alpha(1-q_U)}\,.
\end{align}
So plugging back this supremum in the two-letters potential and using again \eqref{linkPsiI} gives
\begin{align*}
% i_n^{\rm pot}(q_U,q_V,\lambda,\alpha,\rho_U,\rho_V)  &=  \frac{\rho_V\lambda \alpha}{2}-\frac{q_U}{2} -\frac12 \ln (1-q_U)\nn
% &\qquad\qquad-\alpha \EE \ln \int dP_V(v)\exp\Big\{-\frac12\lambda q_U v^2+\lambda q_U Vv+\sqrt{\lambda q_U}Zv\Big\}\,,
i_n^{\rm pot}(q_U,q_V,\lambda,\alpha,1,\rho)  &=  \frac{\lambda \alpha}{2} \frac{\rho-q_V}{1+\lambda \alpha q_V}
+ \frac12\ln(1+\lambda \alpha q_V)+\frac{\rho\alpha\lambda q_U}{2}\nn
&\qquad\qquad\qquad\qquad\qquad-\alpha\EE \ln \int dP_V(v)e^{-\frac12\lambda q_U v^2+\lambda q_U Vv+\sqrt{\lambda q_U}Zv}
\end{align*}
where $q_U=q_U(q_V)$ verifies \eqref{qUfuncqV}. It finally becomes, using the Bernoulli-Rademacher prior for $P_V$ as well as $Z=-Z$ in law (because $Z\sim{\cal N}(0,1)$),
\begin{align}
i_n^{\rm pot}(q_U,q_V,\lambda,\alpha,1,\rho)  &=  \frac{\lambda \alpha}{2} \frac{\rho-q_V}{1+\lambda \alpha q_V}
+ \frac12\ln(1+\lambda \alpha q_V)+\frac{\rho\alpha\lambda q_U}{2}\nn
&\qquad-\alpha\EE\Big[ (1-\rho)\ln \Big\{1-\rho + \frac{\rho}{2}e^{-\frac12\lambda q_U+\sqrt{\lambda q_U}Z}+\frac{\rho}{2}e^{-\frac12\lambda q_U-\sqrt{\lambda q_U}Z}\Big\}\nn
&\qquad\qquad+\rho\ln \Big\{1-\rho + \frac{\rho}{2}e^{\frac12\lambda q_U+\sqrt{\lambda q_U}Z}+\frac{\rho}{2}e^{-\frac32\lambda q_U-\sqrt{\lambda q_U}Z}\Big\}\Big]\,.\label{iuv}
\end{align}
Similarly as for the Wigner case the hypotheses behind this computation are $i)$ that the SNR $\lambda=\lambda(\rho)$ varies with $\rho$ as $\lambda=\sqrt{4\gamma |\ln \rho|/(\alpha\rho)}$ with $\gamma >0$ and independent of $\rho$; that $ii)$ in this SNR regime the potential possesses only two minima $\{q_V^+,q_V^-\}$ that approach, as $\rho\to 0_+$, the boundary values $q_V^-=o(\rho/|\ln \rho|)$ and $q_V^+\to\rho$. This implies that as $\rho \to 0_+$
\begin{align*}
	% \lambda q_U(q_V^+)\to \frac{4\gamma |\ln \rho|}{1+\sqrt{4\gamma\alpha\rho |\ln \rho|}}\to +\infty \quad \text{and}\quad \lambda q_U(q_V^-)\to \frac{4\gamma o(1)}{1+o(\sqrt{4\gamma\alpha\rho /|\ln \rho|})}\to 0_+\,.
	\lambda q_U(q_V^+)\to \frac{4\gamma |\ln \rho|}{1+\sqrt{4\gamma\alpha\rho |\ln \rho|}}\to +\infty \quad \text{and}\quad \lambda q_U(q_V^-)= o(1)\to 0_+\,.
\end{align*}
Because both $\lambda q_V^{+} =\Theta(\sqrt{\rho |\ln \rho|}) \to 0_+$ and  $\lambda q_V^{-} =o(\sqrt{\rho /|\ln \rho|}) \to 0_+$ we have
\begin{align*}
-\frac{\lambda \alpha}{2} \frac{q_V^\pm}{1+\lambda \alpha q_V^\pm}+\frac12\ln(1+\lambda \alpha q_V^\pm)	\approx -\frac{\lambda\alpha q_V^\pm}{2}+\frac{(\lambda\alpha q_V^\pm)^2}{2} +\frac{\lambda\alpha q_V^\pm}{2}-\frac{(\lambda\alpha q_V^\pm)^2}{4}=\frac{(\lambda\alpha q_V^\pm)^2}4\,.
\end{align*}
We start by considering the case $q_V^-$. In this case a Taylor expansion gives at leading order
\begin{align}
i_n^{\rm pot}(q_U(q_V^-),q_V^-,\lambda,\alpha,1,\rho)  &\approx  \frac{\rho \lambda \alpha}{2} -\frac{\rho(\lambda \alpha)^2q_V^-}{2}
+\frac{(\lambda\alpha q_V^-)^2}4+\frac{\rho\alpha\lambda q_U(q_V^-)}{2}-\frac{\alpha\rho (\lambda q_U(q_V^-))^2}8\nn
&\approx \frac{\rho \lambda \alpha}{2}+o(\rho)=\sqrt{\alpha \gamma\rho |\ln \rho|} +o(\rho)\,.\label{iqv-}
\end{align}
We now consider the other minimum $q_V^+\to \rho$. In this case we have $(\lambda\alpha q_V^+)^2/4\to-\alpha \gamma\rho \ln \rho$. As $\lambda q_U(q_V^+)\to +\infty$ the $Z$ contributions in the exponentials appearing in \eqref{iuv} are sub-dominant and therefore dropped. 
 We obtain at leading order
\begin{align*}
i_n^{\rm pot}(q_U(q_V^+),q_V^+,\lambda,\alpha,1,\rho)  &\approx  \frac{\rho \lambda \alpha}{2} -\frac{\rho(\lambda \alpha)^2q_V^+}{2}
+\frac{(\lambda\alpha q_V^+)^2}4+\frac{\rho\alpha\lambda q_U(q_V^+)}{2}-\alpha\rho\ln \Big\{1 + \frac{\rho}{2}e^{\frac12\lambda q_U(q_V^+)}\Big\}\nn
&\approx \frac{\rho \lambda\alpha}{2}+2\alpha \gamma\rho \ln \rho-\alpha \gamma\rho\ln \rho-2\alpha \gamma\rho\ln\rho
-\alpha\rho\ln \Big\{1 + \rho \,e^{\frac12\lambda q_U(q_V^+)}\Big\}\nn
&=\sqrt{\alpha\gamma \rho  |\ln \rho|}-\alpha \gamma\rho \ln \rho
-\alpha\rho\ln \big\{1 + \rho^{1-2\gamma}\big\}\,.
\end{align*}
We need again to distinguish cases. Starting with $\gamma > 1/2$ this becomes
\begin{align*}
i_n^{\rm pot}(q_U(q_V^+),q_V^+,\lambda,\alpha,1,\rho)  &\approx  \sqrt{\alpha\gamma \rho  |\ln \rho|}
-\alpha(1-\gamma)\rho\ln \rho\,.
\end{align*}
We recall that here $\gamma= \Theta(1)$ so in the regime $\rho\to 0_+$ the right hand side remains positive.
If instead $\gamma \le 1/2$ then
\begin{align*}
i_n^{\rm pot}(q_U(q_V^+),q_V^+,\lambda,\alpha,1,\rho)
&\approx \sqrt{\alpha\gamma \rho  |\ln \rho|}-\alpha \gamma\rho \ln \rho\,.
\end{align*}
Comparing these two last expressions with \eqref{iqv-}, we see that equating the potential at its two minima in order to locate the phase transition is possible only when $\gamma > 1/2$ (because $\gamma >0$ and independent of $\rho$). This gives $\gamma=1$ and therefore identifies the transition at $\lambda_c=\sqrt{4 |\ln \rho|/(\alpha\rho)}$.

From this analysis we also obtain that the mutual information divided by $n$ is $\Theta(\sqrt{\rho |\ln \rho|})$ which justifies the normalization $(n \sqrt{\rho|\ln \rho|})^{-1}I((\bU,\bV);\bW)$ for it to have a non-trivial limit as $\rho\to 0_+$.
\section{The Nishimori identity}\label{app:nishimori}

\begin{lemma}[Nishimori identity]\label{NishId}
Let $(\bX,\bY)$ be a couple of random variables with joint distribution $P(\bX, \bY)$ and conditional distribution 
$P(\bX| \bY)$. Let $k \geq 1$ and let $\bx^{(1)}, \dots, \bx^{(k)}$ be i.i.d.\ samples from the conditional distribution. We use the bracket $\langle - \rangle$ for the expectation w.r.t. the product measure $P(\bx^{(1)}| \bY)P(\bx^{(2)}| \bY)\ldots P(\bx^{(k)}| \bY)$ and $\mathbb{E}$ for the expectation w.r.t. the joint distribution. Then, for all continuous bounded function $g$ we have 
\begin{align*}
\mathbb{E} \big\langle g(\bY,\bx^{(1)}, \dots, \bx^{(k)}) \big\rangle
=
\mathbb{E} \big\langle g(\bY, \bX, \bx^{(2)}, \dots, \bx^{(k)}) \big\rangle\,. 
\end{align*}	
\end{lemma}
\begin{proof}
This is a simple consequence of Bayes formula.
It is equivalent to sample the couple $(\bX,\bY)$ according to its joint distribution or to sample first $\bY$ according to its marginal distribution and then to sample $\bX$ conditionally on $\bY$ from the conditional distribution. Thus the two $(k+1)$-tuples $(\bY,\bx^{(1)}, \dots,\bx^{(k)})$ and  $(\bY, \bX, \bx^{(2)},\dots,\bx^{(k)})$ have the same law.	
\end{proof}

\section{Information theoretic properties of gaussian channels}\label{app:gaussianchannels}
In this appendix we prove important information theoretic properties of gaussian channels. These are mostly known \cite{GuoShamaiVerdu_IMMSE,guo2011estimation}, but we adapt them to our setting and provide detailed proofs  for the convenience of the reader.

Let us start with a key relation between the mutual information and the MMSE for gaussian channels. Equation \eqref{I-MMSE-relation} below is called the I-MMSE formula.
\begin{lemma}[I-MMSE formula]\label{app:I-MMSE}
 Consider a signal $\bX\in \mathbb{R}^n$ with $\bX\sim P_X$ that has finite support, and gaussian corrupted data $\bY\sim {\cal N}(\sqrt{R}\,\bX,{\rm I}_n)$ and possibly additional generic data $\bW\sim P_{W|X}(\cdot\,|\bX)$ with $H(\bW)$ bounded. The \emph{I-MMSE formula} linking the mutual information and the MMSE then reads
\begin{align}\label{I-MMSE-relation}
	\frac{d}{dR}I\big(\bX;(\bY,\bW)\big)=\frac{d}{dR}I(\bX;\bY|\bW) = \frac12 {\rm MMSE}(\bX|\bY,\bW)= \frac12 \EE\| \bX -\langle \bx\rangle \|^2\,,
\end{align}
where the Gibbs-bracket $\langle - \rangle$ is the expectation acting on $\bx\sim P(\cdot\,|\bY,\bW)$. 
% \begin{align*}
% P(\bx|\bY(\bX,\bZ),\bW(\bX))&=\frac{P_X(\bx)P_{W|X}(\bW|\bx)e^{-\frac12\|\bY-\sqrt{R}\bx\|^2}}{\int dP_X(\bx')P_{W|X}(\bW|\bx')e^{-\frac12\|\bY-\sqrt{R}\bx'\|^2}}\nn
% &=\frac{P_X(\bx)P_{W|X}(\bW|\bx)e^{-\frac12\|\bZ-\sqrt{R}(\bx-\bX)\|^2}}{\int dP_X(\bx')P_{W|X}(\bW|\bx')e^{-\frac12\|\bZ-\sqrt{R}(\bx'-\bX)\|^2}}\,.	
% \end{align*}
\end{lemma}
\begin{proof}
First note that by the chain rule for mutual information $I(\bX;(\bY,\bW))=I(\bX;\bY|\bW)+I(\bX;\bW)$, so the derivatives in \eqref{I-MMSE-relation} are equal.
We will now look at $\frac{d}{dR}I(\bX;(\bY,\bW))$. Since, conditionally on $\bX$, $\bY$ and $\bW$ are independent,  we have 
\begin{align*}
I\big(\bX;(\bY,\bW)\big) = H(\bY,\bW)-H(\bY,\bW|\bX) = H(\bY,\bW) - H(\bY|\bX) - H(\bW|\bX)\, .
\end{align*}
With gaussian noise contribution $H(\bY|\bX)=\frac n2\ln(2\pi e)$. Therefore only $H(\bY,\bW)$ depends on $R$. 
Let us then compute, using the change of variable $\bY=\sqrt{R}\,\bX+\bZ$,	
\begin{align}
	&\frac{d}{dR}I\big(\bX;(\bY,\bW)\big)=\frac{d}{dR}H(\bY,\bW)&\nn
	&\ =-\frac{d}{dR} \int dP_X(\bX)d\bY d\bW P_{W|X}(\bW|\bX) \frac{e^{-\frac12\|\bY-\sqrt{R}\bX\|^2}}{(2\pi)^{n/2}}\!\ln \!\int dP_X(\bx)P_{W|X}(\bW|\bx)\frac{e^{-\frac12\|\bY-\sqrt{R}\bx\|^2}}{(2\pi)^{n/2}}\nonumber\\
	&\ =- \int dP_X(\bX)d\bZ d\bW P_{W|X}(\bW|\bX)\frac{e^{-\frac12\|\bZ\|^2}}{(2\pi)^{n/2}}\frac{d}{dR}\!\ln \!\int dP_X(\bx)P_{W|X}(\bW|\bx)\frac{e^{-\frac12\|\bZ-\sqrt{R}(\bx-\bX)\|^2}}{(2\pi)^{n/2}}\nonumber\\
	&\ =\frac1{2\sqrt{R}}\EE_{\bX,\bZ,\bW|\bX}\big\langle (\bZ+\sqrt{R}(\bX-\bx))\cdot (\bX-\bx)\big \rangle
	\label{IPPtoapply}
\end{align}
where $\bZ\sim{\cal N}(0,{\rm I}_n)$ and the bracket notation is the expectation w.r.t. the posterior proportional to 
$$
dP_X(\bx)dP_{W|X}(\bW|\bx)d\bZ \exp\Big\{-\frac12\|\bZ-\sqrt{R}(\bx-\bX)\|^2\Big\}\, .
$$
In \eqref{IPPtoapply} the interchange of derivative and integrals is permitted by a standard application of Lebesgue's dominated convergence theorem in the case where the support of $P_X$ is bounded. 
%
%The interchange of $R$-derivative and integrals follows here from standard analysis arguments.
%In \eqref{IPPtoapply} the $R$-derivative and integrals have been safely interchanged. Indeed for any $R\ge R_0>0$ the function $g_R(\bX,\bW,\bZ)\equiv\ln \int  dP_X(\bx)\ldots$ is integrable w.r.t. the joint probability measure of $(\bX,\bW,\bZ)$. Moreover $2|\frac{d}{dR}g_R(\bX,\bW,\bZ)|=|\langle \bZ\cdot(\bx-\bX)/\sqrt{R}+ \|\bX-\bx\|^2\rangle|\le \|\bZ\|\langle\|\bx-\bX\rangle\|/\sqrt{R_0}+\langle\|\bX-\bx\|^2\rangle$ which is integrable because $P_X$ has finite support. Therefore we can use Lebesgue's dominated convergence theorem to exchange expectation and derivation. 
Now we use the following gaussian integration by part formula: for any bounded function ${\bm  g} :\mathbb{R}^n\mapsto \mathbb{R}^n$ of a standard gaussian random vector $\bZ\sim{\cal N}(0,{\rm I}_n)$ we obviously have
\begin{align}
\EE [\bZ\cdot {\bm  g}(\bZ)]=\EE [ \nabla_{\bZ} \cdot {\bm  g}(\bZ)]\,.
\label{SteinLemma}
\end{align}
This formula applied to a Gibbs-bracket associated to a general Gibbs distribution with hamiltonian ${\cal H}(\bx, \bZ)$ (depending on the Gaussian noise and possibly other variables) yields
\begin{align}
	\EE [\bZ\cdot \langle {\bm  h}(\bx)\rangle]&=	\EE\, \nabla_\bZ \cdot \frac{\int dP(\bx)e^{-{\cal H}(\bx,\bZ)} {\bm  h}(\bx)}{\int dP(\bx') e^{-{\cal H}(\bx',\bZ)}}\nn
	&= -\EE\,\frac{\int dP_X(\bx) e^{-{\cal H}(\bx,\bZ)} {\bm  h}(\bx)\cdot \nabla_\bZ {\cal H}(\bx,\bZ)}{\int dP_X(\bx') e^{-{\cal H}(\bx',\bZ)}}\nn
	&\hspace{2cm}+ \EE\Big[\frac{\int dP_X(\bx) e^{-{\cal H}(\bx,\bZ)} {\bm  h}(\bx)}{\int dP_X(\bx') e^{-{\cal H}(\bx',\bZ)}} \cdot \frac{\int dP_X(\bx) e^{-{\cal H}(\bx,\bZ)}\nabla_\bZ {\cal H}(\bx,\bZ) }{\int dP_X(\bx') e^{-{\cal H}(\bx',\bZ)}}\Big]\nn
	&=-\EE \big\langle {\bm  h}(\bx)\cdot \nabla_\bZ  {\cal H}(\bx,\bZ)\big\rangle+\EE \big[\big\langle {\bm  h}(\bx)\big\rangle \cdot\big\langle\nabla_\bZ  {\cal H}(\bx,\bZ)\big\rangle\big]\,.\label{GaussIPP}
	\end{align}
Applied to \eqref{IPPtoapply}, where the ``hamiltonian'' is $\mathcal{H}(\bx, \bZ)=-\ln P_{W|X}(\bW|\bx)+\frac12\|\bZ-\sqrt{R}(\bx-\bX)\|^2$, this identity gives
\begin{align*}
\frac{d}{dR}I\big(\bX;(\bY,\bW)\big)
&=\frac12\EE\big[\big\langle  \|\bX-\bx\|^2\big \rangle + \frac{1}{\sqrt{R}} \nabla_\bZ\cdot \langle \bX-\bx\rangle\big]\nonumber\\
&=\frac12\EE\big[\big\langle  \|\bX-\bx\|^2\big \rangle - \frac{1}{\sqrt{R}}\big\langle (\bX-\bx)\cdot(\bZ+\sqrt{R}(\bX-\bx))\big\rangle\nonumber \\
&\qquad\qquad+ \frac{1}{\sqrt{R}}\big\langle (\bX-\bx)\big\rangle \cdot\big\langle \bZ+\sqrt{R}(\bX-\bx)\big\rangle\big]\nonumber\\
&=\frac12\EE \|\bX-\langle \bx\rangle \|^2\,.
\end{align*}
\end{proof}

The MMSE cannot increase when the SNR increases. This translates into the concavity of 
the mutual information of gaussian channels as a function of the SNR. 

\begin{lemma}[Concavity of the mutual information in the SNR] \label{lemma:Iconcave}Consider the same setting as Lemma~\ref{app:I-MMSE}. Then the mutual informations $I(\bX;(\bY,\bW))$ and $I(\bX;\bY|\bW)$ are concave in the SNR of the gaussian channel:
\begin{align*}
	\frac{d^2}{dR^2}I\big(\bX;(\bY,\bW)\big)&=\frac{d^2}{dR^2}I(\bX;\bY|\bW) \nn
	&= \frac12 \frac{d}{dR}{\rm MMSE}(\bX|\bY,\bW)= -\frac{1}{2n}\sum_{i,j=1}^n\mathbb{E}\big[(\langle x_ix_j\rangle-\langle x_i\rangle\langle x_j\rangle)^2\big]\le 0
\end{align*}
where the Gibbs-bracket $\langle - \rangle$ is the expectation acting on $\bx\sim P(\cdot\,|\bY,\bW)$.
\end{lemma}
\begin{proof}
Set $Q\equiv \bx\cdot \bX/n$ where $\bx\sim P(\cdot\,|\bY,\bW)$. From a Nishimori identity 
${\rm MMSE}(\bX|\bY,\bW) = \EE_{P_X}[X^2]-\EE\langle Q\rangle$. Thus by the I-MMSE formula we have, by a calculation similar to \eqref{GaussIPP},
\begin{align}
-2\frac{d^2}{dR^2}I\big(\bX;(\bY,\bW)\big)=\frac{d\,\EE\langle Q\rangle}{dR}=n\EE [\langle Q\rangle\langle {\cal L}\rangle-\langle Q {\cal L}\rangle]\label{tocombine}
\end{align}
where we have set $${\cal L}\equiv \frac{1}{n}\Big(\frac{1}{2}\|\bx\|^2 - \bx\cdot \bX - \frac{1}{2\sqrt{R}}\bx\cdot \bZ\Big)\,.$$ Now we look at each term on the right hand side of this equality.
%$\bY=\sqrt{R}\,\bX+\bZ$ with $\bZ\sim{\cal N}(0, {\rm I}_n)$ and . $R$ appears both in the numerator and partition function (normalization) that define the bracket $\langle -\rangle$, thus the difference appearing. This is the same mechanism as when we derived identity . 
The calculation of appendix \ref{proof:remarkable_id} shows that
$$
-\EE\langle Q{\cal L}\rangle= \EE\langle Q^2\rangle-\frac12\EE[\langle Q\rangle^2]\, 
$$
so it remains to compute
\begin{align*}
\EE[\langle Q\rangle \langle {\cal L}\rangle]=\mathbb{E}\Big[\langle Q \rangle\frac{\big\langle \|\bx\|^2 \big\rangle}{2n} -  \langle Q \rangle^2 - \langle Q \rangle\frac{\bZ \cdot \langle \bx\rangle}{2n\sqrt{R}}  \Big]	\,.
\end{align*}
By formulas \eqref{SteinLemma} and \eqref{GaussIPP} in which the Hamiltonian is \eqref{Ht} we have
\begin{align*}
-\frac{1}{2n\sqrt{R}}\EE\big[ \bZ \cdot \langle \bx\rangle\langle Q \rangle \big]	&= -\frac{1}{2n\sqrt{R}}\EE\big[\langle Q\rangle \nabla_{\bZ}\cdot \langle \bx\rangle + \langle \bx\rangle\cdot \nabla \langle Q\rangle\big]\nn
&=-\frac{1}{2n}\EE\big[\langle Q\rangle \big(\big\langle \|\bx\|^2\big\rangle-\|\langle\bx\rangle\|^2\big) + \langle \bx\rangle\cdot \big(\langle Q\bx\rangle-\langle Q\rangle \langle \bx\rangle\big)\big]\nn
&\overset{\rm N}{=}-\frac{1}{2n}\EE\big[\langle Q\rangle \big\langle \|\bx\|^2\big\rangle\big]+\frac{1}{n}\EE\big[\langle Q\rangle \|\langle \bx \rangle \|^2\big]-\frac12\EE[\langle Q\rangle^2]\,.
\end{align*}
In the last equality we used the following consequence of the Nishimori identity. Let $\bx,\bx^{(2)}$ be two replicas, i.e., conditionally (on the data) independent samples from the posterior \eqref{tpost}. Then 
$$
\frac1n\EE\big[\langle \bx\rangle\cdot \langle Q \bx\rangle\big]=\frac1{n^2}\EE\big\langle (\bx^{(2)}\cdot \bx) (\bx\cdot \bX)\big\rangle\overset{\rm N}{=}\frac1{n^2}\EE\big\langle (\bx^{(2)}\cdot \bX) (\bX\cdot \bx)\big\rangle=\EE[\langle Q\rangle^2]\,.
$$
Thus we obtain 
\begin{align*}
\EE[\langle Q\rangle \langle {\cal L}\rangle -\langle Q {\cal L}\rangle] &=	\EE\langle Q^2\rangle-2\EE[\langle Q\rangle^2]+\frac{1}{n}\EE\big[\langle Q\rangle \|\langle \bx \rangle \|^2\big]\nn
&=\frac{1}{n^2}\EE\big\langle(\bx\cdot\bX)^2-2(\bx\cdot\bX)(\bx^{(1)}\cdot\bX)+(\bx\cdot\bX)(\bx^{(2)}\cdot\bx^{(3)})\big\rangle\nn
&\overset{\rm N}{=}\frac{1}{n^2}\EE\big\langle(\bx\cdot\bx^{(0)})^2-2(\bx\cdot\bx^{(0)})(\bx^{(1)}\cdot\bx^{(0)})+(\bx\cdot\bx^{(0)})(\bx^{(2)}\cdot\bx^{(3)})\big\rangle
\end{align*}
where $\bx^{(0)},\bx,\bx^{(2)}, \bx^{(3)}$ are replicas and the last equality again follows from a Nishimori identity. Multiplying this identity by $n$ and rewriting the inner products component-wise we get
\begin{align}
\frac{d\,\EE\langle Q\rangle}{dR} &=\frac{1}{n}\sum_{i,j=1}^n\EE\big\langle x_ix_i^{(0)} x_jx_j^{(0)} -2x_ix_i^{(0)}x_j^{(1)}x_j^{(0)}+x_ix_i^{(0)}x_j^{(2)}x_j^{(3)}\big\rangle\nn
&=\frac{1}{n}\sum_{i,j=1}^n\EE\big[\langle x_ix_j\rangle^2 -2\langle x_i\rangle \langle x_j\rangle \langle x_ix_j\rangle+ \langle x_i\rangle^2\langle x_j\rangle^2\big] \label{secondDer_f_pos}
\end{align}
Using \eqref{tocombine} this ends the proof of the lemma. Note that we have also shown the positivity claimed in \eqref{jacPos} of section \ref{sec:adapInterp_XX}.
\end{proof}

\begin{lemma}[Concavity of the average MMSE in $P_X$] \label{lemma:MMSEconcave}Consider the same setting as Lemma~\ref{app:I-MMSE}. The functionnal ${\rm MMSE}(\bX|\bY,\bW)$ is concave in $P_X$.
\end{lemma}
\begin{proof}
Let $B\sim {\rm Ber}(\alpha)$ be a Bernoulli variable. Consider any random variables $\bX_0\sim P_{X_0}$, $\bX_1\sim P_{X_1}$ independent of $B$. Let $\bX_B \sim P_X=(1-\alpha)P_{X_0} + \alpha P_{X_1}$. Consider the problem
of estimating $\bX_B$ given $\bY_B=\sqrt{\lambda}\,\bX_B +\bZ$ with $\bZ\sim{\cal N}(0,{\rm I}_n)$ (and possibly other data $\bW_B\sim P_{W|X}(\cdot|\bX_B)$). If $B$ is given one can then choose the MMSE estimator based on $P_{X_0}$ if $B=0$, or $P_{X_1}$ else. Therefore knowing $B$ can only lower the MMSE in average. In equations,
\begin{align*}
{\rm MMSE}(\bX_B|\bY_B,\bW) &=\EE_B\EE\|\bX_B-\langle \bx\rangle_{B,P_X}\|^2\nn
&\qquad\qquad=(1-\alpha)\EE\|\bX_0-\langle \bx\rangle_{0,P_X}\|^2+\alpha\EE\|\bX_1-\langle \bx\rangle_{1,P_X}\|^2\,.
\end{align*}
Here the bracket notation $\langle \bx\rangle_{b,P_X}$, $b\in\{0,1\}$, means the expectation of $\bx$ distributed according to the probability distribution proportional to $dP_X(\bx)P_{W|X}(\bW_b|\bx)\exp\{-\frac12\|\bY_b -\sqrt{\lambda}\bx\|^2\}$. By definition of the MMSE $${\rm MMSE}(\bX_b|\bY_b,\bW,B=b)=\EE\|\bX_b-\langle \bx\rangle_{b,P_{X_b}}\|^2\le \EE\|\bX_b-\langle \bx\rangle_{b,P_X}\|^2\,.$$ 
Therefore we have
\begin{align*}
{\rm MMSE}(\bX_B|\bY_B,\bW_B)&\ge (1-\alpha){\rm MMSE}(\bX_0|\bY_0,\bW_0,B=0)+\alpha{\rm MMSE}(\bX_1|\bY_1,\bW_1,B=1)\nn
&\equiv{\rm MMSE}(\bX_B|\bY_B,\bW_B,B)
\end{align*}
which proves the desired concavity.
\end{proof}

As a fundamental measure of uncertainty, the MMSE decreases with additional side information available to the estimator. This is because that an informed optimal estimator performs no worse (in average) than any uninformed estimator by simply discarding the side information.

\begin{lemma}[Conditionning reduces the MMSE]\label{lemma:conditioningMMSE}
	Consider the same setting as Lemma~\ref{app:I-MMSE}. For any $\bW'$ jointly distributed with $\bX$ we have
	\begin{align*}
		{\rm MMSE}(\bX|\bY,\bW)\ge {\rm MMSE}(\bX|\bY,\bW,\bW')\,.
	\end{align*}
\end{lemma}
\begin{proof}
	This follows directly from Lemma~\ref{lemma:MMSEconcave} using $P_{X}(\bX)=\int dP_{W'}(\bw')P_{X|W'}(\bX|\bw')$: 
	\begin{align*}
	{\rm MMSE}(\bX|\bY,\bW)\ge \int dP_{W'}(\bw'){\rm MMSE}(\bX|\bY,\bW,\bW'=\bw')={\rm MMSE}(\bX|\bY,\bW,\bW')\,.
	\end{align*}
\end{proof}
\begin{lemma}[Stability of mutual information for gaussian channels]\label{lemma:MIadditivity}
	Consider a random variable $\mathbb{R}^n\times \mathbb{R}^m\ni (\bU,\bV)\sim P_{UV}$ with conditionally (on $(\bU,\bV)$) independent data $\bY_{R_1}\sim {\cal N}(\sqrt{R_1}\bU,{\rm I}_n)$, $\,\bY_{R_2}\sim {\cal N}(\sqrt{R_2}\bU,{\rm I}_n)$, and $\,\bW\sim P_{W|UV}(\cdot\,|\bU,\bV)$. If $\,\bY_{R_1+R_2}\sim {\cal N}(\sqrt{R_1+R_2}\bU,{\rm I}_n)$ independently of the rest then
	\begin{align*}
		I\big((\bU,\bV);(\bY_{R_1},\bY_{R_2},\bW)\big)&=I\big((\bU,\bV);(\bY_{R_1+R_2},\bW)\big)\,.%\label{eqMI}
		% {\rm MMSE}(\bU|\bY_{R_1},\bY_{R_2},\bW)&={\rm MMSE}(\bU|\bY_{R_1+R_2},\bW)\,.\label{eqMMSE}
	\end{align*}
\end{lemma}
\begin{proof}
	The proof is a simple consequence of the stability of the normal law under addition. 
	By conditional independence of the data on $(\bU,\bV)$ we have 
	\begin{align}
	I\big((\bU,\bV);(\bY_{R_1},\bY_{R_2},\bW)\big)&=H(\bY_{R_1},\bY_{R_2},\bW)-H(\bY_{R_1}|\bU)-H(\bY_{R_2}|\bU)-H(\bW|\bU,\bV)	\label{toUse_expMI1}
	\end{align}
	where $H(\bY_{R_1}|\bU)+H(\bY_{R_2}|\bU)=n\ln(2\pi e)$ because the noise is i.i.d. gaussian. Then 
	\begin{align*}
	&H(\bY_{R_1},\bY_{R_2},\bW)-H(\bY_{R_1}|\bU)-H(\bY_{R_2}|\bU)\nn
	&\quad=-\int dP_{UV}(\bU,\bV)d\bY_{R_1}d\bY_{R_2}dP_{W|UV}(\bW|\bU,\bV)\frac{1}{(2\pi)^n}e^{-\frac12\|\bY_{R_1}-\sqrt{R_1}\bU\|^2-\frac12\|\bY_{R_2}-\sqrt{R_2}\bU\|^2}\nn
	&\quad\quad\times \ln\int dP_{UV}(\bu,\bv)P_{W|UV}(\bW|\bu,\bv)\frac{1}{(2\pi)^n}e^{-\frac12\|\bY_{R_1}-\sqrt{R_1}\bu\|^2-\frac12\|\bY_{R_2}-\sqrt{R_2}\bu\|^2}-n\ln(2\pi e)\nn
	&\quad=-n\ln(2\pi e)-\int dP_{UV}(\bU,\bV)d\bZ_1d\bZ_2dP_{W|UV}(\bW|\bU,\bV)\frac{1}{(2\pi)^n}e^{-\frac12\|\bZ_1\|^2-\frac12\|\bZ_2\|^2}\nn
	&\quad\qquad\times \ln\int dP_{UV}(\bu,\bv)P_{W|UV}(\bW|\bu,\bv)\frac{1}{(2\pi)^n}e^{-\frac12\|\bZ_1-\sqrt{R_1}(\bu-\bU)\|^2-\frac12\|\bZ_2-\sqrt{R_2}(\bu-\bU)\|^2}\nn
	&\quad=-\EE\ln\int dP_{UV}(\bu,\bv)P_{W|UV}(\bW|\bu,\bv)e^{-\frac12(R_1+R_2)\|\bU-\bu\|^2+(\sqrt{R_1}\bZ_1+\sqrt{R_2}\bZ_2)\cdot(\bu-\bU)}
	\end{align*}
	where $\EE=\EE_{(\bU,\bV),\bZ_1,\bZ_2,\bW|(\bU,\bV)}$ with $\bZ_1$ and $\bZ_2$ being i.i.d. ${\cal N}(0,{\rm I}_n)$ random variables. Because in law $\sqrt{R_1}\bZ_1+\sqrt{R_2}\bZ_2=\sqrt{R_1+R_2}\bZ$ with $\bZ\sim {\cal N}(0,{\rm I}_n)$ we have
	\begin{align*}
	&H(\bY_{R_1},\bY_{R_2},\bW)-H(\bY_{R_1}|\bU)-H(\bY_{R_2}|\bU)\nn
	&\qquad=-\EE\ln\int dP_{UV}(\bu,\bv)P_{W|UV}(\bW|\bu,\bv)e^{-\frac12(R_1+R_2)\|\bU-\bu\|^2+\sqrt{R_1+R_2}\bZ\cdot(\bu-\bU)}\,.\label{alsoequalto}
	\end{align*}
	Similarly we obtain that $H(\bY_{R_1+R_2},\bW)-H(\bY_{R_1+R_2}|\bU)$, with $H(\bY_{R_1+R_2}|\bU)=\frac n2 \ln(2\pi e)$, also equals the right hand side of the above equality. Then $I((\bU,\bV);(\bY_{R_1+R_2},\bW))=H(\bY_{R_1+R_2},\bW)-H(\bY_{R_1+R_2}|\bU)-H(\bW|\bU,\bV)$ combined with \eqref{toUse_expMI1} implies the result.
	% By the chain rule for mutual information 
	% \begin{align*}
	% I\big((\bU,\bV);(\bY_{R_1},\bY_{R_2},\bW)\big)&=I\big((\bU,\bV);(\bY_{R_1+R_2},\bW)\big)	
	% \end{align*}
	% The equality \eqref{eqMMSE} then follows directly by taking the derivative w.r.t. $R_1$ or $R_2$ on both sides of \eqref{eqMI} and then using the I-MMSE formula Lemma~\ref{app:I-MMSE}.
\end{proof}
\small{
\bibliographystyle{unsrt_abbvr}
\newpage    
\bibliography{refs}  } 
\end{document}